\documentclass[a4paper,UKenglish,cleveref, autoref, thm-restate]{lipics-v2021}
\usepackage{amsmath}
\usepackage{amssymb}
\usepackage{leftindex}

\usepackage{stmaryrd}
\usepackage{textcomp}

\usepackage[thinc]{esdiff} % \diffp{f}{x}

\usepackage{cancel}

\usepackage{bbm}
\usepackage{proof} %proofrules

\usepackage{array}

\usepackage{color}

\usepackage{subcaption} 

\usepackage[table]{xcolor}
\usepackage{todonotes}
\usepackage{menukeys}
\presetkeys{todonotes}{}{}
\setuptodonotes{inline}

\usepackage{etoolbox}

\usepackage{wasysym} %smileys

\usepackage{xcolor}
\usepackage{nicefrac}
\usepackage{tabularx}
\newcommand{\ifempty}[3]{\ifthenelse{\equal{#1}{}}{#2}{#3}}

\usepackage{tikz}% Tikz
\usetikzlibrary{
	arrows, % arrows heads, e.g. for inhibitor arcs
	automata, % automata
	positioning, % positioning of nodes
	calc, % Calculating coordinates
	shapes.geometric,
}

\tikzstyle{labelnode}=[font=\footnotesize]
\tikzstyle{dist}=[circle,  inner sep=0.8pt, solid, draw,fill]
\tikzstyle{ma}=[every label/.style={labelnode}]
\tikzstyle{mdp}=[ma]
\tikzset{ms/.style={draw,rectangle,text centered,minimum size=6mm,text width=#1},ms/.default=3mm}
\tikzset{ps/.style={draw,circle,text centered,minimum size=6mm,text width=#1},ps/.default=3mm}
\tikzset{pswide/.style={draw,ellipse,minimum size=6mm,text centered,text width=#1},pswide/.default=7mm}
\tikzstyle{mtrans}=[->,densely dotted, semithick,labelnode]
\tikzstyle{ptrans}=[->,semithick,labelnode]
\tikzstyle{every initial by arrow}=[inner sep=0pt] % make sure that initial arrow doesn't have white space in front of it
\tikzset{init/.style={initial #1, initial text={}, initial distance=4mm},init/.default=left}
\newcommand*{\toptionalfrac}[2][]{\ifempty{#1}{#2}{\nicefrac{#1}{#2}}}

\newcommand*{\tprob}[2][]{{\ensuremath{\color{colorProbability}\toptionalfrac[#1]{#2}}}}
\newcommand*{\tact}[2][]{{\ensuremath{\color{colorAction}\toptionalfrac[#1]{#2}}}}
\newcommand*{\tactext}[2][]{{\ensuremath{\color{colorextendedAction}\toptionalfrac[#1]{#2}}}}

\newcommand{\modelType}{pa}
\newcommand{\extended}{true} % false
\newcommand{\prodAutomata}{true} % false

\newcommand{\mless}{\mathit{mless}}
\newcommand{\comp}{\mathit{cmp}}
\newcommand{\safe}{\mathit{safe}}
\newcommand{\prt}{\mathit{prt}}
\newcommand{\fairWrtRegionModel}[2]{{{\mathit{fair}}_{#2\ }^{#1}}\!\!} %{{f\!air}_{#1}^{ #2}}
\newcommand{\fairC}{\fairWrtRegionModel{}{\decomp}} %{{f\!air}_{#1}^{ #2}}

\colorlet{colorRate}{orange}
\colorlet{colorProbability}{teal}
\colorlet{colorAction}{magenta}
\colorlet{colorReward}{green}
\colorlet{colorextendedAction}{blue}

\newcommand{\domain}{\mathit{dom}}

\newcommand{\rats}{\mathbb{Q}}
\newcommand{\nats}{\mathbb{N}}
\newcommand{\reals}{\mathbb{R}}

\newcommand{\cupdot}{\mathbin{\mathaccent\cdot\cup}}

\newcommand{\iverson}[1]{\llbracket\,#1\,\rrbracket}

\newcommand{\set}[1]{\{#1\}}
\newcommand{\bigset}[1]{\big\{#1\big\}}

\newcommand{\subDist}[1]{\mathit{SubDist}(#1)}
\newcommand{\dist}[1]{\mathit{Dist}(#1)}
\newcommand{\paramDist}[2]{\mathit{Dist}_{#1}(#2)}% #2 \to (\rats(#1) \cup \reals \cup \{\bot \})}

\newcommand{\modelsWrt}[1]{
    \def\temp{#1}
    \ifx\temp\empty%
      \models{}\!
    \else
      \models^{#1}
    \fi
  }

\newcommand{\last}[1]{\mathit{last}(#1)}
\newcommand{\pref}[1]{\mathit{pref}(#1)}

\newcommand{\indicatorFct}[1]{\mathbbm{1}_{#1}}

\newcommand{\setOfdecompsOf}[1]{2^{#1}}

\newcommand{\enabledActsOf}[1]{\mathit{Act}\left(#1 \right)}

\newcommand{\decomp}[0]{\mathcal{C}}

\newcommand{\lab}{\mathsf{a}}
\newcommand{\altlab}{\mathsf{b}}
\newcommand{\altaltlab}{\mathsf{c}}

\newcommand{\actSetOf}[1]{
    \def\temp{#1}
    \ifx\temp\empty%
      \mathit{Act}
    \else
      \mathit{Act}_{#1}
    \fi
  }

  \newcommand{\syncOf}[1]{
    \def\temp{#1}
    \ifx\temp\empty%
      L
    \else
      L_{#1}
    \fi
  }

\newcommand{\alphabetOf}[1]{
    \def\temp{#1}
    \ifx\temp\empty%
      \Sigma
    \else
      \Sigma_{#1}
    \fi
  }

\newcommand{\stateSetOf}[1]{
    \def\temp{#1}
    \ifx\temp\empty%
      S 
    \else
      S_{#1}
    \fi
  }
\newcommand{\transRelationOf}[1]{
  \def\temp{#1}
  \ifx\temp\empty%
    \delta 
  \else
    \delta_{#1}
  \fi
}
\newcommand{\transFctOf}[1]{
	\def\temp{#1}
	\ifx\temp\empty%
	\mathbf{P} 
	\else
  \mathbf{P}_{#1}
	\fi
}

\newcommand{\initialOf}[1]{
  \def\temp{#1}
  \ifx\temp\empty%
    {s^{\mathit{init}}}
  \else
    {s^{\mathit{init}}_{#1}}
  \fi
}
\newcommand{\parameterSetOf}[1]{
    \def\temp{#1}
    \ifx\temp\empty%
      V 
    \else
      V_{#1}
    \fi
  }

\newcommand{\idSetOf}[1]{
    \def\temp{#1}
    \ifx\temp\empty%
      ID
    \else
      ID_{#1}
    \fi
  }

\newcommand{\paTupleOf}[1]{\left(\stateSetOf{#1}, \initialOf{#1}, \actSetOf{#1}, \transFctOf{#1},   \syncOf{#1}  \right)} %
\newcommand{\ppaTupleOf}[1]{\left(\stateSetOf{#1}, \initialOf{#1},\parameterSetOf{#1},  \actSetOf{#1} , \transFctOf{#1}, \syncOf{#1} \right)} % \

\newcommand{\bpAutomatonOf}[1]{\mathcal{A}^{bad}_{#1}}
\newcommand{\QstateSetOf}[1]{
    \def\temp{#1}
    \ifx\temp\empty%
      Q
    \else
      Q_{#1}
    \fi
  }
\newcommand{\qInitialOf}[1]{
      \def\temp{#1}
      \ifx\temp\empty%
        {q^{init}}
      \else
        {q^{init}_{#1}}
      \fi
}

\newcommand{\regionIntersection}[0]{\cap} %curlywedge
\newcommand{\regionUnion}[0]{\cup} %curlyvee

\newcommand{\regionIntersectionOf}[3]{
	\def\temp{#2}
	\ifx\temp\empty%
    {{#1} \! \regionIntersection{} \! {#3}}
	\else
    {{#1 \! \regionIntersection{} #2 \! \regionIntersection{} #3} }
	\fi
  }

\newcommand{\stratProjOfToValuation}[3]{
	\def\temp{#3}
	\ifx\temp\empty%
	{\restrOfTo{#1}{#2}}
	\else
	{{\restrOfToVal{#1}{#2}{#3}}}
	\fi}

\newcommand{\regionUnionOf}[3]{
	\def\temp{#2}
	\ifx\temp\empty%
	{{#1} \!\regionUnion\! {#3}}
	\else
	{{#1}\! \regionUnion \!{#2} \regionUnion\! {#3}}
	\fi}

\newcommand{\APOf}[1]{
	\def\temp{#1}
	\ifx\temp\empty%
	{AP}
	\else
	{AP_{#1}}
	\fi}
	
\newcommand{\labellingOf}[2]{
    \def\temp{#2}
      \ifx\temp\empty%
      {StateLab_{#1}}
      \else
       {StateLab_{#1}}\left(#2\right)
      \fi}

\newcommand{\PrOf}[3]{Pr_{#1}^{#2}
    \def\temp{#3}
      \ifx\temp\empty%
      \else
        \!\left( #3 \right)
      \fi}

\newcommand{\bdownarrow}{\big \downarrow}  %todo big
\newcommand{\buparrow}{\big \uparrow}  %todo big
\newcommand{\budarrow}{\big \updownarrow}  %todo big

\newcommand{\monotonicOnRegionParameter}[5]{
  \def\temp{#3}
  {#2}{\!#1}_{
    {\ifx\temp\empty%
    {#4}
    \else
    {#3}, {#4}
    \fi}
  }^{#5}
}

\newcommand{\ExpTot}[3]{Ex_{#1}^{#2}
	\def\temp{#3}
	\ifx\temp\empty%
	\else
	\!\left( #3 \right)
	\fi}

\newcommand{\probPredicate}[2]{\mathbb{P}_{#1}\!(#2)}  %Not _{#1} last, ugly with = afterwards 
\newcommand{\expPredicate}[2]{\mathbb{E}_{#1}\!(#2)}%\bigl(#2\bigr)}
\newcommand{\generalPredicate}[0]{\varphi} %NOT NEEDED?

\newcommand{\infpath}{\pi}
\newcommand{\ppath}{\pi}
\newcommand{\finpath}{\widehat{\pi}}
\newcommand{\infPathsOf}[1]{Paths_{#1}^{inf}}%{I\!Paths_{#1}^{#2}}
\newcommand{\finPathsOf}[1]{Paths_{#1}^{fin}}%{F\!Paths_{#1}^{#2}}
\newcommand{\liftedpaths}[2]{ #1 \otimes #2 }

\newcommand{\traceOf}[1]{\mathit{tr}({#1})}
\newcommand{\trace}{\rho}

\newcommand{\alphabetExtensionOfTo}[2]{#1{\langle#2\rangle}}

\newcommand{\restrOfToVal}[3]{{#1{\upharpoonright}_{\!#2}^{#3}}}
\newcommand{\restrOfTo}[2]{{#1{\upharpoonright}_{\!#2}}} %for traces

\newcommand{\strategy}{\sigma} % or \strategy and valuation v
\newcommand{\strategyset}{Str}
\newcommand{\strategysetOf}[2]{\strategyset_{\!#1}^{\!#2}}
\newcommand{\agTriple}[4]{ {#1} \modelsWrt{#2} {#3}  {\rightarrow}  {#4}}  %A,M,star,G   { {#1} \Vert^{{#2}} \overset{ {#3}}{\sqsubset} {#4}}     %

\newcommand{\multiobjectiveQuery}[2]{\textsf{#2}^{#1}}

\newcommand{\solutionFctMdpObjective}[2]{\textsf{sol}_{#1}^{#2}}

\newcommand{\cyl}[0]{\textsf{Cyl}}

\newcommand{\ppa}{\mathcal{M}}
\newcommand{\pa}{\mathcal{N}}
\newcommand{\product}{{\otimes}}

\newcommand{\regLang}{\mathcal{L}}
\newcommand{\rewFct}{\mathcal{R}}

\newcommand{\valuation}{\valuationVector{v}}
\newcommand{\valuationVector}[1]{\mathpzc{#1}}

\DeclareFontFamily{OT1}{pzc}{}
\DeclareFontShape{OT1}{pzc}{m}{it}{<-> s * [1.10] pzcmi7t}{}
\DeclareMathAlphabet{\mathpzc}{OT1}{pzc}{m}{it}
\newcommand{\region}{{R}}   %{\mathpzc{reg}}

\definecolor{myorange}{RGB}{230,159,0}
\definecolor{myblue}{RGB}{55,126,184}
\definecolor{mygreen}{RGB}{77,175,74}
\definecolor{myred}{RGB}{228,26,28}
\definecolor{mypink}{RGB}{204,121,167}
\definecolor{mygray}{rgb}{0.66, 0.66, 0.66}
\newtoggle{extended} % by default, we compile the camera-ready version
\toggletrue{extended} % uncomment to compile the extended arxiv version

\nolinenumbers

\iftoggle{extended}{% extended version
    \pdfoutput=1 %uncomment to ensure pdflatex processing (mandatatory e.g. to submit to arXiv)
    \hideLIPIcs  %uncomment to remove references to LIPIcs series (logo, DOI, ...), e.g. when preparing a pre-final version to be uploaded to arXiv or another public repository
    \allowdisplaybreaks
}
{
}

\title{Compositional Reasoning for Parametric Probabilistic Automata} 
\author{Hannah Mertens\footnote{Corresponding author}}{RWTH Aachen University, Aachen, Germany}{hannah.mertens@cs.rwth-aachen.de}{https://orcid.org/0009-0009-6815-3285}{}
\author{Tim Quatmann}{RWTH Aachen University, Aachen, Germany}{tim.quatmann@cs.rwth-aachen.de}{https://orcid.org/0000-0002-2843-5511}{This research was funded by a KI-Starter grant from the Ministerium für Kultur
und Wissenschaft NRW}

\author{Joost-Pieter Katoen}{RWTH Aachen University, Aachen, Germany}{katoen@cs.rwth-aachen.de}{https://orcid.org/0000-0002-6143-1926}{}

\authorrunning{H. Mertens, T.Quatmann and J.\,P. Katoen} %TODO mandatory. First: Use abbreviated first/middle names. Second (only in severe cases): Use first author plus 'et al.'

\Copyright{Hannah Mertens, Tim Quatmann and Joost-Pieter Katoen} %TODO mandatory, please use full first names. LIPIcs license is "CC-BY";  http://creativecommons.org/licenses/by/3.0/

	
 % mandatory: Please choose ACM 2012 classifications from https://dl.acm.org/ccs/ccs_flat.cfm 

 % mandatory; please add comma-separated list of keywords

\category{} %optional, e.g. invited paper

\EventEditors{Patricia Bouyer and Jaco van de Pol}
\EventNoEds{2}
\EventLongTitle{36th International Conference on Concurrency Theory (CONCUR 2025)}
\EventShortTitle{CONCUR 2025}
\EventAcronym{CONCUR}
\EventYear{2025}
\EventDate{August 26--29, 2025}
\EventLocation{Aarhus, Denmark}
\EventLogo{}
\SeriesVolume{348}
\ArticleNo{23}
\begin{document}

\nottoggle{extended}{
    \todo{add link and add bibtex extendedMer+25}
    \relatedversiondetails[cite={extendedMer+25}]{Full Version}{https://arxiv.org/abs/...}
}{
    \relatedversion{This is the full version of the paper accepted for publication at CONCUR 2025}
}

\maketitle

\begin{abstract}
We establish an assume-guarantee (AG) framework for compositional reasoning about multi-objective queries in parametric probabilistic automata (pPA)\textemdash{}an extension to probabilistic automata (PA), where transition probabilities are functions over a finite set of parameters. 
We lift an existing framework for PA to the pPA setting, incorporating asymmetric, circular, and interleaving proof rules. 
Our approach enables the verification of a broad spectrum of multi-objective queries for pPA, encompassing probabilistic properties and (parametric) expected total rewards. 
Additionally, we introduce a rule for reasoning about monotonicity in composed pPAs.
\end{abstract}

\section{Introduction}\label{sec:intro}

% Success of Probmc
\emph{Probabilistic Model Checking}~\cite{For+11,Kat16} studies the automated verification of Markov models for systems with random behavior.
Applications include network and security protocols, biochemical processes, and planning under uncertainty~\cite{NS06,KNP08,FWHT15}.
Common properties such as reachability probabilities in Markov decision processes (MDPs) can often be verified efficiently in PTIME~\cite{BK08}.

When the probabilities with which the system evolves are not known exactly, verification results must be robust towards slight perturbations.
\emph{Parametric Markov models}~\cite{Daw04,Jun+24} allow representing uncertain model quantities---for example the bias of a coin-flip or the probability of a sensor misreading---using parameters.
\emph{Feasibility} is a fundamental verification problem for parametric systems and asks whether there is an instantiation of the parameters under which a given specification holds.
Deciding feasibility for reachability probability specifications in parametric MDPs is ETR-complete, i.e., at least NP-hard and within PSPACE~\cite{JKPW21}.
The dual \emph{verification} problem that asks if the specification holds under all instantiations is co-ETR complete.
Checking parametric Markov models is therefore significantly more complex compared to Markov models without parameters.

The number of system states grows exponentially with the number of system components.
The resulting \emph{state-space explosion} is an omnipresent challenge when model checking complex systems, often rendering analysis computationally infeasible. 
Compositional verification techniques such as \emph{assume-guarantee} (AG) reasoning \cite{Jon83,Pnu84} address this problem by decomposing the verification task into smaller sub-tasks that consider individual components in isolation. 
This modular verification approach has been successfully applied in various domains, including service-based workflow verification \cite{Bou+16}, large-scale IT systems \cite{Cal+12}, and autonomous systems incorporating deep neural networks \cite{Pas+18,Pas+23}. 
Recent advancements include circular AG reasoning \cite{Elk+18} and verification-repair techniques \cite{Fre+22}. 
Extensions to probabilistic systems have further expanded the scope of AG reasoning, as demonstrated in works such as \cite{Kwi+13}, where AG reasoning was applied to Segala's probabilistic automata (PA) \cite{Seg+95}---a compositional extension of Markov decision processes.
Automated approaches to AG reasoning for probabilistic systems have also been explored \cite{Fen+11, LL19}, enabling more scalable verification. 
Other works have considered parametric, but non-probabilistic timed automata~\cite{Lac+16} as well as parameterized programs \cite{BR+06, SH13, NT16}---where the concurrent system is parameterized by the number of processes or threads in a configured instance.

\emph{This work introduces a framework for compositional reasoning about parametric probabilistic automata (pPA).} 
The case studies presented by Kwiatkowska et al.~\cite{Kwi+13} demonstrate the practical applicability of AG reasoning within a non-parametric setting. These findings provide strong motivation for extending this approach to the parametric domain. In this work, we develop the theoretical foundations of an AG reasoning framework for pPAs, leveraging results from (non-parametric) PAs~\cite{Kwi+13}. 
Due to the aforementioned ETR-hardness, compositional reasoning has a large potential and can be crucial to verify complex parameterized systems that are too large to handle monolithically.

\begin{example}\label{ex:high-level}
Consider a communication system, where a \emph{sender} $\mathcal{S}$ broadcasts messages to a \emph{receiver} $\mathcal{R}$ through a \emph{broadcast channel} $\mathcal{B}$. 
The system is modeled by the parallel composition $\mathcal{S} \parallel \mathcal{B} \parallel \mathcal{R}$.
The components are faulty: $\mathcal{S}$ might face a collision, broadcasting in $\mathcal{B}$ might fail due to message loss, and $\mathcal{R}$ might miss broadcasts.
% For simplicity, we abstract away from further details of the individual components
The reliability of $\mathcal{S}$, $\mathcal{B}$, and $\mathcal{R}$ is influenced by parameters and the precise values of these parameters vary depending on network conditions, interference, or other factors. 
Our goal is to verify that under all parameter instantiations in a given parameter space $R$, the message is successfully \emph{received} with at least probability $0.7$, formally denoted by 
\[
(\mathcal{S} \parallel \mathcal{B} \parallel \mathcal{R}), R ~\modelsWrt{}~ \probPredicate{\geq 0.7}{\mathit{received}}.
\]
AG reasoning allows to verify the specification without explicitly considering the (potentially large) composition $\mathcal{S} \parallel \mathcal{B} \parallel \mathcal{R}$.
To this end, assume that we have established the following statements:
\begin{itemize}
    \item
    $\mathcal{S}, R \modelsWrt{} \probPredicate{<0.1}{\mathit{collision}}$---the probability that $\mathcal{S}$ faces a collision is below $0.1$
    \item
    $\agTriple{\mathcal{B}, R}{}{\probPredicate{<0.1}{\mathit{collision}}}{\probPredicate{\ge0.8}{\mathit{broadcast}}}$---if $\mathcal{B}$ observes a collision with probability below $0.1$, the message is broadcast with probability at least $0.8$
    \item
    $\agTriple{\mathcal{R}, R}{}{\probPredicate{\ge0.8}{\mathit{broadcast}}}{\probPredicate{\ge0.7}{\mathit{received}}}$---If the message is broadcast with probability at least $0.8$, then $\mathcal{R}$ receives the message with probability at least $0.7$

\end{itemize}
We reason about the composed system using the AG rule stated in \Cref{theo:pag_asym_rule} in \Cref{sec:pag}:
\[
\infer{(\mathcal{S} \parallel \mathcal{B}), R \modelsWrt{} \probPredicate{\ge0.8}{\mathit{broadcast}}}
  {\deduce{\agTriple{\mathcal{B}, R}{}{\probPredicate{<0.1}{\mathit{collision}}}{\probPredicate{\ge0.8}{\mathit{broadcast}}}}
   {\mathcal{S}, R \modelsWrt{} \probPredicate{<0.1}{\mathit{collision}}}
  }
\qquad
\infer{(\mathcal{S} \parallel \mathcal{B} \parallel \mathcal{R}), R \modelsWrt{} \probPredicate{\ge0.7}{\mathit{received}}}
  {\deduce{\agTriple{\mathcal{R}, R}{}{\probPredicate{\ge0.8}{\mathit{broadcast}}}{\probPredicate{\ge0.7}{\mathit{received}}}}
   {(\mathcal{S} \parallel \mathcal{B}), R \modelsWrt{} \probPredicate{\ge0.8}{\mathit{broadcast}}}
  }
\]
\end{example}

\noindent\textbf{Contributions.}
To the best of our knowledge, \emph{we provide the first framework for compositional reasoning of parametric Markov models}.
Our main contributions are as follows. 
\begin{itemize}
    \item We introduce and formalize pPAs, i.e., compositional probabilistic automata with parametric transitions.
    \item We provide a conservative extension of \emph{strategy projections}~\cite{Seg95,Kwi+13} to pPAs, including a more natural definition based on conditional probabilities.
    Strategy projections are essential for correctness of compositional reasoning as they allow to link measures of a composed model to measures of its constituting components.
    \item We present rules for assume-guarantee reasoning, generalizing an established framework by Kwiatkowska et al.~\cite{Kwi+13} to the parametric setting---which requires some technically intricate proofs.
    The framework applies to $\omega$-regular and expected total reward properties as well as multi-objective combinations thereof.
\item 
We provide a new rule for compositional reasoning about monotonicity in pPAs.
Knowing that a measure of interest---either the probability to satisfy an $\omega$-regular specification or an expected total reward---is monotone in one or more parameters can significantly speed up verification~\cite{Jun+24,Spe+19}.
Our rule allows to derive monotonicity w.r.t.\ a composed pPA by only determining monotonicity for its components. 
\end{itemize}

We introduce pPAs in \Cref{sec:preliminaries_math} and discuss strategy projections in \Cref{sec:strat_projections}.
\Cref{sec:verification} outlines properties of interest and \Cref{sec:pag} presents our AG rules.
We outline results for monotonicity in \Cref{sec:pag_mono} and related work in \Cref{sec:related_work}.
\Cref{sec:conclusion} concludes the paper.
\iftoggle{extended}{% extended version
    Proofs omitted in the main part of the paper are given in \Cref{app:proofs}.
}{% Conefernce version: 
    Proofs omitted in the main part of the paper are given in the extended version~\cite{extendedMer+25}. 
}

\section{Preliminaries}
\label{sec:preliminaries_math}
For sets $X$ and $Y$, let $f \colon X \hookrightarrow Y$ denote a \emph{partial function} from $X$ to $Y$ with domain $\domain(f) \subseteq X$.
The projection of $f$ to a set $Z$ is written as $\restrOfTo{f}{Z} \colon (X \cap Z) \to Y$.
Iverson brackets $\iverson{\varphi} \in \{0,1\}$ map a Boolean condition $\varphi$ to $1$ if $\varphi$ holds and $0$ otherwise.

$\rats[\parameterSetOf{}]$ denotes the set of (multivariate) \emph{polynomials} over a finite set of real-valued \emph{parameters} $\parameterSetOf{}  = \{p_1, \dots, p_n\}$.  
A (parameter) \emph{valuation} for $\parameterSetOf{}$ is a function $\valuation \colon \parameterSetOf{} \to \reals$. 
Evaluating a polynomial $f \in \rats[\parameterSetOf{{}}]$ at $\valuation$ yields $f[\valuation] \in \reals$. 
A \emph{region} $\region$ for $\parameterSetOf{}$ is a set of valuations.
For $p \in \parameterSetOf{}$, we define the valuation $\valuationVector{e}_p$ with $\valuationVector{e}_p(q) = \iverson{p{=}q}$ for $q \in \parameterSetOf{}$.

A \emph{parametric distribution}\footnote{We use the term parametric \emph{distribution}\textemdash{}rather than parametric \emph{function}\textemdash{}to emphasize that we are typically interested in functions $\mu$ where $\mu[\valuation]$ is a (sub)probability distribution.} for $\parameterSetOf{}$ over a finite set $S$ is a function $\mu \colon \stateSetOf{} \to (\rats[\parameterSetOf{}] \cup \reals)$.
Applying valuation $\valuation$ to $\mu$ yields $\mu[\valuation] \colon S \to \reals$ with $\mu[\valuation](s) = \mu(s)[\valuation]$ for all $s \in S$. 
We call $\mu \colon S \to [0,1]$ a \emph{subdistribution} if $\sum_{s \in \stateSetOf{}} \mu(s) \leq 1$ and a \emph{distribution} if $\sum_{s \in \stateSetOf{}} \mu(s) =1$.
The sets of parametric distributions, subdistributions, and distributions over $S$ are denoted by $\paramDist{\parameterSetOf{}}{S}$, $\subDist{\stateSetOf{}}$, and $\dist{\stateSetOf{}}$, respectively. 
For $s \in S$, $\indicatorFct{s} \in \dist{\stateSetOf{}}$ is the \emph{Dirac} distribution with
$\indicatorFct{s}(s') = \iverson{s'{=}s}$.
For sets $S_1,S_2$, the \emph{product} of $\mu_1 \in \paramDist{\parameterSetOf{}}{S_1}$ and $\mu_2 \in \paramDist{\parameterSetOf{}}{S_2}$ is the distribution $\mu_1 {\times} \mu_2 \in \paramDist{\parameterSetOf{}}{S_1\times S_2}$ with $(\mu_1 {\times} \mu_2)(s_1,s_2) = \mu_1(s_1) \cdot \mu_2(s_2)$.

\subsection{Parametric Probabilistic Automata}
\label{sec:preliminaries_ppa}
We combine probabilistic automata (PA)~\cite{Seg+95,Sto+02} with parametric Markov models~\cite{Jun+24}.
\begin{definition}%[Parametric Probabilistic Automaton]
A \emph{parametric probabilistic automaton (pPA)} over a finite alphabet $\alphabetOf{}$ is a tuple $\ppa = \ppaTupleOf{}$, where 
%	\begin{itemize}
$\stateSetOf{}$, $\parameterSetOf{}$, and $\actSetOf{}$ are finite sets of states, parameters, and actions, respectively, $\initialOf{} \in \stateSetOf{}$ is an initial state,
		 $\transFctOf{} \colon  (\stateSetOf{} \times \actSetOf{}) \hookrightarrow \paramDist{\parameterSetOf{}}{\stateSetOf{}}$ is a parametric transition function, and  
	 $\syncOf{} \colon \domain(\transFctOf{}) \to  \alphabetOf{}$ is a labeling function.
\end{definition}
Let $\ppa = \ppaTupleOf{}$ be a pPA. For $s \in S$, $\enabledActsOf{s} = \set{\alpha \in \actSetOf{}  \mid (s,\alpha) \in \domain(\transFctOf{})}$ denotes the set of enabled actions in $s$ and $\ppa_s$ is the pPA where the initial state is set to~$s$.
We set $\transFctOf{}(s, \alpha, s') = \transFctOf{}(s, \alpha)(s')$ if $(s,\alpha) \in \domain(\transFctOf{})$ and otherwise $\transFctOf{}(s, \alpha, s') = 0$.
$\ppa$ is a \emph{(non-parametric)} PA if $\transFctOf{}(s, \alpha) \in \dist{\stateSetOf{}}$ for all $(s,\alpha) \in \domain(\transFctOf{})$.
In this case, $\transFctOf{}(s,\alpha,s')$ is the probability to transition to successor state $s'$ when action $\alpha$ is selected at state $s$.

The \emph{instantiation} of $\ppa$ at valuation $\valuation$ for $\parameterSetOf{}$ is the pPA 
	$\ppa[\valuation] = \left(\stateSetOf{}, \initialOf{},\emptyset,  \actSetOf{} , \transFctOf{}[\valuation], \syncOf{} \right)$, 
    where $\domain(\transFctOf{}[\valuation]) = \domain(\transFctOf{})$ and $\transFctOf{}[\valuation](s,\alpha) = \transFctOf{}(s,\alpha)[\valuation]$.
If $\ppa[\valuation]$ is a non-parametric PA, we say $\valuation$ is \emph{well-defined} for $\ppa$.
A valuation $\valuation$ is \emph{graph-preserving} for $\ppa$ if it is well-defined and for all $s,  s' \in \stateSetOf{}$ and $\alpha \in \actSetOf{\ppa}$: $\transFctOf{}(s,\alpha,s')[\valuation] = 0$ iff $\transFctOf{}(s,\alpha,s') = 0$.
A region $\region$ is well-defined (graph-preserving) if all its valuations $\valuation \in \region$ are well-defined (graph-preserving).

\begin{figure}[t] 
	\begin{subfigure}{.35\textwidth}
		\renewcommand{\modelType}{\ppa}
		\centering 	
		\begin{tikzpicture}[mdp]

	\node[ps, init=left] (0)  {$s_0$};

	\node[ps,below=1 of 0] (1)  {$s_1$};

	% Probabilistic transitions
	\path[ptrans]
	
	(0) edge[loop right] node[pos=0.45,below right] {\tact{\altlab},\tact{\altaltlab}} node[dist] (d1a) {} (0)
	
	(0) edge[bend right=0] node[pos=0.5,right] {\tact{\lab}} node[dist] (d0a) {} node[pos=0.75,left] {\tprob{{\ifthenelse{\equal{\modelType}{pa}}{\frac{9}{10}}{1-p}}}} (1)
	
	(d0a) edge[bend left=60] node[left,pos=0.45] {\tprob{{\ifthenelse{\equal{\modelType}{pa}}{\frac{1}{10}}{p}}}} node[above left=3pt,pos=0.4] {} (0)
	
	(1) edge[loop right] node[pos=0.5,above right] {\tact{\altlab}} node[dist] (d1c) {} node[pos=0.75,below] {\tprob{}} node[pos=0.5,below] {} (1)

	;
\end{tikzpicture}%
		\caption{pPA $\ppa_1$}\label{fig:pPAM1}%
	\end{subfigure}\hfill
	\begin{subfigure}{.55\textwidth} 
		\renewcommand{\extended}{false}
		\renewcommand{\modelType}{\ppa}
		\centering%
		\begin{tikzpicture}[mdp]
	\node[ps, init=left] (t0)  {$t_0$};
	
	\node[ps,above right=0.25 and 2 of t0] (t1)  {$t_1$};

	\node[ps,below right=0.25 and 2 of t0] (t2)  {$t_2$};
	
	\node[ps,right=2 of t1, fill=lightgray] (t3)  {$t_3$};
	
	\node[ps,right=2 of t2] (t4)  {$t_4$};
	
	\node[dist, right=0.5 of t0] (d0a) {};
		
	% Probabilistic transitions
	\path[ptrans]

		(t0) edge[-] node[pos=0.95, below] {\tact{\lab}} (d0a)
		(d0a) edge[bend right=0] node[pos=0.5,above,yshift=0.5ex,xshift=-0.5ex] {\tprob{\ifthenelse{\equal{\modelType}{pa}}{\frac{9}{10}}{1-p}}} (t1)
		(d0a) edge[bend right=0] node[below,pos=0.55,yshift=-0.5ex,xshift=-0.5ex] {\tprob{\ifthenelse{\equal{\modelType}{pa}}{\frac{1}{10}}{p}}} node[above,pos=0.4] {} (t2)

		(t1) edge[bend right=0] node[pos=0.25,above] {\tact{\lab}} node[dist, pos=0.25] (d1a) {} node[pos=0.6,above] {\tprob{\ifthenelse{\equal{\modelType}{pa}}{\frac{9}{10}}{q}}} (t3)
		(d1a) edge[bend left=25] node[below,pos=0.25,xshift=-1ex] {\tprob{\ifthenelse{\equal{\modelType}{pa}}{\frac{1}{10}}{1-q}}} node[above,pos=0.4] {} (t4)
		
		(t2) edge[bend right=0] node[pos=0.25,below] {\tact{\altaltlab}} node[dist, pos=0.25] (d1b) {} node[pos=0.6,below] {\tprob{\ifthenelse{\equal{\modelType}{pa}}{\frac{9}{10}}{\frac{9}{10}}}} (t4)
		(d1b) edge[bend right=25] node[above, pos=0.25] {\tprob{\ifthenelse{\equal{\modelType}{pa}}{\frac{1}{10}}{\frac{1}{10}}}} node[above,pos=0.4] {} (t3)

		(t3) edge[loop right] node[pos=0.45,below right] {{\ifthenelse{\equal{\extended}{true}}{\tactext{\altlab},}{}}\tact{\small \frownie}} node[dist] (d3frown) {} node[pos=0.15,below] {\tprob{}} node[pos=0.25,below] {} (t3)
		
		(t4) edge[loop right] node[pos=0.5,above right] {\ifthenelse{\equal{\extended}{true}}{\tactext{\altlab},}{}{\tact{\altaltlab}}} node[dist] (d4b) {} node[pos=0.15,below] {\tprob{}} node[pos=0.25,below] {} (t4)

	;	
	
	{\ifthenelse{\equal{\extended}{true}}{\path[ptrans] (t2) edge[loop above] node[pos=0.75,right] {{\ifthenelse{\equal{\extended}{true}}{\tactext{\altlab}}{}}} node[dist] (d1c) {} (t2);}{}}
	
	{\ifthenelse{\equal{\extended}{true}}{\path[ptrans] (t1) edge[loop above] node[pos=0.75,right] {{\ifthenelse{\equal{\extended}{true}}{\tactext{\altlab}}{}}} node[dist] (d1c) {} (t2);}{}}
	
	{\ifthenelse{\equal{\extended}{true}}{\path[ptrans] (t0) edge[loop above] node[pos=0.75,right] {{\ifthenelse{\equal{\extended}{true}}{\tactext{\altlab}}{}}} node[dist] (d1c) {} (t0);}{}}

\end{tikzpicture}%
		\caption{pPA $\ppa_2$}\label{fig:pPAM2}%
	\end{subfigure} 
	\caption{Example pPAs $\ppa_1$ and $\ppa_2$.}\label{fig:ppa1_and_ppa2}
\end{figure}
\begin{example}[pPA]\label{ex:pPAM1M2}
	Consider the pPA $\ppa_1$ in \Cref{fig:pPAM1} and $\ppa_2$ in \Cref{fig:pPAM2}. 
	$\ppa_1 = \ppaTupleOf{1}$, where $\stateSetOf{1} = \{s_0,s_1\}$, $\initialOf{1} = s_0$ and $\parameterSetOf{1}=\{p\}$. 
	Actions $\actSetOf{1} = \{\lab,\altlab,\altaltlab\}$ and alphabet $\alphabetOf{1}=\{\lab,\altlab,\altaltlab\}$. 
	In this example, the alphabet coincides with actions as these uniquely define the transitions. 
	Similarly, $\ppa_2 = \ppaTupleOf{2}$,  where $\stateSetOf{2} = \{t_0,\dots, t_4\}$, $\initialOf{2} = t_0$, and $\parameterSetOf{2}=\{p,q\}$. 
	Again, $\actSetOf{2}= \alphabetOf{2}=\{\lab,\altaltlab,\frownie\}$. 
\end{example}

An \emph{infinite path} of $\ppa$ is an alternating sequence $\infpath = s_0, \alpha_0, s_1, \alpha_1, \dots$ of states $s_i \in \stateSetOf{}$ and actions $\alpha_i \in \actSetOf{}$ such that %$s_0 = \initialOf{}$ and 
$(s_i, \alpha_i) \in \domain(\transFctOf{})$ for all $i \ge 0$.
A finite path of length $n \in \nats$ is a prefix $\finpath = s_0, \alpha_0, \dots, s_n$ of an infinite path, ending in a state $\last{\finpath} = s_n \in \stateSetOf{}$.
$\infPathsOf{\ppa}{}$ and $\finPathsOf{\ppa}{}{}$ are the sets of infinite and finite paths of $\ppa$, respectively.
For a (finite or infinite) path $\infpath \in \infPathsOf{\ppa}{} \cup \finPathsOf{\ppa}{}$, we write $\vert \infpath \vert \in \nats \cup \set{\infty}$ for its length and $\pi[0,j]$ for its prefix of length $j \le \vert \infpath \vert$.
We deliberately allow paths that take transitions with probability 0. 
As a consequence, a path of a pPA M is always also a path of any of its instantiations $\ppa[\valuation]$\textemdash{}even if $\valuation$ is not graph-preserving.

Strategies\textemdash{}also known as schedulers or adversaries\textemdash{}resolve nondeterminism by assigning (sub-)distributions over enabled actions based on the history---i.e., a finite path---observed so far.
We allow for partial strategies that, intuitively, can choose none of the enabled actions to reflect the case that no further transition is executed.
\begin{definition}%[Strategy]%[Strategies (Partial and Complete)]\label{def:strategy}
	A \emph{(partial) strategy} for $\ppa$ is a function $\strategy \colon \finPathsOf{\ppa}{} \to \subDist{\actSetOf{}}$ such that 
	$\strategy(\finpath)(\alpha) > 0 $ implies $(\last{\finpath},\alpha) \in \domain(\transFctOf{})$.
	A strategy $\strategy \colon \finPathsOf{\ppa} \to \dist{\actSetOf{\ppa}}$ is called \emph{complete}. 
	The set of all partial and complete strategies on $\ppa$ are denoted by $\strategysetOf{\ppa}{\star}$, where $\star \in \{\prt, \comp\}$, respectively. 
	A \emph{memoryless strategy} only depends on the last state of $\finpath$.
	The set of memoryless strategies on $\ppa$ is denoted by $\strategysetOf{\ppa}{\mless,\star}$.
\end{definition}
For a strategy $\strategy$ for $\ppa$, we may write $\strategy(\finpath, \alpha)$ instead of $\strategy(\finpath)(\alpha)$.
If $\strategy$ is memoryless, we write $\strategy(s_n, \alpha)$ instead of $\strategy(\finpath, \alpha)$, where $s_n = \last{\finpath}$.

A well-defined instantiation $\valuation$ and a strategy $\strategy$ for $\ppa$ yield a purely probabilistic process described by the \emph{(sub)probability measure} 
$\PrOf{\ppa}{\valuation,\strategy}{}$ on the measurable subsets of $\infPathsOf{\ppa}{}$, which is obtained by a standard cylinder set construction~\cite{BK08}: 
\[\cyl(\finpath) = \{ \pi \in \infPathsOf{\ppa}{} \mid \finpath \text{ is a prefix of } \pi \}\] is the \emph{cylinder set} of a \emph{finite} path $\finpath = s_0, \alpha_0,  \dots, s_n$ of $\ppa$ and we set
	\[
	\PrOf{\ppa}{\valuation,\strategy}{\cyl (\finpath)}
	~=~ 
	\iverson{s_0 = \initialOf{}} 
	\cdot \prod_{ i= 0}^{n-1} \strategy(\finpath[0,i], \alpha_i) \cdot \transFctOf{}(s_{i}, \alpha_{i}, s_{i+1})[\valuation].
\] 
This definition extends uniquely to a probability measure on \emph{all} measurable sets of infinite paths. 
We further lift $\PrOf{\ppa}{\valuation,\strategy}{}$ to (sets of) finite paths and write, e.g., $\PrOf{\ppa}{\valuation,\strategy}{\finpath}$ for $\finpath \in \finPathsOf{\ppa}$ or $\PrOf{\ppa}{\valuation,\strategy}{\Pi}$ for $\Pi \subseteq \finPathsOf{\ppa}$---implicitly referring to (unions of) cylinder sets.
If $\ppa$ is a (non-parametric) PA, we may omit $\valuation$ and write $\PrOf{\ppa}{\strategy}{}$.
Well-defined $\valuation$ yields $\PrOf{\ppa[\valuation]}{\strategy}{} = \PrOf{\ppa}{\valuation,\strategy}{}$.
\begin{example}
    For the pPA $\ppa_2$ from \Cref{fig:pPAM2} and a well-defined valuation $\valuation$, the probability to reach the state $t_3$ under valuation $\valuation$ is $(1-\valuation(p)) \cdot \valuation(q) + \valuation(p) \cdot \frac{1}{10}$. 
\end{example}
\begin{remark}\label{rem:padef}
Our definition of PA slightly deviates from related work \cite{Seg+95,Kwi+13,Kom+12,LL19}, which commonly define a transition \emph{relation} $\transRelationOf{} \subseteq \stateSetOf{} \times \alphabetOf{} \times \dist{\stateSetOf{}}$ instead of functions $\transFctOf{}$ and $\syncOf{}$.
In our setting, a pair $(s,\alpha) \in \domain(\transFctOf{})$ uniquely identifies both, a label $\syncOf{}(s,\alpha) \in \alphabetOf{}$, and a distribution over successor states $\transFctOf{}(s,\alpha) \in \dist{\stateSetOf{}}$, which significantly simplifies formalizations related to pPAs.
In particular, any strategy for $\ppa$ immediately also applies to instantiations of $\ppa$ and vice versa, i.e., we have $\strategysetOf{\ppa}{} = \strategysetOf{\ppa[\valuation]}{}$ for any valuation $v$. 
On the other hand, our variant does not affect expressiveness of non-parametric PA as one can convert between the two formalisms.
\end{remark}
We lift parallel composition of PA~\cite{Seg+95} to pPAs. Composed pPAs synchronize on common transition labels while behaving autonomously on non-common labels.
For simplicity, we assume that composed pPAs consider a common set of parameters $\parameterSetOf{}$.\footnote{If two pPAs have different parameter sets $\parameterSetOf{1} \neq \parameterSetOf{2}$, the assumption can be established by considering $\parameterSetOf{} = \parameterSetOf{1} \cup \parameterSetOf{2}$ instead, potentially adding (unused) parameters to the individual pPAs.}
\begin{definition}[Parallel Composition]
\label{def_ppa_composition}
	For $i=1,2$, let $\ppa_{i} = (\stateSetOf{i}, \initialOf{i}, \parameterSetOf{}, \actSetOf{i}, \transFctOf{i}, \syncOf{i})$  be two pPAs over alphabets $\alphabetOf{i}$ with $\actSetOf{i} \cap (\alphabetOf{1} \cup \alphabetOf{2}) = \emptyset$.
The \emph{parallel composition} of $\ppa_1$ and $\ppa_2$ is given by the pPA 
$\ppa_{1} \parallel \ppa_{2} = \big(\stateSetOf{1} \times \stateSetOf{2}, (\initialOf{1}, \initialOf{2}), \parameterSetOf{}, \actSetOf{\parallel}, \transFctOf{\parallel}, \syncOf{\parallel}\big)$
 over $\alphabetOf{1} \cup \alphabetOf{2}$, where 
\begin{itemize}
\item $\actSetOf{\parallel} = (\actSetOf{1} \times \actSetOf{2}) \cupdot (\actSetOf{1} \times \alphabetOf{1} \setminus \alphabetOf{2}) \cupdot (\alphabetOf{2} \setminus \alphabetOf{1} \times \actSetOf{2})$,
\item
for each $(s_1, \alpha_1) \in \domain(\transFctOf{1})$, $(s_2, \alpha_2) \in \domain(\transFctOf{2})$ with $\syncOf{1}(s_1, \alpha_1) = \syncOf{2}(s_2, \alpha_2) \in \alphabetOf{1} \cap \alphabetOf{2}$:
\[\transFctOf{\parallel}((s_1,s_2), (\alpha_1, \alpha_2)) = \transFctOf{1}(s_1,\alpha_1) \times \transFctOf{2}(s_2,\alpha_2) 
\quad\text{and}\quad \syncOf{\parallel}((s_1,s_2), (\alpha_1, \alpha_2)) = \syncOf{1}(s_1, \alpha_1),\]
\item for each $(s_1, \alpha_1) \in \domain(\transFctOf{1})$, $s_2 \in \stateSetOf{2}$ with $\syncOf{1}(s_1, \alpha_1) = \lab_1 \in \alphabetOf{1} \setminus \alphabetOf{2}$: 
\[ \transFctOf{\parallel}((s_1,s_2), (\alpha_1, \lab_1)) = \transFctOf{1}(s_1,\alpha_1) \times \indicatorFct{s_2}
\quad\text{and}\quad \syncOf{\parallel}((s_1,s_2), (\alpha_1, \lab_1)) = \lab_1 = \syncOf{1}(s_1, \alpha_1),\]
\item for each $s_1 \in \stateSetOf{1}$, $(s_2, \alpha_2) \in \domain(\transFctOf{2})$ with $\syncOf{2}(s_2, \alpha_2) = \lab_2 \in \alphabetOf{2} \setminus \alphabetOf{1}$: 
\[ \transFctOf{\parallel}((s_1,s_2), (\lab_2, \alpha_2)) = \indicatorFct{s_1} \times \transFctOf{2}(s_2,\alpha_2)
\quad\text{and}\quad \syncOf{\parallel}((s_1,s_2), (\lab_2, \alpha_2)) = \lab_2 = \syncOf{2}(s_2, \alpha_2).\]
\end{itemize}
\end{definition}
The parallel composition of parametric probabilistic automata (pPAs) is associative, meaning that $(\ppa_{1} \parallel \ppa_{2}) \parallel \ppa_3$ and $\ppa_{1} \parallel (\ppa_{2}\parallel \ppa_3)$ are equivalent up to state renaming. 
Therefore, we denote this composition as $\ppa_{1} \parallel \ppa_{2}\parallel \ppa_3$. 
\renewcommand{\modelType}{\ppa}
\renewcommand{\prodAutomata}{false}
\begin{figure}[!t]
    \centering
	\begin{tikzpicture}[mdp]
	\node[pswide, init=left] (00)  {$s_0, t_0\ifthenelse{\equal{\prodAutomata}{true}}{,p_0}{}$};	
	\node[pswide,right=6 of 00,color=white] (10)  {};
	
	\node[pswide,below right=1.5 and 0.6 of 00] (01)  {$s_0,t_1\ifthenelse{\equal{\prodAutomata}{true}}{,p_0}{}$};
	\node[pswide,below left=1.5 and 0.6 of 00] (02)  {$s_0,t_2\ifthenelse{\equal{\prodAutomata}{true}}{,p_0}{}$};
	
	\node[pswide,below =1.5 of 01, fill=lightgray] (03)  {$s_0,t_3\ifthenelse{\equal{\prodAutomata}{true}}{,p_0}{}$};
	\node[pswide,below =1.5 of 02] (04)  {$s_0,t_4\ifthenelse{\equal{\prodAutomata}{true}}{,p_0}{}$};

	\node[pswide,below right=1.5 and 0.6 of 10] (11)  {$s_1,t_1\ifthenelse{\equal{\prodAutomata}{true}}{,p_0}{}$};
	\node[pswide,below left=1.5 and 0.6 of 10] (12)  {$s_1,t_2\ifthenelse{\equal{\prodAutomata}{true}}{,p_0}{}$};
	
  	\node[pswide,below =1.5 of 11, fill=lightgray] (13)  {$s_1,t_3\ifthenelse{\equal{\prodAutomata}{true}}{,p_0}{}$};
	\node[pswide,below =1.5 of 12] (14)  {$s_1,t_4\ifthenelse{\equal{\prodAutomata}{true}}{,p_0}{}$};

	% Probabilistic transitions

	\path[ptrans]

	% loops
	(00) edge[loop above] node[pos=0.65, right]  {\tact{\altlab}} node[dist] (d00b) {} node[pos=0.5,left] {\tprob{}} node[pos=0.35,below] {} (00)
		
	(03) edge[loop left] node[pos=0.55, below left] {\tact{\altlab\ifthenelse{\equal{\prodAutomata}{true}}{}{,\small \frownie}}} node[dist] (d01c) {} node[pos=0.5,left] {\tprob{}} node[pos=0.35,below] {} (03)
	
	(02) edge[loop left] node[pos=0.55, below left]  {\tact{\altlab}} node[dist] (d01c) {} node[pos=0.5,left] {\tprob{}} node[pos=0.35,below] {} (03)
	
	(04) edge[loop left] node[pos=0.55, below left]  {\tact{\altlab},\tact{\altaltlab}} node[dist] (d01c) {} node[pos=0.5,left] {\tprob{}} node[pos=0.35,below] {} (04)
	
	(01) edge[loop left] node[pos=0.55, below left]  {\tact{\altlab}} node[dist] (d01c) {} node[pos=0.5,below] {\tprob{}} node[pos=0.35,below] {} (01)
	
	(11) edge[loop left] node[pos=0.55, below left]  {\tact{\altlab}} node[dist] (d01c) {} node[pos=0.5,left] {\tprob{}} node[pos=0.35,below] {} (11)
	
	(12) edge[loop left] node[pos=0.55, below left]  {\tact{\altlab}} node[dist] (d01c) {} node[pos=0.5,left] {\tprob{}} node[pos=0.35,below] {} (12)
	
	(13) edge[loop left] node[pos=0.55, below left]  {\tact{\altlab\ifthenelse{\equal{\prodAutomata}{true}}{}{,\small \frownie}}} node[dist] (d01c) {} node[pos=0.5,left] {\tprob{}} node[pos=0.35,below] {} (13)
	
	(14) edge[loop left] node[pos=0.55, below left]  {\tact{\altlab}} node[dist] (d01c) {} node[pos=0.5,left] {\tprob{}} node[pos=0.35,below] {} (14)

	% 00 a
	(00) edge[bend left=20] node[pos=0.15,above] {\tact{\lab}} node[dist, pos=0.25] (d00a) {} node[pos=0.8,above,yshift=0.8ex,xshift=0.9ex] {\tprob{\ifthenelse{\equal{\modelType}{pa}}{\frac {81}{100}}{(1-p)^2}}} (11)
	(d00a) edge[bend right=0] node[left, pos=0.75] {\tprob{\ifthenelse{\equal{\modelType}{pa}}{\frac{9}{100}}{p\cdot(1-p)}}} node[above,pos=0.75] {} (01)
	(d00a) edge[bend right=10] node[above, pos=0.75] {\tprob{\ifthenelse{\equal{\modelType}{pa}}{\frac{1}{100}}{p^2}}} node[above,pos=0.75] {} (02)
	(d00a) edge[bend left=10] node[right, pos=0.75,yshift=0.1ex,xshift=0.1ex] {\tprob{\ifthenelse{\equal{\modelType}{pa}}{\frac{9}{10}}{(1-p)\cdot p}}} node[above,pos=0.75] {} (12)

	% 01 a
	(01) edge[bend left=0] node[pos=0.2,right, yshift=0.2ex] {\tact{\lab}} node[dist, pos=0.3] (d01a) {} node[pos=0.65,left] {\tprob{\ifthenelse{\equal{\modelType}{pa}}{\frac {9}{100}}{p \cdot q}}} (03)
	(d01a) edge[bend left=15] node[above, pos=0.25] {\tprob{\ifthenelse{\equal{\modelType}{pa}}{\frac{81}{100}}{(1-p) \cdot q}}} node[above,pos=0.75] {} (13)
	(d01a) edge[bend right=10] node[above, pos=0.25,rotate=15] {\tprob{\ifthenelse{\equal{\modelType}{pa}}{\frac{1}{100}}{p\cdot(1-q)}}} node[below,pos=0.15] {} (04)
	(d01a) edge[bend left=10] node[below, pos=0.35,rotate=-15] {\tprob{\ifthenelse{\equal{\modelType}{pa}}{\frac{9}{100}}{(1-p)\cdot(1-q)}}} node[above,pos=0.75] {} (14)

	(02) edge[bend right=5] node[pos=0.25,left] {\tact{\altaltlab}} node[dist, pos=0.35] (d02c) {} node[pos=0.58,left] {\tprob{\ifthenelse{\equal{\modelType}{pa}}{\frac{9}{10}}{\frac{9}{10}}}} (04)
	(d02c) edge[bend left=15] node[above, pos=0.25] {\tprob{\ifthenelse{\equal{\modelType}{pa}}{\frac{1}{10}}{\frac{1}{10}}}} node[left,pos=0.75] {} (03)
	;
	% product with badG
	{\ifthenelse{\equal{\prodAutomata}{true}}{
			\node[pswide,below =2 of 13,  fill=lightgray, label=-180:{$\{bad_{\regLang_G}\}$}] (13bad)  {$s_1,t_3,p_1$};
			\node[pswide,below =2 of 03,  fill=lightgray, label=-180:{$\{bad_{\regLang_G}\}$}] (03bad)  {$s_0,t_3,p_1$};
			
			\path[ptrans] 
			(13) edge[bend right=0] node[pos=0.25,left] {\tact{\frownie}} node[dist, pos=0.25] (d13frown) {} (13bad)

			(03) edge[bend right=0] node[pos=0.25,left] {\tact{\frownie}} node[dist, pos=0.25] (d03frown) {} (03bad)
			(03bad) edge[loop right] node[pos=0.25, above] {\tact{\lab},\tact{\altlab},\tact{\altaltlab},\tact{\frownie}} node[dist] (d01c) {} node[pos=0.25,above] {\tprob{}} node[pos=0.35,below] {} (03bad)
			(13bad) edge[loop right] node[pos=0.25, above] {\tact{\lab},\tact{\altlab},\tact{\altaltlab},\tact{\frownie}} node[dist] (d01c) {} node[pos=0.25,above] {\tprob{}} node[pos=0.35,below] {} (03bad)
			;
		}{}}
\end{tikzpicture}\caption{Parallel composition of pPAs $\ppa_{1}$ and $\ppa_{2}$ from \Cref{fig:ppa1_and_ppa2}.}\label{fig:pPA_composition}
\end{figure}
\begin{example}[Parallel Composition]
	\Cref{fig:pPA_composition} shows the composition of pPAs $\ppa_{1}$ and $\ppa_{2}$ from \cref{fig:ppa1_and_ppa2}. Actions with labels $\lab$ or $\altaltlab$ are synchronized while $\altlab$ and $\frownie$ are asynchronous. 
\end{example}
Similar to \cite[Section 3.5]{Kwi+13}, we sometimes assume fairness of strategies, meaning that specific sets of labels $\alphabetOf{i} \subseteq \alphabetOf{}$ are visited infinitely often. 
\begin{definition}[Fair Strategy] 
    Let $\decomp \subseteq \setOfdecompsOf{\alphabetOf{}}$ and $\valuation$ be a well-defined valuation for $\ppa$ over $\alphabetOf{}$.  
    A complete strategy $\strategy \in \strategysetOf{\ppa[\valuation]}{\comp}$ is \emph{fair w.r.t. $\decomp$} (denoted $\fairC$) if 
    \[
        \PrOf{\ppa}{\valuation, \strategy}{\bigset{ s_0,\alpha_0,s_1,\alpha_1,\dots  \in \infPathsOf{\ppa}{} \mathbin{\big|} \forall \alphabetOf{i} \in \decomp \colon\, \forall j\in\nats\colon\, \exists k\ge j\colon\, \syncOf{}(s_k,\alpha_k) \in \alphabetOf{i} } }= 1.
    \]  
    The set of all $\fairWrtRegionModel{}{\decomp}$ strategies of $\ppa[\valuation]$ is denoted $\strategysetOf{\ppa[\valuation]}{\fairWrtRegionModel{}{\decomp}}$. 
\end{definition}
Almost-sure repeated reachability in PA only depends on the graph structure, which yields: 
\begin{proposition}\label{theo:graph_preserving_fairness}
For any graph-preserving valuations $\valuation, \valuation'$ for $\ppa$ we have $\strategysetOf{\ppa[\valuation]}{\fairC} = \strategysetOf{\ppa[\valuation']}{\fairC}$, i.e., a strategy is $\fairC$ for $\ppa[\valuation]$ iff it is $\fairC$  for $\ppa[\valuation']$.
\end{proposition}

\section{Strategy Projections}\label{sec:strat_projections}
In this section, we define the projection of a strategy of a composite pPA $\ppa = \ppa_1 \parallel \ppa_2$ onto a single component $\ppa_i$ for $i=1,2$.
Projections for PA are defined in \cite[Definition 6]{Kwi+13} and originate from \cite[page 65, Definition of Projection]{Seg95}.
They are intuitively used to relate probability measures for $\ppa$ and $\ppa_i$.

The projection of a finite path $\ppath \in \finPathsOf{\ppa}$ onto component $\ppa_i$ is the finite path $\restrOfTo{\ppath}{i}\in \finPathsOf{\ppa_i}$ obtained by restricting $\pi$ to the steps performed by $\ppa_i$.
Formally, $\restrOfTo{(\initialOf{1},\initialOf{2})}{i} = \initialOf{i}$ and for $\ppath = \ppath',(\alpha_1,\alpha_2),(s_1,s_2)$:
\[
\restrOfTo{\ppath}{i} = 
\begin{cases}
\restrOfTo{\ppath'}{i}, \alpha_i, s_i &\text{if } \alpha_i \in \actSetOf{i}\\
\restrOfTo{\ppath'}{i}                &\text{otherwise.} %&\text{if } \alpha_i = \lab_{3-i} \notin \actSetOf{i}.
\end{cases}
\]
Path projections are neither injective nor surjective, i.e., we might have $\restrOfTo{\ppath}{i	} = \restrOfTo{\ppath'}{i}$ for two distinct paths $\ppath \neq \ppath'$ of $\ppa$, and for some $\ppath_i \in \finPathsOf{\ppa_i}$ there might not be any $\ppath \in \finPathsOf{\ppa}$ with $\ppath_i= \restrOfTo{\ppath}{i}$.
We define the set of paths of $\ppa_1 \parallel \ppa_2$ that are projected to $\ppath_i \in \finPathsOf{\ppa_i}$ as 
\[
\liftedpaths{\ppath_i}{\ppa_{3-i}} = \bigset{\ppath \in \finPathsOf{\ppa} \mid \ppath_i= \restrOfTo{\ppath}{i}}. 
\]
We first focus on strategy projections for non-parametric PA.
Then, we lift our notions to the parametric setting. 
% \iftoggle{extended}{ 
%     % extended version
%     Omitted proofs can be found in \Cref{app:proofs_of_strat_projections}. 
% }
% { 
%     % Conefernce version: 
%     Omitted proofs can be found in the extended version~\cite{extendedMer+25}. 
% }

\subsection{Projections for non-parametric PAs}
We fix the parallel composition $\pa = \pa_1 \parallel \pa_2 = \paTupleOf{\parallel}$ of (non-parametric) PAs $\pa_1$ and $\pa_2$ with $\pa_i = \paTupleOf{i}$ and alphabets $\alphabetOf{i}$ for $i=1,2$.

\begin{definition}[Strategy Projection for PA]\label{def:projectionStrategyPA}
%	Let the valuation $\valuation$ be well-defined for $\ppa_1$ and $\ppa_2$. 
	%TODO? graph-preservong cause need Pr val be well-def for all pi in Mi
%Let $\pa_1$ and $\pa_2$ be non-parametric PAs and let $\pa = \pa_1 \parallel \pa_2$.
The \emph{projection} of a strategy $\strategy \in \strategysetOf{\pa}{\prt}$ to $\pa_{i}$ is the strategy 
    $\stratProjOfToValuation{\strategy}{i}{\pa}\in\strategysetOf{\pa_i}{\prt}$,
where for $\ppath_i \in \finPathsOf{\ppa_i}$ and $\alpha_i \in \actSetOf{i}$:
    \begin{align*}
		\stratProjOfToValuation{\strategy}{i}{\pa}(\ppath_i , \alpha_i) 
		= 
		\begin{cases}\displaystyle
			\frac{\sum_{s_i \in \stateSetOf{i}}
					\PrOf{\pa}{\strategy}{ \liftedpaths{(\ppath_i,\alpha_i,s_i)}{\pa_{3-i}}}
			}{\PrOf{\pa}{\strategy}{\liftedpaths{\ppath_i}{\pa_{3-i}}}}
				& \text{if } \PrOf{\pa}{\strategy}{\liftedpaths{\ppath_i}{\pa_{3-i}}} > 0 \\
			0 & \text{otherwise.}
		\end{cases}
	\end{align*}
\end{definition}
If $\pa$ is clear, we simply write $\stratProjOfToValuation{\strategy}{i}{}$ instead of $\stratProjOfToValuation{\strategy}{i}{\pa}$.
\Cref{lem:projection:alternative,lem:projection:measure} below yield that \Cref{def:projectionStrategyPA} is equivalent to the projection defined in \cite[Def.\ 6]{Kwi+13}.
We argue that our variant is more intuitive since the numerator and denominator of the fraction consider the same probability measure $\PrOf{\pa}{\strategy}{}$.
Intuitively, $\stratProjOfToValuation{\strategy}{i}{}(\ppath_i , \alpha_i)$ coincides with the conditional probability that---under $\PrOf{\pa}{\strategy}{}$ and given that a path $\ppath$ with projection $\restrOfTo{\ppath}{i} = \ppath_i$ is observed---the next action of component $\pa_i$ is $\alpha_i$.
We now provide an alternative characterization for the numerator given in \Cref{def:projectionStrategyPA}.
\begin{restatable}{lemma}{projectionAlternative}
\label{lem:projection:alternative}
For $\strategy \in \strategysetOf{\pa}{\prt}$, $\ppath_i \in \finPathsOf{\pa_i}$, and $\alpha_i \in \actSetOf{i}$:
\[
\sum_{s_i \in \stateSetOf{i}} \PrOf{\pa}{\strategy}{\liftedpaths{(\ppath_i,\alpha_i,s_i)}{\pa_{3-i}}} 
~=~\sum_{\ppath \in (\liftedpaths{\ppath_i}{\pa_{3-i}})}~\sum_{\substack{(\hat\alpha_1,\hat\alpha_2) \in \actSetOf{\parallel},\,\hat\alpha_i = \alpha_i}}  \PrOf{\pa}{\strategy}{\ppath} \cdot \strategy(\ppath,(\hat\alpha_1,\hat\alpha_2)).
\]
\end{restatable}
The next lemma is the key observation for strategy projections as it connects the probability measure $\PrOf{\pa}{\strategy}{}$ for $\pa$ and $\PrOf{\pa_i}{\stratProjOfToValuation{\strategy}{i}{}}{}$ for each component $\pa_i$.
\begin{restatable}{lemma}{projectionMeasure}
\label{lem:projection:measure}
For $\strategy \in \strategysetOf{\pa}{\prt}$ and $\ppath_i \in \finPathsOf{\pa_i}$: 
$\PrOf{\pa_i}{\stratProjOfToValuation{\strategy}{i}{}}{\ppath_i} = \PrOf{\pa}{\strategy}{\liftedpaths{\ppath_i}{\pa_{3-i}}}$.
\end{restatable}

The following result lifts \cite[Lemma 2]{Kwi+13} to our setting and states that fair (and therefore also complete) strategies have fair and complete projections.
\begin{lemma}
	\label{theo:strategy_projection_partial_fair_preserved}
Let $\decomp_1 \subseteq \setOfdecompsOf{\alphabetOf{1}}$ and $\decomp_2 \subseteq \setOfdecompsOf{\alphabetOf{2}}$. 
If $\strategy \in \strategysetOf{\pa}{\fairWrtRegionModel{}{\decomp_1 \cup \decomp_2}}$, then $\stratProjOfToValuation{\strategy}{i}{} \in \strategysetOf{\pa_i}{\fairWrtRegionModel{}{\decomp_i}}$. 
\end{lemma}
The converse of \Cref{theo:strategy_projection_partial_fair_preserved} does not hold:
The projection $\restrOfTo{\strategy}{i}$ might not be a complete strategy for $\pa_i$, even though $\strategy$ is complete for $\pa$.
Furthermore, $\restrOfTo{\strategy}{i}$ might not be memoryless, even if $\strategy$ is memoryless.
The following example shows both cases.
\begin{example}[Strategy Projection for pPA]\label{ex:projectionPA}
    Consider the pPA \( \ppa =  \ppa_1 \parallel \ppa_2 \) from \Cref{fig:pPA_composition}, using \( \ppa_1 \) and \( \ppa_2 \) depicted in \Cref{fig:ppa1_and_ppa2}.  
    % We fix the valuations \( \valuation_1 \), and \( \valuation_2 \) for $\ppa_1$ and $\ppa_2$, respectively, where \( \valuation_1(p) = \valuation_1(q) = 0.1 \), and \( \valuation_2(p) = \valuation_2(q) = 0.9 \). 
    We fix the valuation \( \valuation \) with \( \valuation(p) = \valuation(q) = 0.1 \) and set $\pa = \ppa[\valuation]$ and $\pa_i = \ppa_i[\valuation]$ for $i=1,2$.
    % , and \( \valuation(p) = \valuation(q) = 0.9 \). 
    Let \( \strategy \in \strategysetOf{\pa_1\parallel\pa_2}{\mless,\comp} \) that always selects the actions with label \( \lab \), \( \altaltlab \), or \( \frownie \) with probability 1 when available; otherwise, it chooses \( \altlab \).
    
    We compute the projection \( {\stratProjOfToValuation{\strategy}{2}{}} \) to the PA \( \ppa_2[\valuation] \): 
    \begin{itemize}
        \item \( {\stratProjOfToValuation{\strategy}{2}{}}(t_0,\lab) = 1 \)
        \item \( {\stratProjOfToValuation{\strategy}{2}{}}((t_0,\lab, t_2), \altaltlab) = \frac{(\valuation(p))^2}{\valuation(p)} = \valuation(p) = 0.1 \)
        \item \( {\stratProjOfToValuation{\strategy}{2}{}}((t_0,\lab, t_2, \altaltlab, t_3), \frownie) = \frac{(\valuation(p))^2 \cdot \frac{1}{10}}{\valuation(p) \cdot \frac{1}{10}} = \valuation(p) = 0.1 \)
    \end{itemize}
    Similarly,  \( {\stratProjOfToValuation{\strategy}{2}{}}((t_0, \lab, t_1), \lab) = \valuation(p) = 0.1 \) and \( {\stratProjOfToValuation{\strategy}{2}{}}((t_0, \lab, t_1, \lab, t_3), \frownie) = \valuation(p) = 0.1 \). 
    We observe that the projection is not a complete strategy as $\frownie$ is the only available action in $t_3$ for $\ppa_2$. 
    Moreover, the projection is not memoryless, because, for example, ${\stratProjOfToValuation{\strategy}{2}{}}((t_0, \lab, t_2, \altaltlab, t_4), \altaltlab) = \valuation(p) = 0.1$
    is not equal to ${\stratProjOfToValuation{\strategy}{2}{}}((t_0, \lab, t_1, \lab, t_4), \altaltlab) = 1.$
\end{example}
	
\subsection{Projections for Parametric PAs}
We now lift strategy projections to the parametric case. 
We fix the parallel composition $\ppa = \ppa_1 \parallel \ppa_2$ of two pPAs $\ppa_1$ and $\ppa_2$ with a common set of parameters
$\parameterSetOf{}$. Let $i=1,2$ and let $\valuation_i \colon \parameterSetOf{} \to \reals$ be a well-defined valuation for $\ppa_i$.
We define strategy projections for pPAs in terms of the instantiated PAs.
\begin{definition}[Strategy Projection for pPA]\label{def:projectionStrategy}
The \emph{projection} of a strategy $\strategy \in \strategysetOf{\ppa}{\prt}$ to pPA $\ppa_{i}$ w.r.t. $\valuation_1,\valuation_2$ is the strategy $\stratProjOfToValuation{\strategy}{i}{\valuation_1,\valuation_2} \in \strategysetOf{\ppa_i}{\prt}$, defined by
$\stratProjOfToValuation{\strategy}{i}{\valuation_1,\valuation_2} = \stratProjOfToValuation{\strategy}{i}{\ppa_1[\valuation_1] \parallel \ppa_2[\valuation_2]}$.
\end{definition}
If the valuations coincide, i.e., $\valuation_1=\valuation_2=\valuation$, we write $\stratProjOfToValuation{\strategy}{i}{\valuation}$ instead of $\stratProjOfToValuation{\strategy}{i}{\valuation_1,\valuation_2}$.
Strategy projections depend on the parameter instantiations as the following example illustrates.
Such strategies are also referred to as \emph{(parameter) dependent strategies}~\cite[Def.\ {2.6}]{Spel23}. 
\begin{example}[Strategy Projection for pPA]\label{ex:projectionPPA}
    Consider \( \ppa_1 \parallel \ppa_2 \) from \Cref{fig:pPA_composition} and the strategy \( \strategy \) as in \Cref{ex:projectionPA}. 
    Let \( \valuation_1 \) and \( \valuation_2 \) be the valuations used for \( \ppa_1 \) and \( \ppa_2 \), respectively, with \( \valuation_1(p) = \valuation_1(q) = 0.1 \), and \( \valuation_2(p) = \valuation_2(q) = 0.9 \). 
    For the projection \( \stratProjOfToValuation{\strategy}{2}{\valuation_1, \valuation_2} \) to \( \ppa_2 \), the following values are obtained:
    \begin{itemize}
        \item \( {\stratProjOfToValuation{\strategy}{2}{\valuation_1,\valuation_2}}(t_0,\lab) = 1 \)
        \item \( {\stratProjOfToValuation{\strategy}{2}{\valuation_1,\valuation_2}}((t_0,\lab, t_2), \altaltlab) = \frac{\valuation_1(p)\cdot \valuation_2(p)}{\valuation_2(p)} = \valuation_1(p) = 0.1 \)
        \item \( {\stratProjOfToValuation{\strategy}{2}{\valuation_1,\valuation_2}}((t_0,\lab, t_2, \altaltlab, t_3), \frownie) = \frac{\valuation_1(p) \cdot \valuation_2(p) \cdot \frac{1}{10}}{\valuation_2(p) \cdot \frac{1}{10}} = \valuation_1(p) = 0.1 \), 
        \item \( {\stratProjOfToValuation{\strategy}{2}{\valuation_1,\valuation_2}}((t_0, \lab, t_1), \lab) = \valuation_1(p) = 0.1 \), and 
        \item \( {\stratProjOfToValuation{\strategy}{2}{\valuation_1,\valuation_2}}((t_0, \lab, t_1, \lab, t_3), \frownie) = \valuation_1(p) = 0.1 \).
    \end{itemize}
    Note that the resulting projection coincides with the one computed in \Cref{ex:projectionPA}, where both components were instantiated using \( \valuation_1 \), i.e., $
    \stratProjOfToValuation{\strategy}{2}{\valuation_1, \valuation_2} = \stratProjOfToValuation{\strategy}{2}{\valuation_1} = \stratProjOfToValuation{\strategy}{2}{\valuation_1, \valuation_1}
    $.
\end{example}
The next lemma states that---when restricting to valuations that yield the same non-zero transitions---the strategy projection to $\ppa_i$ only depends on the parameter instantiation applied to $\ppa_{3-i}$.
This observation is the key insight for the correctness of the proof rule in \Cref{theo:pag_mono_rule}, which enables compositional reasoning about monotonicity.
\begin{restatable}{lemma}{changeValuationProjection}
\label{theo:change_valuation_projection}
For $i=1,2$, and well-defined valuations $\valuation_i, \valuation_i' \colon \parameterSetOf{} \to \reals$ for $\ppa_i$ such that $\transFctOf{i}(s_i,\alpha_i,s_i')[\valuation_i] = 0$ iff $\transFctOf{i}(s_i,\alpha_i,s_i')[\valuation_i'] = 0$ we have:
$\stratProjOfToValuation{\strategy}{1}{\valuation_1, \valuation_2} = \stratProjOfToValuation{\strategy}{1}{\valuation_1', \valuation_2}$
and 
$\stratProjOfToValuation{\strategy}{2}{\valuation_1, \valuation_2} = \stratProjOfToValuation{\strategy}{2}{\valuation_1, \valuation_2'}$.
\end{restatable}

\section{Verification Objectives for pPAs}\label{sec:verification}
We define properties of interest for pPA verification.
Let $\alphabetOf{}$ be a finite alphabet.
$\alphabetOf{}^{\infty} = \alphabetOf{}^* \cup \alphabetOf{}^{\omega}$ denotes the set of all finite and infinite words over $\alphabetOf{}$.
For a word $\trace = \lab_0, \lab_1, \dots \in \alphabetOf{}^\infty$ and another alphabet $\widehat{\alphabetOf{}}$, let $\restrOfTo{\trace}{\widehat{\alphabetOf{}}} \in \widehat{\alphabetOf{}}^\infty$ denote the projection of $\trace$ onto $\widehat{\alphabetOf{}}$---obtained by dropping all $\lab_i \in \alphabetOf{} \setminus \widehat{\alphabetOf{}}$ from $\trace$.
The restriction $\restrOfTo{\trace}{\widehat{\alphabetOf{}}}$ can be finite, even if $\trace$ is infinite.

We fix a pPA $\ppa = \ppaTupleOf{}$.
The trace of $\infpath = s_0, \alpha_0, s_1, \alpha_1, \dots \in \infPathsOf{\ppa}$ is the sequence $\traceOf{\infpath} = \syncOf{}(s_0, \alpha_0), \syncOf{}(s_1, \alpha_1), \dots $ of transition labels. 
The \emph{probability of a language} $\regLang \subseteq \alphabetOf{}^\infty$ at a well-defined valuation $\valuation$ under strategy $\strategy$ of $\ppa$ is given by
	\[
	\PrOf{\ppa}{\valuation,\strategy}{\regLang} = \PrOf{\ppa}{\valuation,\strategy}{\{ \pi \in \infPathsOf{\ppa}{} \mid \restrOfTo{\traceOf{\pi}}{\alphabetOf{}} \in \regLang  \}}. 
	\]

We also consider (parametric) \emph{expected total reward properties}. Let $\parameterSetOf{}$ be a set of parameters.
A \emph{reward function} $\rewFct \colon \alphabetOf{} \to \rats[{\parameterSetOf{}}] \cup \reals_{\geq 0}$ over $\alphabetOf{}$ assigns a (potentially parametric) reward to each symbol $\lab \in \alphabetOf{}$.
\emph{Instantiation} of $\rewFct$ at a valuation $\valuation \colon \parameterSetOf{} \to \reals$ yields $\rewFct[\valuation]$ with $\rewFct[\valuation](\lab) = \rewFct(\lab)[\valuation]$ for all $\lab \in \alphabetOf{}$. Valuation $\valuation$ is \emph{well-defined} for $\rewFct$ if $\rewFct[\valuation] \colon  \alphabetOf{} \to \reals_{\geq 0}$.
In this case, the accumulated reward for a word $\trace = \lab_0, \lab_1, \dots \in \alphabetOf{}^\infty$ is given by $\rewFct[\valuation](\trace) = \sum_{i=0}^{|\trace|} \rewFct[\valuation](\lab_i) \in \reals_{\geq 0} \cup \set{\infty}$.

When applied to a pPA $\ppa$, a reward function $\rewFct$ assigns the reward $\rewFct(\syncOf{}(s,\alpha))$ to the enabled state-action-pairs $(s,\alpha) \in \domain(\transFctOf{})$ with $\syncOf{}(s,\alpha) \in \alphabetOf{}$.
For a well-defined valuation $\valuation$ for $\ppa$ and $\rewFct$, we define the \emph{expected total reward} under strategy $\strategy$ as
\[\ExpTot{\ppa}{\valuation, \strategy}{\rewFct}= \textstyle \int_{\pi \in \infPathsOf{\ppa}{}}^{} \rewFct[\valuation]\big(\restrOfTo{\traceOf{\pi}}{\alphabetOf{}}\big) \,d \PrOf{\ppa}{\valuation, \strategy}{\pi}.\]

We consider probabilistic and reward-based objectives as well as their multi-objective combinations.
\begin{definition}[Objectives]
For ${\sim} \in \{>, \geq, <, \leq\}$, $p \in [0,1]$, and $r \in \reals_{\geq 0}$, we denote 
\begin{itemize}
    \item a \emph{probabilistic objective} over $\regLang \subseteq \alphabetOf{}^\infty$ by $\probPredicate{\sim p}{\regLang}$ and
    \item a \emph{reward objective} over $\rewFct \colon \alphabetOf{} \to \rats[{\parameterSetOf{}}] \cup \reals_{\geq 0}$ by $\expPredicate{\sim r}{\rewFct}$. 
    \end{itemize}
Their satisfaction for a well-defined valuation $\valuation$ and strategy $\strategy$ is defined by
	\[
		\ppa, \valuation, \strategy  \modelsWrt{} \probPredicate{\sim p}{\regLang}  \  \Leftrightarrow  \  \PrOf{\ppa}{\valuation,\strategy}{\regLang} \sim p 
        \quad  \text{ and}  \quad
        \ppa, \valuation, \strategy \modelsWrt{} \expPredicate{\sim r}{\rewFct}  \  \Leftrightarrow  \  \ExpTot{\ppa}{\valuation,\strategy}{\rewFct} \sim r . 
	\]
\end{definition}
Let $\generalPredicate{} \in \set{\probPredicate{\sim p}{\regLang}, \expPredicate{\sim r}{\rewFct}}$ refer to a (probabilistic or reward) objective. 
If neither $\ppa$ nor $\generalPredicate{}$ consider any parameters, we may drop the valuation from the notation and just write $\ppa, \strategy  \modelsWrt{} \generalPredicate{}$.
We lift the satisfaction relation $\modelsWrt{}$ to regions, i.e., sets of valuations.

\begin{definition}[Region Satisfaction Relation]\label{def:satregion}
	Let $\star \in \{\prt, \comp\} \cup \set{\fairC \mid \decomp \subseteq \setOfdecompsOf{\alphabetOf{\ppa}}}$.  
	For objective $\generalPredicate$
    and well-defined region $\region$ for $\ppa$---and $\rewFct$ if $\generalPredicate{} = \expPredicate{\sim r}{\rewFct}$---the \emph{region satisfaction relation} $\modelsWrt{\star}$ is given by: 
	\[
		\ppa, \region \modelsWrt{\star} \generalPredicate{} \quad \Leftrightarrow \quad \forall \valuation \in \region:  \forall \strategy \in \strategysetOf{\ppa}{\star}: \ppa, \valuation, \strategy \modelsWrt{} \generalPredicate{}. 
	\]
    Satisfaction under memoryless strategies---denoted by $\modelsWrt{\mless,\star}$---is defined similarly.
\end{definition}

\begin{remark}
For $\star \in \{\prt, \comp\}$ and any well-defined valuation $\valuation$ we have $\strategysetOf{\ppa[\valuation]}{\star} = \strategysetOf{\ppa}{\star}$.
Thus, for well-defined $\region$, we can swap the quantifiers in \Cref{def:satregion}: 
	\[
		\ppa, \region \modelsWrt{\star} \generalPredicate{} 
		\quad \Leftrightarrow \quad  \forall \strategy \in \strategysetOf{\ppa}{\star}: 
		\forall \valuation \in \region: \ppa, \valuation, \strategy \modelsWrt{} \generalPredicate{}. 
	\]
	However, this is not the case for \emph{fair} strategies and regions that are not graph-preserving: 
A strategy that is not $\fairC$ for $\ppa$ (under graph-preserving instantiations) might be $\fairC$ for $\ppa[\valuation]$ if $\valuation$ is not graph-preserving, because states that violate the fairness condition might not be reachable in $\ppa[\valuation]$.
	For a \emph{graph-preserving} region $\region$ and all $\valuation \in \region$, we have $\strategysetOf{\ppa}{\fairC} = \strategysetOf{\ppa[\valuation]}{\fairC}$. 
	Thus, we can swap quantifiers as above.
\end{remark}

Our framework also handles conjunctions of multiple objectives.
\begin{definition}[MO-Query]
A \emph{multi-objective} query (mo-query) is a set  $\multiobjectiveQuery{}{X}=\{ \generalPredicate{}_1, \dots , \generalPredicate{}_n\}$ of $n$ probabilistic or reward objectives with
	$
		\ppa, \valuation, \strategy  \modelsWrt{} \multiobjectiveQuery{}{X}  \  \Leftrightarrow  \  \ppa,\valuation, \strategy \modelsWrt{} \generalPredicate{}_i$  for all $\generalPredicate{}_i \in \multiobjectiveQuery{}{X}
	$.
\end{definition}
The conjunction of two mo-queries is a union of sets: $\multiobjectiveQuery{}{X}_1 \land \multiobjectiveQuery{}{X}_2 = \multiobjectiveQuery{}{X}_1 \cup \multiobjectiveQuery{}{X}_2$.
We lift objective satisfaction for regions (\Cref{def:satregion}) to mo-queries in a straightforward way.

\begin{remark}
\iftoggle{extended}{% extended version
    \Cref{app:equivalence_partial_complete_measures,app:equivalence_partial_complete_qmo} show that model checking under partial strategies in $\ppa$ for probabilistic properties, rewards, and multi-objective queries reduces to model checking under complete strategies in a modified pPA, denoted $\ppa_\tau$. 
    This result extends \cite[Proposition 2]{Kwi+13} to pPA while preserving memorylessness of the strategies. 
}{% Conefernce version: 
    In~\cite{extendedMer+25} it is shown that model checking under partial strategies in $\ppa$ for probabilistic properties, rewards, and multi-objective queries reduces to model checking under complete strategies in a modified pPA, denoted $\ppa_\tau$. 
    This result extends \cite[Proposition 2]{Kwi+13} to pPA while preserving memorylessness of the strategies. 
    % TODO add both sections from extended version
}
\end{remark}
We consider \emph{safety objectives} as a special type of probabilistic objectives.
\begin{definition}[Safety Objective]\label{def:safety_objective}
        $\probPredicate{\geq p}{\regLang}$ is a \emph{safety objective}\footnote{Note that safety objectives contain all finite prefixes of words in $\regLang$, i.e., they are prefix-closed. 
        This is different in \cite{Kwi+13}, where only infinite words are considered, leading to technical problems. See \Cref{sec:explain_differences}.} if  $\regLang$ can be characterized by a DFA $\bpAutomatonOf{\regLang}$ accepting a language of finite words (bad prefixes):
        \[
        \regLang = \{ w \in  \alphabetOf{}^{\infty} \mid \text{no prefix of $w$ is accepted by } \bpAutomatonOf{\regLang} \}.
        \]
\end{definition}
A mo-query is called \emph{safe}, denoted $\multiobjectiveQuery{\safe}{X}$, if each $\generalPredicate{}_i$ is a probabilistic safety objective. 

For PA, computing the probability for a safety objective reduces to maximal reachability properties in the PA-DFA product \cite[Lemma 1]{Kwi+13}. 
\iftoggle{extended}{% extended version
    This result can be lifted to pPA in a straightforward manner; see \Cref{theo:prob_safety_pPA_prop_reach_poduct_bijection} in \Cref{app:proofs_of_verification}. 
}{% Conefernce version:
    This result can be lifted to pPA in a straightforward manner; see~\cite{extendedMer+25}. 
    % TODO add theorem and chapter from extended version
}%
For reachability and safety objectives, it is equivalent to quantify over complete or partial strategies.
\Cref{lemma_safety_partial_vs_complete} lifts \cite[Proposition 1]{Kwi+13} to pPA. 
\begin{restatable}{lemma}{restatableSafetyPartVsComplete}\label{lemma_safety_partial_vs_complete}
    Let $\ppa$ be a pPA, let $\region$ be a well-defined region and let $\probPredicate{\geq p}{\regLang}$ be a safety objective for $\ppa$. 
    It holds that: 
    $\ppa, \region \modelsWrt{\comp} \probPredicate{\geq p}{\regLang}  \Leftrightarrow 
    \ppa, \region \modelsWrt{\prt} \probPredicate{\geq p}{\regLang}$.
    Same for $\modelsWrt{\mless,\star}$.
\end{restatable}

\subsection{Preservation Under Projection}\label{sec:projection_properties}
We generalize the result from \cite[Lemma 3]{Kwi+13}---originally stated in \cite[Lemma 7.2.6]{Seg95}---to the parametric setting. In particular, we show that probabilistic and reward properties in a composed pPA $\ppa = \ppa_1 \parallel \ppa_2$ under a strategy $\strategy$ are preserved when projecting to either component $\ppa_i$ over $\alphabetOf{i}$, assuming a well-defined valuation $\valuation$. 
\begin{restatable}{theorem}{restatableLemmaThreeNewPre}
    \label{theo:lemma3NewDependentPre}
    \label{lemma_3_aDependent}
    \label{lemma_3_cDependent}
		For $i=1,2$, let $\regLang$ be a language over $\alphabetOf{i}$ and $\rewFct$ be a reward function over $\alphabetOf{i}$. 
		Then, for a  well-defined valuation $\valuation$: 
        % NOT enough to be well-defined for $\ppa_1 \parallel \ppa_2$!!!
			\begin{align*}
				& \PrOf{{{\ppa_i}}}{\valuation, {\stratProjOfToValuation{\strategy}{i}{\valuation}}}{\regLang}
                %\alphabetExtensionOfTo{\ppa_i}{\alphabetOf{}}
				 =  
				\PrOf{{{{\ppa_1}} \parallel \ppa_2}}{\valuation, \strategy}{\regLang}
				\quad  \text{ and } \quad 
				\ExpTot{{\ppa_i}}{\valuation, {\stratProjOfToValuation{\strategy}{i}{\valuation}}}{\rewFct}
                %\{\ppa_i}
				 =  
				\ExpTot{{\ppa_1} \parallel {\ppa_2}}{\valuation, \strategy}{\rewFct}
			\end{align*} 	
\end{restatable}
\begin{example}\label{ex:projection_same_prob}
    Let $\regLang = \{ w \in \{\lab,\altaltlab,\frownie\}^{\infty} \mid \exists i \in \nats: w_i = \frownie\}$ be the language of words in which $\frownie$ occurs. 
    Reconsider the strategy $\strategy$ of $\ppa_1 \parallel \ppa_2$ from \Cref{ex:projectionPPA} and the projection ${\stratProjOfToValuation{\strategy}{{2}}{\valuation}}$ to $\ppa_2$. 
    We have 
    $\PrOf{\ppa_1 \parallel \ppa_2 }{\valuation, \strategy}{\regLang} % = p^2 \cdot \frac{1}{10} + p \cdot(1-p) \cdot q = 
    = \PrOf{\ppa_2}{\valuation, {\stratProjOfToValuation{\strategy}{2}{\valuation}}}{\regLang} %= p \cdot p \cdot \frac{1}{10} + (1-p) \cdot p \cdot q.$ 
    = (\valuation(p))^2 \cdot \frac{1}{10} + \valuation(p) \cdot(1-\valuation(p)) \cdot \valuation(q))$. 
\end{example}
Theorem~\ref{theo:lemma3NewDependentPre} assumes that the property only involves action labels from a single component $\ppa_i$.
To allow objectives over arbitrary alphabets $\alphabetOf{}$, we can add a self-loop labeled $\lab$ at every state, for each label $\lab \notin \alphabetOf{i}$.
\begin{definition}[Alphabet Extension]\label{def_alphabet_extension}
Let $\ppa = \ppaTupleOf{}$ be a pPA over $\alphabetOf{\ppa}$ and let $\alphabetOf{}$ be an alphabet with $\actSetOf{} \cap (\alphabetOf{} \setminus \alphabetOf{\ppa}) = \emptyset$.
The \emph{alphabet extension} of $\ppa$ with respect to $\alphabetOf{}$ is the pPA
$\alphabetExtensionOfTo{\ppa}{\alphabetOf{}} = (\stateSetOf{}, \initialOf{}, \parameterSetOf{}, \actSetOf{} \cupdot (\alphabetOf{} \setminus \alphabetOf{\ppa}), \transFctOf{\alphabetOf{}}, \syncOf{\alphabetOf{}})$ over alphabet $\alphabetOf{\ppa} \cup \alphabetOf{}$, where
\begin{itemize}
\item $\transFctOf{\alphabetOf{}}(s,\alpha) = \transFctOf{}(s,\alpha)$ and $\syncOf{\alphabetOf{}}(s,\alpha) = \syncOf{}(s,\alpha)$ for all $(s,\alpha) \in \domain(\transFctOf{})$ and
\item $\transFctOf{\alphabetOf{}}(s,\lab) = \indicatorFct{s}$ and $\syncOf{\alphabetOf{}}(s,\lab) = \lab$ for all $s \in \stateSetOf{}$ and $\lab \in \alphabetOf{}\setminus\alphabetOf{\ppa}$.
\end{itemize}
\end{definition}
\begin{figure}[ht]
    	\renewcommand{\modelType}{\ppa}
    	\renewcommand{\extended}{true} %extended with b
        \centering
    	{\begin{tikzpicture}[mdp]
	\node[ps, init=left] (t0)  {$t_0$};
	
	\node[ps,above right=0.25 and 2 of t0] (t1)  {$t_1$};

	\node[ps,below right=0.25 and 2 of t0] (t2)  {$t_2$};
	
	\node[ps,right=2 of t1, fill=lightgray] (t3)  {$t_3$};
	
	\node[ps,right=2 of t2] (t4)  {$t_4$};
	
	\node[dist, right=0.5 of t0] (d0a) {};
		
	% Probabilistic transitions
	\path[ptrans]

		(t0) edge[-] node[pos=0.95, below] {\tact{\lab}} (d0a)
		(d0a) edge[bend right=0] node[pos=0.5,above,yshift=0.5ex,xshift=-0.5ex] {\tprob{\ifthenelse{\equal{\modelType}{pa}}{\frac{9}{10}}{1-p}}} (t1)
		(d0a) edge[bend right=0] node[below,pos=0.55,yshift=-0.5ex,xshift=-0.5ex] {\tprob{\ifthenelse{\equal{\modelType}{pa}}{\frac{1}{10}}{p}}} node[above,pos=0.4] {} (t2)

		(t1) edge[bend right=0] node[pos=0.25,above] {\tact{\lab}} node[dist, pos=0.25] (d1a) {} node[pos=0.6,above] {\tprob{\ifthenelse{\equal{\modelType}{pa}}{\frac{9}{10}}{q}}} (t3)
		(d1a) edge[bend left=25] node[below,pos=0.25,xshift=-1ex] {\tprob{\ifthenelse{\equal{\modelType}{pa}}{\frac{1}{10}}{1-q}}} node[above,pos=0.4] {} (t4)
		
		(t2) edge[bend right=0] node[pos=0.25,below] {\tact{\altaltlab}} node[dist, pos=0.25] (d1b) {} node[pos=0.6,below] {\tprob{\ifthenelse{\equal{\modelType}{pa}}{\frac{9}{10}}{\frac{9}{10}}}} (t4)
		(d1b) edge[bend right=25] node[above, pos=0.25] {\tprob{\ifthenelse{\equal{\modelType}{pa}}{\frac{1}{10}}{\frac{1}{10}}}} node[above,pos=0.4] {} (t3)

		(t3) edge[loop right] node[pos=0.45,below right] {{\ifthenelse{\equal{\extended}{true}}{\tactext{\altlab},}{}}\tact{\small \frownie}} node[dist] (d3frown) {} node[pos=0.15,below] {\tprob{}} node[pos=0.25,below] {} (t3)
		
		(t4) edge[loop right] node[pos=0.5,above right] {\ifthenelse{\equal{\extended}{true}}{\tactext{\altlab},}{}{\tact{\altaltlab}}} node[dist] (d4b) {} node[pos=0.15,below] {\tprob{}} node[pos=0.25,below] {} (t4)

	;	
	
	{\ifthenelse{\equal{\extended}{true}}{\path[ptrans] (t2) edge[loop above] node[pos=0.75,right] {{\ifthenelse{\equal{\extended}{true}}{\tactext{\altlab}}{}}} node[dist] (d1c) {} (t2);}{}}
	
	{\ifthenelse{\equal{\extended}{true}}{\path[ptrans] (t1) edge[loop above] node[pos=0.75,right] {{\ifthenelse{\equal{\extended}{true}}{\tactext{\altlab}}{}}} node[dist] (d1c) {} (t2);}{}}
	
	{\ifthenelse{\equal{\extended}{true}}{\path[ptrans] (t0) edge[loop above] node[pos=0.75,right] {{\ifthenelse{\equal{\extended}{true}}{\tactext{\altlab}}{}}} node[dist] (d1c) {} (t0);}{}}

\end{tikzpicture}}
    	\caption{Alphabet extension $\alphabetExtensionOfTo{\ppa_2}{\{\lab,\altlab\}}$ of the pPA $\ppa_2$ from \Cref{fig:pPAM2} to the alphabet $\{\lab, \altlab\}$.
        }%
    	\label{fig:ppaExtension}
    \end{figure}
\begin{example}[Alphabet Extension]
 \Cref{fig:ppaExtension} shows $\alphabetExtensionOfTo{\ppa_2}{\{\lab,\altlab\}}$ for pPA $\ppa_2$ over $\alphabetOf{2} = \set{\lab,\altaltlab,\frownie}$ from \Cref{fig:pPAM2}.
 Transitions with label $\lab$ remain unchanged as $\lab \in \alphabetOf{2}$, but an additional self-loop with action $\altlab\not \in \alphabetOf{2}$ is added to every state. 
\end{example}
We now lift \cite[Lemma 3]{Kwi+13} to the parametric setting, covering properties and mo-queries over an alphabet that is not necessarily shared by $\ppa_i$: 
\begin{restatable}{theorem}{restatableLemmaThreeNew}
%[Projections Preserve Properties and Rewards]
    \label{theo:lemma3NewDependent}
		Let $\alphabetOf{} \subseteq \alphabetOf{\ppa_1\parallel \ppa_2}$, and $\sigma$ be a strategy for ${{\alphabetExtensionOfTo{\ppa_1}{\alphabetOf{}}} \parallel {\alphabetExtensionOfTo{\ppa_2}{\alphabetOf{}}}}$. 
		Let $\regLang$ be a language over $\alphabetOf{}$ and $\rewFct$ be a reward function over $\alphabetOf{}$. 
		Then, for well-defined valuation $\valuation$:      
        \label{lemma_3_bDependent}
        \label{lemma_3_dDependent}
			\begin{align*}
				& \PrOf{{\alphabetExtensionOfTo{\ppa_i}{\alphabetOf{}}}}{\valuation, {\stratProjOfToValuation{\strategy}{i}{\valuation}}}{\regLang}
                %\alphabetExtensionOfTo{\ppa_i}{\alphabetOf{}}
				 =  
				\PrOf{{{\alphabetExtensionOfTo{\ppa_1}{\alphabetOf{}}} \parallel {\alphabetExtensionOfTo{\ppa_2}{\alphabetOf{}}}}}{\valuation, \strategy}{\regLang}
				\quad  \text{ and } \quad 
				\ExpTot{{\alphabetExtensionOfTo{\ppa_i}{\alphabetOf{}}}}{\valuation, {\stratProjOfToValuation{\strategy}{i}{\valuation}}}{\rewFct}
                %\alphabetExtensionOfTo{\ppa_i}{\alphabetOf{}}
				 =  
				\ExpTot{{{\alphabetExtensionOfTo{\ppa_1}{\alphabetOf{}}} \parallel {\alphabetExtensionOfTo{\ppa_2}{\alphabetOf{}}}}}{\valuation, \strategy}{\rewFct}
			\end{align*} 	  
        Let  $\multiobjectiveQuery{}{X}$ be a mo-query over $\alphabetOf{}$. 
	    Then, for any well-defined valuation $\valuation$: 
          \label{theo:lemma3NewMulti}
         \label{lemma_3_eDependent} 
        \label{lemma_3_fDependent} 
		\begin{align*}	
			& { \alphabetExtensionOfTo{\ppa_i}{\alphabetOf{} }}, \valuation, { {\stratProjOfToValuation{\strategy}{i}{\valuation}} } 
            %\alphabetExtensionOfTo{\ppa_i}{\alphabetOf{}}
            \modelsWrt{}  \multiobjectiveQuery{}{X}
			\quad {\Leftrightarrow} \quad 
			({{\alphabetExtensionOfTo{\ppa_1}{\alphabetOf{}}} \parallel {\alphabetExtensionOfTo{\ppa_2}{\alphabetOf{}}}}), \valuation, {\strategy} \modelsWrt{} \multiobjectiveQuery{}{X} 
		\end{align*}	
\end{restatable}
\begin{remark}\label{rem:extension}
    Since alphabet extensions add self-loop transitions for new labels, and thus do not change the system's state, the pPAs ${{\alphabetExtensionOfTo{\ppa_1}{\alphabetOf{}}} \parallel {\alphabetExtensionOfTo{\ppa_2}{\alphabetOf{}}}}$ and $\ppa_1 \parallel \ppa_2$ satisfy the same properties and mo-queries over the alphabet 
    $\alphabetOf{} \subseteq \alphabetOf{\ppa_1 \parallel \ppa_2}$. 
\end{remark}
\Cref{theo:lemma3NewDependentPre,theo:lemma3NewDependent} play a key role in the proof of the AG framework for reasoning about mo-queries and monotonicity, which will be established in \Cref{sec:pag,sec:pag_mono}.

\section{Assume-Guarantee Reasoning for pPA}\label{sec:pag}
Kwiatkowska et al.\ \cite{Kwi+13} introduced assume-guarantee (AG) reasoning proof rules for PA. 
This section extends their proof rules to the parametric setting. 
%In this section, we adapt this style of reasoning to parametric probabilistic automata (pPA). 
We first generalize the concept of AG triples to pPAs in \Cref{sec:triples_for_pPA}. 
Then, we extend the asymmetric and circular proof rule in \Cref{sec:AG_for_pPA}. 
\iftoggle{extended}{% extended version
    Additional proof rules from \cite{Kwi+13} are presented in \Cref{app:pag_extension}. 
}
{% Conefernce version: 
    Additional proof rules from \cite{Kwi+13} are presented in~\cite{extendedMer+25}. 
% TODO add Chapter
}

\subsection{Assume-Guarantee Triples for pPA}\label{sec:triples_for_pPA}
We extend compositional reasoning to the parametric setting by generalizing assume-guarantee (AG) triples. Intuitively, an AG triple states that if a component satisfies an assumption, it also satisfies the guarantee under the same strategy and valuation. 
\begin{definition}[AG Triple]\label{def_pag_triple}
    The \emph{assume-guarantee triple} for $\ppa$, (parametric) mo-queries $\multiobjectiveQuery{}{A}$ (assumption) and $\multiobjectiveQuery{}{G}$ (guarantee), well-defined region $\region$, and $\star \in \{\comp, \prt, \fairC\}$ is 
	\begin{align*}
	\agTriple{\ppa, \region}{\star}{\multiobjectiveQuery{}{A}}{\multiobjectiveQuery{}{G}} \quad 
	& \Leftrightarrow \quad 
    \left(\forall \valuation \in \region : \forall \strategy \in \strategysetOf{\ppa[\valuation]}{\star} : \quad  \ppa, \valuation, {\strategy} \models \multiobjectiveQuery{}{A} \quad \rightarrow  \quad \ppa, \valuation, {\strategy} \models \multiobjectiveQuery{}{G} \right)
	\end{align*}
\end{definition}

\subsection{Assume-Guarantee Rules for pPA}
\label{sec:AG_for_pPA}
We present AG proof rules for the compositional verification of parametric probabilistic automata (pPAs). 
In the remainder of this section, we fix two pPAs $\ppa_1$ and $\ppa_2$ with alphabets $\alphabetOf{1}$, and $\alphabetOf{2}$, respectively. Further, let $\region_i$ be a well-defined region for $\ppa_i$. 

First, we establish the asymmetric proof rule for safety and mo-queries---based on \cite[Theorem 1 and 2]{Kwi+13}, respectively---for pPA. %
\begin{restatable}[Asymmetric Rule]{theorem}{restatableAsymRule}\label{theo:pag_asym_rule}
    Let $\multiobjectiveQuery{}{A}$ and $\multiobjectiveQuery{}{G}$ be mo-queries over $\alphabetOf{\multiobjectiveQuery{}{A}} \subseteq \alphabetOf{1}$ and 
    $\alphabetOf{\multiobjectiveQuery{}{G}} \subseteq \alphabetOf{2} \cup \alphabetOf{\multiobjectiveQuery{}{A}}$, respectively. 
    Let $\decomp_1 \subseteq \setOfdecompsOf{\alphabetOf{1}}$ and $\decomp_2 \subseteq \setOfdecompsOf{\alphabetOf{2} \cup \alphabetOf{\multiobjectiveQuery{}{A}}}$. 
    Then, the two proof rules holds:
    
     \begin{tabularx}{\linewidth}{p{0.4\linewidth}p{0.4\linewidth}} 
            %\\  \hline
            %\parbox{0.4\textwidth}
                 $
                 \infer{
                     \ppa_1 \parallel \ppa_2,  \regionIntersectionOf{\region_1}{}{\region_2} \modelsWrt{\comp} \multiobjectiveQuery{\safe}{G}
                 } 
                 {
                     \deduce{
                         \agTriple{\alphabetExtensionOfTo{\ppa_2}{\alphabetOf{\multiobjectiveQuery{\safe}{A}}}, \region_2 }{\prt}{\multiobjectiveQuery{\safe}{A}}{\multiobjectiveQuery{\safe}{G}} 
                     }{
                                \deduce{
                                    \ppa_1, \region_1 \modelsWrt{\comp} \multiobjectiveQuery{\safe}{A}
                                    }{}
                    }
                 }
                 $
         & 
        % \parbox{0.4\textwidth}
         %     \begin{align*}
          \hspace{30pt}
             $
                 \infer{
                     \ppa_1 \parallel \ppa_2,   \regionIntersectionOf{\region_1}{}{\region_2}
                     \modelsWrt{\fairWrtRegionModel{}{
                     \decomp_1 \cup  \decomp_2}} \multiobjectiveQuery{}{G}
                 }
                 {
                    \deduce{
                        \agTriple{\alphabetExtensionOfTo{\ppa_2}{\alphabetOf{\multiobjectiveQuery{}{A}}}, \region_2}{\fairWrtRegionModel{}{\decomp_2}}{\multiobjectiveQuery{}{A}}{\multiobjectiveQuery{}{G}} 
                    }{
                        \deduce{
                            \ppa_1, \region_1 \modelsWrt{\fairWrtRegionModel{}{\decomp_1}} \multiobjectiveQuery{}{A}
                        }{}
                    }
                 }
              $
         % \end{align*}
        %  \\ %
        %  \\
        %  &  % $\decomp_1 \in \setOfdecompsOf{\ppa_1}$ and 
        % for .  
     \end{tabularx}
     
\end{restatable}
\begin{proof}[Proof sketch]
Let $\valuation \in \regionIntersectionOf{\region_1}{}{\region_2}$, and $\strategy$ be a strategy for the composed pPA $\ppa_1 \parallel \ppa_2$. 
To prove validity of the rule, we need to show that $\ppa_1 \parallel \ppa_2$ instantiated with $\valuation$ satisfies $\multiobjectiveQuery{\safe}{G}$. 
\begin{enumerate}
    \item Since $\valuation \in \region_1$, the first premise implies $\ppa_1, \valuation \modelsWrt{\comp} \multiobjectiveQuery{\safe}{A}$, which is equivalent to $\ppa_1, \valuation \modelsWrt{\prt} \multiobjectiveQuery{\safe}{A}$ by \Cref{lemma_safety_partial_vs_complete}. 
    This implies that $\ppa_1$ under the partial strategy $\stratProjOfToValuation{\strategy}{\ppa_1}{\valuation}$ also satisfies $\multiobjectiveQuery{\safe}{A}$. 
    Since strategies and their projections satisfy the same properties (as shown in \Cref{theo:lemma3NewDependent}), we conclude that $\ppa_1 \parallel \ppa_2$ instantiated at $\valuation$ under the strategy $\strategy$ satisfies $\multiobjectiveQuery{\safe}{A}$. 
    \item As $\valuation \in \region_2$, the second premise implies that $\ppa_2[\valuation]$ under the strategy $\stratProjOfToValuation{\strategy}{\ppa_2}{\valuation}$ satisfies $\multiobjectiveQuery{\safe}{G}$. 
    Again,  \Cref{theo:lemma3NewDependent} implies that $(\ppa_1 \parallel \ppa_2)[\valuation]$ under the strategy $\strategy$ satisfies $\multiobjectiveQuery{\safe}{G}$.
\end{enumerate}
Thus, we conclude that $\ppa_1 \parallel \ppa_2$ instantiated at $\valuation$ under $\strategy$ satisfies $\multiobjectiveQuery{\safe}{G}$. 
The rule on the right holds by a similar reasoning, where, in addition, \Cref{theo:strategy_projection_partial_fair_preserved} is used to establish that projections of fair strategies remain fair. 
\end{proof}

\begin{example}[Asymmetric Rule]
    We illustrate the left proof rule from \Cref{theo:pag_asym_rule} for the pPA $\ppa_1 \parallel \ppa_2$ in \Cref{fig:pPA_composition}---composed of the pPAs $\ppa_1$ and $\ppa_2$ depicted in \Cref{fig:ppa1_and_ppa2}---and w.r.t.\ $\multiobjectiveQuery{}{G} = \probPredicate{\geq 0.8}{\regLang_G}$, where $\regLang_G = \{ w \in \{\lab, \altlab, \altaltlab, \frownie\}^{\infty} \mid \vert w \vert_{\frownie} = 0 \}$. 
    Let $\multiobjectiveQuery{}{A} = \{\probPredicate{\geq 0.9}{\regLang_A}\}$, where $\regLang_A = \{ w \in \{\lab, \altlab\}^{\infty} \mid  \vert  w  \vert_{\lab} \leq 1  \}$. 
    The pPA $\alphabetExtensionOfTo{\ppa_2}{\{\lab,\altlab\}}$ is depicted in \Cref{fig:ppaExtension}. 
    For the premises of the proof rule, we obtain that the (largest) region $\region$ for which $\ppa_1, \region_1 \modelsWrt{\comp} \probPredicate{\geq 0.9}{\regLang_A}$ is $\region_1 = \{\valuation \colon \{p,q\} \to \reals \mid \valuation(p) \in [0,0.1]\}$ and 
    the (largest) region $\region_2$ for which $\agTriple{\alphabetExtensionOfTo{\ppa_2}{\{\lab,\altlab\}}, \region_2}{\prt}{\multiobjectiveQuery{}{A}}{\multiobjectiveQuery{}{G}}$ holds, is $\region= \{\valuation \colon \{p,q\} \to \reals \mid (\valuation(p) \in [0,0.5], \valuation(q) \in [0,1]) \lor  
    (\valuation(p) \in (0.5,1), \valuation(q) \in [0,2-2\cdot p]) \}$. 
    The intersection $\regionIntersectionOf{\region_1}{}{\region_2}$---for which $\ppa_1 \parallel \ppa_2, \regionIntersectionOf{\region_1}{}{\region_2} \modelsWrt{\prt}  \multiobjectiveQuery{}{G}$ holds by \Cref{theo:pag_asymN_rule}---contains all valuations with $\valuation(p) \in [0,0.1], \valuation(q) \in [0,1]$. 
    The (largest) region $\region$ for which $\ppa_1 \parallel \ppa_2, \region \models^{\prt}  \multiobjectiveQuery{}{G}$ is $\region = \{\valuation \colon {p,q} \to \reals \mid 
    (\valuation(p) \in ([0, \frac{1}{4}] \cup \{1\} ),  \valuation(q) \in [0,1]) \lor (\valuation(p) \in (\frac{1}{4}, 1),  \valuation(q) \in [0,\frac{p+1}{5 \cdot p}]) \}$. 
    This satisfies $(\regionIntersectionOf{\region_1}{}{\region_2}) \subset \region$. 
\end{example}
\iftoggle{extended}{% extended version
   The proof rules in \Cref{theo:pag_asym_rule} can be extended to systems with more than two components, as detailed in \Cref{theo:pag_asymN_rule} in \Cref{app:pag_extension}. 
}{% Conefernce version: 
    The proof rules in \Cref{theo:pag_asym_rule} can be extended to systems with more than two components, as detailed in \cite{extendedMer+25}. 
        % TODO add [Theorem Number] (and chapter)
}
Next, we lift the circular proof rule given in \cite[Theorem 5]{Kwi+13} to pPAs: 
\begin{restatable}[Circular Rule]{theorem}{restatableCircRule}\label{theo:pag_circ_rule}
    Let $\multiobjectiveQuery{}{A}_1$,$\multiobjectiveQuery{}{A}_2$ and $\multiobjectiveQuery{}{G}$ be (parametric) mo-queries over $\alphabetOf{\multiobjectiveQuery{}{A}_1} \subseteq \alphabetOf{2}$, $\alphabetOf{\multiobjectiveQuery{}{A}_2} \subseteq \alphabetOf{1} \cup \alphabetOf{\multiobjectiveQuery{}{A}_1}$ and  
    $\alphabetOf{\multiobjectiveQuery{}{G}} \subseteq \alphabetOf{2} \cup \alphabetOf{\multiobjectiveQuery{}{A}_2}$, respectively. 
    Let $\decomp_i \in \setOfdecompsOf{\alphabetOf{i}\cup \alphabetOf{\multiobjectiveQuery{}{A}_{i}}}$ for $i \in \{1,2\}$, and $\decomp_3 \in \setOfdecompsOf{\alphabetOf{2}}$. 
    Then: 
    
    \begin{tabularx}{\linewidth}{p{0.45\linewidth} p{0.45\linewidth}} 

                $
                \infer{
                    \ppa_1 \parallel \ppa_2, \regionIntersectionOf{\region_1}{\region_2}{\region_3} \modelsWrt{\comp} \multiobjectiveQuery{\safe}{G}
                } 
                {
                    \deduce{
                        \deduce{
                            \ppa_2, \region_3 \modelsWrt{\comp} \multiobjectiveQuery{\safe}{A}_1
                            }{
                                \agTriple{\alphabetExtensionOfTo{\ppa_2}{\alphabetOf{\multiobjectiveQuery{\safe}{A}_{2}}}, \region_2}{\prt}{\multiobjectiveQuery{\safe}{A}_{2}}{\multiobjectiveQuery{\safe}{G}} 
                            }
                    }{\agTriple{\alphabetExtensionOfTo{\ppa_1}{\alphabetOf{\multiobjectiveQuery{\safe}{A}_{1}}}, \region_1 }{\prt}{\multiobjectiveQuery{\safe}{A}_{1}}{\multiobjectiveQuery{\safe}{A}_2} 
                    }
                }
                $
        & 
       % \parbox{0.4\textwidth}
        \hspace{8pt}
               $
                \infer{
                    \ppa_1 \parallel \ppa_2, \regionIntersectionOf{\region_1}{\region_2}{\region_3} \modelsWrt{\fairWrtRegionModel{}{\decomp_1 \cup \decomp_2 \cup \decomp_3}} \multiobjectiveQuery{}{G}
                } 
                {
                    \deduce{
                        \deduce{
                            \ppa_2, \region_3 \modelsWrt{\fairWrtRegionModel{}{\decomp_3}} \multiobjectiveQuery{}{A}_1
                            }{
                                \agTriple{\alphabetExtensionOfTo{\ppa_2}{\alphabetOf{\multiobjectiveQuery{}{A}_{2}}}, \region_2}{\fairWrtRegionModel{}{\decomp_2}}{\multiobjectiveQuery{}{A}_{2}}{\multiobjectiveQuery{}{G}} 
                            }
                    }{\agTriple{\alphabetExtensionOfTo{\ppa_1}{\alphabetOf{\multiobjectiveQuery{}{A}_{1}}}, \region_1 }{\fairWrtRegionModel{}{\decomp_1}}{\multiobjectiveQuery{}{A}_{1}}{\multiobjectiveQuery{}{A}_2} 
                    }
                }
                $
              %  \\ %
              %  \\
    \end{tabularx}
\end{restatable}
% proof: \restatableCircRule*
\begin{proof}[Proof sketch]
Similar to \Cref{theo:pag_asym_rule}, the proof of the circular rules makes use of \Cref{theo:lemma3NewDependent}, which establishes that the composition under a strategy satisfies the same properties as the individual components under their corresponding projections.
For safety, \Cref{lemma_safety_partial_vs_complete} allows us to verify the condition for complete strategies rather than partial strategies in the third premise. 
For fairness, \Cref{theo:strategy_projection_partial_fair_preserved} ensures that strategy projections remain fair. 
\end{proof}
\begin{remark}
    The inclusion of fairness in the premises of the right rules in \Cref{theo:pag_asym_rule} and \Cref{theo:pag_circ_rule} enables recursive application  and thus supports the compositional verification of systems with more than two components. 
    In the case of a single application of one of the rules, it is sufficient to verify with respect to complete strategies, which, while a stronger condition, simplifies the verification process.
\end{remark}

\section{Compositional Reasoning about Monotonicity}\label{sec:pag_mono}
%TODO adjust fair notation!
Exploiting monotonicity can significantly enhance the efficiency of parameter synthesis~\cite{Spe+21}. 
However, determining monotonicity is computationally hard\footnote{For deterministic pPA (Markov chains) determining monotonicity is coETR-hard~\cite[Sec.\ 3.4]{Spel23}.} and it would be beneficial to determine monotonicity in a compositional way. 
Additionally, monotonicity in composed pPA is challenging due to the complexities introduced by parameter dependencies and interactions among components. 
%Therefore, we elaborate on reasoning about monotonicity in a compositional manner. %. in a an AG reasoning-style.
While we focus on global monotonicity, the following results can be extended to \emph{local monotonicity}, which considers only the first transition from a given state. See \cite[Definitions 4.4 and 4.5]{Spel23}.

The probability of a language or the expected total reward for a pPA $\ppa$ can be viewed as a function\textemdash{}called \emph{solution function}\textemdash{}that maps a well-defined parameter valuation to the corresponding probability or expected total reward, respectively \cite[Definition~4.7]{Jun20}. 
\begin{definition}[Solution Function]
    Let $\ppa$ be a pPA over $\alphabetOf{}$, let $\strategy \in \strategysetOf{\ppa}{}$ and let $\region$ be a well-defined region.
    The \emph{solution function} for $\ppa$ and language $\regLang \subseteq \alphabetOf{}^\infty$ 
    is $\solutionFctMdpObjective{\ppa, \strategy}{\PrOf{}{}{\regLang}} \colon \region  \to [0,1]$, where $\solutionFctMdpObjective{\ppa, \strategy}{\PrOf{}{}{\regLang}}(\valuation) = 
    \PrOf{\ppa}{\valuation,\strategy}{\regLang}$. 
    The solution function for $\ppa$ and a reward function $\rewFct$ over $\alphabetOf{}$ is $\solutionFctMdpObjective{\ppa, \strategy}{\ExpTot{}{}{\rewFct}} \colon \region  \to \reals_{\geq 0}$, where $\solutionFctMdpObjective{\ppa, \strategy}{\ExpTot{}{}{\rewFct}}(\valuation) = 
    \ExpTot{\ppa}{\valuation,\strategy}{\rewFct}$. 
\end{definition}
When referring to a solution function without specifying whether it pertains to probabilities or expected rewards, we simply write $\solutionFctMdpObjective{\ppa, \strategy}{}$. 

\begin{example}[Solution Function]\label{ex:solfct}
    Consider the pPA $\ppa_1\parallel\ppa_2$ in \Cref{fig:pPA_composition} and the region $\region =\{\valuation \colon \{p,q\} \to [0,1]\}$ which is well-defined for $\ppa_1 \parallel \ppa_2$. 
    Let $\regLang = \{ w \in \{\lab,\altaltlab,\frownie\}^{\infty} \mid \vert w \vert_{\frownie} = 0 \}$ be the language of words over $\{\lab,\altlab,\altaltlab, \frownie\}$ that do not contain $\frownie$. 
    Let $\strategy$ be the complete strategy of $\ppa_1\parallel\ppa_2$ from \Cref{ex:projectionPPA}, which always selects action $\lab,\altaltlab$ or $\frownie$ with probability 1 whenever any of them is enabled; otherwise, it chooses $\altlab$ with probability 1. 
    The solution function $\solutionFctMdpObjective{\ppa_1 \parallel \ppa_2, \strategy}{\PrOf{}{}{\regLang}} \colon \region \to [0,1]$ is defined by 
    $\solutionFctMdpObjective{\ppa_1 \parallel \ppa_2, \strategy}{\PrOf{}{}{\regLang}}(p,q) = 
    1- \left(p^2 \cdot \frac{1}{10} + p \cdot(1-p) \cdot (p\cdot q + (1-p) \cdot q)\right) 
   % p^2 \cdot \frac{9}{10} + p \cdot (1-p) \cdot \left( 1-q \right) +p
   = 1- \left(p^2 \cdot \frac{1}{10} + (p-p^2) \cdot q\right) 
    $. 
\end{example}

% \subsection{Monotonicity}
\label{sec:monotonicity}
We extend the standard notion of monotonicity~\cite{Spel23} by differentiating between different strategy classes, including complete, partial, and fair strategies. 
\begin{definition}[Monotonicity]
    \label{def:monotonicity}
    Let $\strategy$ be a strategy of $\ppa$. A solution function $\solutionFctMdpObjective{\ppa,\strategy}{}$ is \emph{monotonic increasing} in $p\in \parameterSetOf{}$ on region $\region$\textemdash{}denoted $\monotonicOnRegionParameter{\buparrow}{\solutionFctMdpObjective{\ppa, {\strategy}}{}}{p}{\region}{}$\textemdash{}if for all $\valuation, \valuation_{+} \in \region$ with $\valuation_{+}(q) = \valuation(q) + x \cdot \iverson{p{=}q}$ for $q \in \parameterSetOf{}$ and some $x \geq 0$, we have: 
    ${\solutionFctMdpObjective{{\ppa},{\strategy}}{}}{}(\valuation) \leq {\solutionFctMdpObjective{{\ppa},{\strategy}}{}}{}(\valuation_{+}).$

    For $\star \in \{ \prt, \comp\}$, %\fairWrtRegionModel{\region}{\ppa}
    we write $\monotonicOnRegionParameter{\buparrow}{\solutionFctMdpObjective{\ppa}{}}{p}{\region}{\star}$ if $\monotonicOnRegionParameter{\buparrow}{\solutionFctMdpObjective{\ppa, {\strategy}}{}}{p}{\region}{}$ for all $\strategy \in \strategysetOf{\ppa}{\star}$.
    If $\region$ is graph-preserving, we write $\monotonicOnRegionParameter{\buparrow}{\solutionFctMdpObjective{\ppa}{}}{p}{\region}{\fairWrtRegionModel{}{\decomp}}$ if $\monotonicOnRegionParameter{\buparrow}{\solutionFctMdpObjective{\ppa, {\strategy}}{}}{p}{\region}{}$ holds for all fair strategies $\strategy \in \strategysetOf{\ppa[\valuation]}{\fairWrtRegionModel{}{\decomp}}$, $\valuation \in \region$.

    \noindent Notations $\monotonicOnRegionParameter{\bdownarrow}{\solutionFctMdpObjective{\ppa, {\strategy}}{}}{p}{\region}{}$ and $\monotonicOnRegionParameter{\bdownarrow}{\solutionFctMdpObjective{\ppa}{}}{p}{\region}{\star}$ for \emph{monotonic decreasing} $\solutionFctMdpObjective{\ppa,\strategy}{}$ are defined analogously. 
\end{definition}
We require the region to be graph-preserving when defining monotonicity w.r.t.\ fair strategies. 
This ensures that for any two valuations, $\valuation, \valuation_+$, we have $\strategysetOf{\ppa[\valuation]}{\fairWrtRegionModel{}{\decomp}}= \strategysetOf{\ppa[\valuation_+]}{\fairWrtRegionModel{}{\decomp}}$; see \Cref{theo:graph_preserving_fairness}. 
\begin{remark}
     \iftoggle{extended}{% extended version
        Monotonicity for partial strategies w.r.t.\ general properties is equivalent to monotonicity for complete strategies in a modified pPA; see \Cref{app:equivalence_partial_complete_mono}. 
    }
    {% Conefernce version: 
        Monotonicity for partial strategies w.r.t.\ general properties is equivalent to monotonicity for complete strategies in a modified pPA; see~\cite{extendedMer+25}. 
        % TODO add [Theorem Number]
    }
\end{remark}

The following theorem states that  monotonicity of a composed system can be verified by analyzing its individual components. 
\begin{restatable}[Monotonicity]{theorem}{restatablMonoRule}\label{theo:pag_mono_rule}
    Let $\ppa_1$, $\ppa_2$ be pPAs with alphabets $\alphabetOf{1}$ and $\alphabetOf{2}$ and $\region_i$ be a graph-preserving region for $\ppa_i$. 
    %\localbuparrow, \localbdownarro
    Let $\solutionFctMdpObjective{}{} \in \{\solutionFctMdpObjective{}{\PrOf{}{}{\regLang}}, \solutionFctMdpObjective{}{\ExpTot{}{}{\rewFct}} \}$ 
	be a solution function w.r.t.\ the language $\regLang$ or reward function $\rewFct$ over $\alphabetOf{} \subseteq (\alphabetOf{1}\cup \alphabetOf{2})$ and 
    let $\budarrow \in \{\buparrow, \bdownarrow\}$. % $\alphabetOf{} \subseteq \alphabetOf{\ppa_1} \cup \alphabetOf{\ppa_2}$ 
    Let $\decomp_i \subseteq \setOfdecompsOf{\alphabetOf{1} \cup {\alphabetOf{}}}$. 
    Then the following two proof rules hold:

    \begin{tabularx}{\linewidth}{p{0.45\linewidth}p{0.45\linewidth}} 
            \hspace{20pt}
                $
                \infer{
                    \monotonicOnRegionParameter{\budarrow}{\solutionFctMdpObjective{\ppa_1 \parallel \ppa_2}{}}{p}{\regionIntersectionOf{\region_1}{}{\region_2}}{\prt}
                } 
                {
                    \deduce{\monotonicOnRegionParameter{\budarrow}{\solutionFctMdpObjective{\alphabetExtensionOfTo{\ppa_2}{\alphabetOf{}}}{}}{p}{\region_2}{\prt}
                    }{
                        \monotonicOnRegionParameter{\budarrow}{\solutionFctMdpObjective{\alphabetExtensionOfTo{\ppa_1}{\alphabetOf{}}}{}}{p}{\region_1}{\prt} 
                    }
                }
                $
        & 
       \hspace{20pt}
       $
       \infer{
           \monotonicOnRegionParameter{\budarrow}{\solutionFctMdpObjective{\ppa_1 \parallel \ppa_2}{}}{p}{\regionIntersectionOf{\region_1}{}{\region_2}}{\fairWrtRegionModel{}{\decomp_1 \cup \decomp_2}}
       } 
       {
           \deduce{\monotonicOnRegionParameter{\budarrow}{\solutionFctMdpObjective{\alphabetExtensionOfTo{\ppa_2}{\alphabetOf{}}}{}}{p}{\region_2}{\fairWrtRegionModel{}{\decomp_2}} 
           }{
               \monotonicOnRegionParameter{\budarrow}{\solutionFctMdpObjective{\alphabetExtensionOfTo{\ppa_1}{\alphabetOf{}}}{}}{p}{\region_1}{\fairWrtRegionModel{}{\decomp_1}}  
           }
       }
       $
    \end{tabularx} 
\end{restatable}

\begin{proof}
We show the premises imply $\monotonicOnRegionParameter{\budarrow}{\solutionFctMdpObjective{\alphabetExtensionOfTo{\ppa_1}{\alphabetOf{}} \parallel \alphabetExtensionOfTo{\ppa_2}{\alphabetOf{}}}{}}{p}{\regionIntersectionOf{\region_1}{}{\region_2}}{\star}$ for $\star \in \set{\prt,\fairWrtRegionModel{}{\decomp_1\cup\decomp_2}}$, which directly implies that $\monotonicOnRegionParameter{\budarrow}{\solutionFctMdpObjective{\ppa_1 \parallel \ppa_2}{}}{p}{\regionIntersectionOf{\region_1}{}{\region_2}}{\star}$ holds, see \Cref{rem:extension}.
We focus on the left rule, i.e., $\star = \prt$.
\begin{remark}
     \iftoggle{extended}{
        % extended version
       The proof for $\star = \fairWrtRegionModel{}{\decomp_1\cup\decomp_2}$ is similar but additionally requires \Cref{theo:strategy_projection_partial_fair_preserved}; see \Cref{app:proofs_of_pag_mono}. 
    }
    {
        % Conefernce version: 
        The proof for $\star = \fairWrtRegionModel{}{\decomp_1\cup\decomp_2}$ is similar but additionally requires \Cref{theo:strategy_projection_partial_fair_preserved}; see see~\cite{extendedMer+25}. 
        % TODO add [Theorem Number]
    }
\end{remark}
We further consider $\budarrow  = \buparrow$. 
The case $\budarrow =\bdownarrow$ follows analogously. 
Our proof is by contradiction.
        Assume that the premises hold but $\monotonicOnRegionParameter{\buparrow}{\solutionFctMdpObjective{\alphabetExtensionOfTo{\ppa_1}{\alphabetOf{}} \parallel \alphabetExtensionOfTo{\ppa_2}{\alphabetOf{}}}{}}{p}{\region_1 \cap \region_2}{\star}$ does not hold. 
        Thus, there is a strategy $\strategy \in  \strategysetOf{(\alphabetExtensionOfTo{\ppa_1}{\alphabetOf{}} \parallel \alphabetExtensionOfTo{\ppa_2}{\alphabetOf{}})}{\star}$ and 
        valuations $\valuation, \valuation_+ \in \region_1 \cap \region_2$
        with $\valuation_{+}(q) = \valuation(q) + x \cdot \iverson{p{=}q}$ for $q \in \parameterSetOf{}$ and some $x \geq 0$ and 
        \begin{align}\label{proof:contradiction}
            & {\solutionFctMdpObjective{{(\alphabetExtensionOfTo{\ppa_1}{\alphabetOf{}} \parallel \alphabetExtensionOfTo{\ppa_2}{\alphabetOf{}})},{\strategy}}{}}{}(\valuation) > 	
            {\solutionFctMdpObjective{{(\alphabetExtensionOfTo{\ppa_1}{\alphabetOf{}} \parallel \alphabetExtensionOfTo{\ppa_2}{\alphabetOf{}})},{\strategy}}{}}{}(\valuation_{+}). 
        \end{align}
        \Cref{theo:lemma3NewDependent} yields 
        $ {\solutionFctMdpObjective{{(\alphabetExtensionOfTo{\ppa_1}{\alphabetOf{}} \parallel \alphabetExtensionOfTo{\ppa_2}{\alphabetOf{}})},{\strategy}}{}}{}(\valuation) 
        = {\solutionFctMdpObjective{{\alphabetExtensionOfTo{\ppa_1}{\alphabetOf{}}},{{\stratProjOfToValuation{\strategy}{1}{\valuation, \valuation}}}}{}}{}(\valuation)$. 
        We  have $\transFctOf{\alphabetExtensionOfTo{\ppa_1}{\alphabetOf{}}}(s,\alpha,s')[\valuation] = 0$ iff $\transFctOf{\alphabetExtensionOfTo{\ppa_1}{\alphabetOf{}}}(s,\alpha,s')[\valuation_+] = 0$ as $\valuation$ and $\valuation_+$ are graph preserving for $\ppa_1$. 
        Thus, we can apply \Cref{theo:change_valuation_projection} and obtain 
        \begin{align*}
            {\solutionFctMdpObjective{{\alphabetExtensionOfTo{\ppa_1}{\alphabetOf{}}},{{\stratProjOfToValuation{\strategy}{1}{\valuation, \valuation}}}}{}}{}(\valuation) 
            & = {\solutionFctMdpObjective{{\alphabetExtensionOfTo{\ppa_1}{\alphabetOf{}}},{{\stratProjOfToValuation{\strategy}{1}{\valuation_+, \valuation}}}}{}}{}(\valuation) 
            \tag{by \Cref{theo:change_valuation_projection}}\\
            & \leq  {\solutionFctMdpObjective{{\alphabetExtensionOfTo{\ppa_1}{\alphabetOf{}}},{{\stratProjOfToValuation{\strategy}{1}{\valuation_+, \valuation}}}}{}}{}(\valuation_+)
             \tag{ ${\stratProjOfToValuation{\strategy}{1}{\valuation_+, \valuation}}\in  \strategysetOf{\alphabetExtensionOfTo{\ppa_1}{\alphabetOf{}}}{\prt}$,  $\monotonicOnRegionParameter{\buparrow}{\solutionFctMdpObjective{\alphabetExtensionOfTo{\ppa_1}{\alphabetOf{}}}{}}{p}{\region_1}{\prt}$} \\
            & = {\solutionFctMdpObjective{{(\alphabetExtensionOfTo{\ppa_1}{\alphabetOf{}} \parallel \alphabetExtensionOfTo{\ppa_2}{\alphabetOf{}}[\valuation])},{\strategy}}{}}{} (\valuation_+)
             \tag{by \Cref{theo:lemma3NewDependent}}\\
        \end{align*}
        We observe that
        \[\big(\alphabetExtensionOfTo{\ppa_1}{\alphabetOf{}} \parallel \alphabetExtensionOfTo{\ppa_2}{\alphabetOf{}}[\valuation]\big)[\valuation_+]
        ~=~\big(\alphabetExtensionOfTo{\ppa_1}{\alphabetOf{}}[\valuation_+] \parallel \alphabetExtensionOfTo{\ppa_2}{\alphabetOf{}}[\valuation]\big)
        ~=~\big(\alphabetExtensionOfTo{\ppa_1}{\alphabetOf{}}[\valuation_+] \parallel \alphabetExtensionOfTo{\ppa_2}{\alphabetOf{}}\big)[\valuation].\]
        Consequently,
        ${\solutionFctMdpObjective{{(\alphabetExtensionOfTo{\ppa_1}{\alphabetOf{}} \parallel \alphabetExtensionOfTo{\ppa_2}{\alphabetOf{}}[\valuation])},{\strategy}}{}}{} (\valuation_+) = {\solutionFctMdpObjective{{(\alphabetExtensionOfTo{\ppa_1}{\alphabetOf{}}[\valuation_+] \parallel \alphabetExtensionOfTo{\ppa_2}{\alphabetOf{}}},{\strategy}}{}}{} (\valuation)$. 
        By a similar reasoning as above, we obtain    
        \begin{align*}
            {\solutionFctMdpObjective{{(\alphabetExtensionOfTo{\ppa_1}{\alphabetOf{}}[\valuation_+] \parallel \alphabetExtensionOfTo{\ppa_2}{\alphabetOf{}})},{\strategy}}{}}{} (\valuation)
            &= {\solutionFctMdpObjective{{\alphabetExtensionOfTo{\ppa_2}{\alphabetOf{}}}, {{\stratProjOfToValuation{\strategy}{2}{\valuation_+, \valuation}}}}{}}{}(\valuation)
              \tag{by \Cref{theo:lemma3NewDependent}} 
            \\
            & \leq  {\solutionFctMdpObjective{{\alphabetExtensionOfTo{\ppa_2}{\alphabetOf{}}},{{\stratProjOfToValuation{\strategy}{2}{\valuation_+, \valuation}}}}{}}{}(\valuation_+)
             \tag{ ${\stratProjOfToValuation{\strategy}{2}{\valuation_+, \valuation}}\in  \strategysetOf{\alphabetExtensionOfTo{\ppa_2}{\alphabetOf{}}}{\prt}$,  $\monotonicOnRegionParameter{\buparrow}{\solutionFctMdpObjective{\alphabetExtensionOfTo{\ppa_2}{\alphabetOf{}}}{}}{p}{\region_2}{\prt}$} \\
            & = {\solutionFctMdpObjective{{\alphabetExtensionOfTo{\ppa_2}{\alphabetOf{}}}, {{\stratProjOfToValuation{\strategy}{2}{\valuation_+, \valuation_+}}}}{}}{}(\valuation_+) 
         \tag{by \Cref{theo:change_valuation_projection} }\\
        &= {\solutionFctMdpObjective{{(\alphabetExtensionOfTo{\ppa_1}{\alphabetOf{}} \parallel \alphabetExtensionOfTo{\ppa_2}{\alphabetOf{}})},{\strategy}}{}}{}(\valuation_+)  \tag{by \Cref{theo:lemma3NewDependent}}
    \end{align*}
        Thus, $ {\solutionFctMdpObjective{{(\alphabetExtensionOfTo{\ppa_1}{\alphabetOf{}} \parallel \alphabetExtensionOfTo{\ppa_2}{\alphabetOf{}})},{\strategy}}{}}{}(\valuation) \leq {\solutionFctMdpObjective{{(\alphabetExtensionOfTo{\ppa_1}{\alphabetOf{}} \parallel \alphabetExtensionOfTo{\ppa_2}{\alphabetOf{}})},{\strategy}}{}}{}(\valuation_+) $, violating \Cref{proof:contradiction}. 
        \end{proof}

\begin{example}\label{ex:mono_compositional}
    Reconsider the pPA $\ppa_1\parallel\ppa_2$ in \Cref{fig:pPA_composition}, the region $\region =\{\valuation \colon \{p,q\} \to [0,1]\}$, and the language $\regLang = \{ w \in \{a,c,\frownie\}^{\infty} \mid \vert w \vert_{\frownie} = 0 \}$ from \Cref{ex:solfct}. 
    The pPA $\ppa_1\parallel\ppa_2$ is composed of the pPAs $\ppa_1$ and $\ppa_2$ shown in \Cref{fig:ppa1_and_ppa2}. 
    The region $\region$ is well-defined for $\ppa_1$ and $\ppa_2$. 
    We check whether $\solutionFctMdpObjective{\ppa_1 \parallel \ppa_2}{\PrOf{}{}{\regLang}}$ is monotonic in $q$ on $\region$ via \Cref{theo:pag_mono_rule}. 
    Since the premises $\monotonicOnRegionParameter{\bdownarrow}{\solutionFctMdpObjective{\alphabetExtensionOfTo{\ppa_i}{\alphabetOf{\regLang}}}{\PrOf{}{}{\regLang}}}{q}{\region}{\prt}$ for $i \in \{1,2\}$ are satisfied, we conclude $\monotonicOnRegionParameter{\bdownarrow}{\solutionFctMdpObjective{\ppa_1 \parallel \ppa_2}{\PrOf{}{}{\regLang}}}{q}{\region}{\prt}$. 
\end{example}

\section{Related Work}\label{sec:related_work}
Compositional verification has been widely studied in both probabilistic and non-probabilistic systems. 
We summarize key contributions related to our work. 

% AG reasoning 
Jones' rely-guarantee method \cite{Jon83} and Pnueli's compositional proof system \cite{Pnu84} for temporal logic laid the foundation for AG reasoning. 
Subsequent work focused on AG rules for systems using CTL$^*$ \cite{Cla+89} and developed AG reasoning for \emph{reactive modules} \cite{Hen+98,AH+99}. 
Automated AG reasoning techniques---developed by Pasareanu et al.\ \cite{Cob+03, Pas+08}---include learning-based assumption generation. 
More recent work has focused on circular AG reasoning \cite{Elk+18} and verification-repair approaches \cite{Fre+22}. 

AG reasoning has been lifted to probabilistic settings. 
Initial work by de Alfaro et al.\ \cite{Alf+01} introduced AG rules for a probabilistic extension of reactive modules \cite{Hen+98,AH+99}. 
Their model is similar to PA \cite{Seg+95, Seg95}, but limited to synchronous composition. 

Kwiatkowska et al.\ \cite{Kwi+10,For+11} generalized AG verification for PA, allowing more flexible parallel compositions and extending AG reasoning to probabilistic safety properties.   
Their approach reduces AG verification to multi-objective model checking, as proposed by Etessami et al.\ \cite{Ete+08}. 
This was further refined in \cite{Kwi+13}, enabling AG reasoning over a broader class of quantitative properties, including conjunctions over probabilistic liveness and expected rewards. 
Algorithmic learning-based assumption generation techniques \cite{Cob+03, Gup+08} were later adapted for AG reasoning in probabilistic settings \cite{Fen+10, Fen+11, LL19}. 
Other assumption generation approaches include abstraction-refinement methods \cite{Kom+12, Chat+15}, based on the CEGAR paradigm \cite{Cla+00}, and symbolic learning-based methods \cite{He+16, BB18}. 
AG reasoning has been applied to various real-world domains, including service-based workflow verification \cite{Bou+16}, large-scale IT systems \cite{Cal+12}, and autonomous systems with deep neural networks \cite{Pas+18,Pas+23}.

AG reasoning has also been extended to systems with uncertainty, for example, \cite{Zha+16} introduced an AG framework for verifying systems with components modeled by MDPs and \emph{partially observable MDPs} (POMDPs). 
In contrast, our work considers a different type of uncertainty; We extend AG reasoning to parametric probabilistic automata (PA), leveraging research on parametric MDPs \cite{Jun20, Qua+16, Jun+24} and previous AG verification techniques \cite{Kwi+13}. 
Our framework allows to reason about monotonicity \cite{Spel23, Spe+19, Spe+21} in a compositional manner. 
To the best of our knowledge, this the first AG-based compositional verification framework for parametric PA. 
Existing modular proof systems have focused on parametric timed automata \cite{Lac+16} or non-probabilistic parameterized systems \cite{BR+06, SH13, NT16}, where concurrent programs are parameterized by the number of processes or threads in a configured instance. 

Another recent line of research focuses on the \emph{sequential} composition of MDPs rather than parallel decomposition: 
Junges and Spaan \cite{JS22} introduced an abstraction-refinement approach for hierarchical probabilistic models, leveraging parametric MDPs to represent sets of similar subroutines. 
Recent work by Watanabe et al.\ \cite{Wat+24} on mean-payoff games, applies category-theoretic string diagrams to the verification of sequentially composed MDPs.

\section{Conclusion}\label{sec:conclusion}
We presented an assume-guarantee framework for compositional verification of parametric probabilistic automata, building on the proof rules for Segala's PA by Kwiatkowska et al.\ \cite{Kwi+13}. 
%In addition to adapting classical AG rules to the parametric setting, 
In addition, we introduced new compositional proof rules to reason about monotonicity in composed systems. 
These contributions lay the theoretical foundations for modular verification of pPA. % and monotonicity. 
To the best of our knowledge, these are the first AG proof rules for probabilistic models with parametric transition probabilities. 

Future work involves implementing the framework and demonstrating its effectiveness through case studies. 
Another direction is to deduce additional assume-guarantee rules---for example, reasoning about robust valuations or strategies, i.e, properties of the form: %\hm{PLA hard?} 
$\exists \valuation \in \region: \forall \strategy \in \strategysetOf{\ppa}{\star}: \ppa[\valuation], {\strategy} \modelsWrt{} \multiobjectiveQuery{}{X}$. 
Additionally, interesting directions include the modular verification of other properties, such as long-run average rewards or expected visiting times~\cite{Mer+24}. 
Other areas include extending verification to logics such as parametric LTL~\cite{CK14} and probabilistic CTL. 
Further research could also explore Markov automata with parameters, building on preliminary work in modular reasoning for continuous-time and continuous-space models~\cite{Car+11}. 
Another interesting direction is adapting assume-guarantee reasoning for stochastic games~\cite{Wil15} to a parametric setting.

%%
%% Bibliography
%%

%% Please use bibtex, 

\bibliography{biblio}

\iftoggle{extended}{

\appendix
\section{Omitted Proofs}
\label{app:proofs}
\subsection{Proofs of \Cref{sec:strat_projections}}
\label{app:proofs_of_strat_projections}

\projectionAlternative*
\begin{proof}
Observe that the set $\bigcup_{\ppath \in (\liftedpaths{(\ppath_i,\alpha_i,s_i)}{\pa_{3-i}})} \cyl(\ppath)$ can be partitioned into disjoint subsets $\cyl\big(\ppath, (\hat\alpha_1,\hat\alpha_2), (s_1,s_2)\big)$ with $\ppath \in (\liftedpaths{\ppath_i}{\pa_{3-i}})$, $\hat\alpha_i = \alpha_i$, and $s_{3-i} \in \stateSetOf{3-i}$.
This yields
\begin{align*}
&~\sum_{s_i \in \stateSetOf{i}} \PrOf{\pa}{\strategy}{\liftedpaths{(\ppath_i,\alpha_i,s_i)}{\pa_{3-i}}}\\
=&~ \sum_{s_i \in \stateSetOf{i}} \bigg(\sum_{\ppath \in (\liftedpaths{\ppath_i}{\pa_{3-i}})} ~\sum_{\substack{(\hat\alpha_1,\hat\alpha_2) \in \actSetOf{\parallel},\\\hat\alpha_i = \alpha_i}}~ \sum_{\substack{s_{3-i} \in \stateSetOf{3-i}}}  \PrOf{\pa}{\strategy}{\ppath,(\hat\alpha_1,\hat\alpha_2),(s_1,s_2)}\bigg)\\
=&~ \sum_{\ppath \in (\liftedpaths{\ppath_i}{\pa_{3-i}})} ~\sum_{\substack{(\hat\alpha_1,\hat\alpha_2) \in \actSetOf{\parallel},\\\hat\alpha_i = \alpha_i}} \PrOf{\pa}{\strategy}{\ppath} \cdot  \strategy(\ppath,(\hat\alpha_1,\hat\alpha_2)) \cdot \overbrace{\sum_{\substack{s_1 \in \stateSetOf{1}\\s_2 \in \stateSetOf{2}}} 
	  \transFctOf{\parallel}\big(\last{\ppath}, (\hat\alpha_1, \hat\alpha_2), (s_1, s_2)\big)}^{=1}.
\end{align*}
\end{proof}

\projectionMeasure*
\begin{proof}
       We show the statement by induction over the length of a path $\ppath_i \in \finPathsOf{\pa_i}$. 
        %IA $n=1$: 
        For $\vert \ppath_i \vert =0$, i.e., $\ppath_i =\initialOf{i} \in \finPathsOf{\pa_i}$, we have $\PrOf{\pa_i}{{\stratProjOfToValuation{\strategy}{i}{}}}{\initialOf{i}} = 1 = \PrOf{\pa}{\strategy}{\liftedpaths{\initialOf{i}}{\pa_{3-i} }}$.  
        %By \Cref{def:projectionDependentStrategy}: 
For the induction step, assume that the statement holds for $\ppath_i'\in \finPathsOf{\pa_i}$ and consider $\ppath_i =  \ppath_i', \alpha_i, s_i \in \finPathsOf{\pa_i}$.
If $\PrOf{\pa}{\strategy}{\liftedpaths{\ppath'_i}{\pa_{3-i}}} = 0$, the induction hypothesis implies $\PrOf{\pa_i}{\stratProjOfToValuation{\strategy}{i}{}}{\ppath'_i} = 0$, yielding $\PrOf{\pa_i}{\stratProjOfToValuation{\strategy}{i}{}}{\ppath_i} = 0=\PrOf{\pa}{\strategy}{\liftedpaths{\ppath_i}{\pa_{3-i}}}$.
Otherwise, we have $\PrOf{\pa}{\strategy}{\liftedpaths{\ppath'_i}{\pa_{3-i}}} > 0$ and get:
        \begin{align*}
         &~   \PrOf{\pa_i}{{\stratProjOfToValuation{\strategy}{i}{}}}{\ppath_i}\\
       ~= &~  \PrOf{\pa_i}{{\stratProjOfToValuation{\strategy}{i}{}}}{\ppath_i'} \cdot \stratProjOfToValuation{\strategy}{i}{}(\ppath_i' , \alpha_i) \cdot \transFctOf{i}(\last{\ppath_i'}, \alpha_i, s_i)\\
       = &~  \PrOf{\pa}{\strategy}{\liftedpaths{\ppath'_i}{\pa_{3-i}}} \cdot \stratProjOfToValuation{\strategy}{i}{}(\ppath_i' , \alpha_i) \cdot \transFctOf{i}(\last{\ppath_i'}, \alpha_i, s_i) \tag{By induction hypothesis}\\
       = &~  \sum_{\hat{s}_i \in \stateSetOf{i}} \PrOf{\pa}{\strategy}{\liftedpaths{(\ppath'_i,\alpha_i,\hat{s}_i)}{\pa_{3-i}}}  \cdot \transFctOf{i}(\last{\ppath_i'}, \alpha_i, s_i) \tag{By \Cref{def:projectionStrategyPA}, simplify}\\
       = &  \sum_{\ppath' \in (\liftedpaths{\ppath'_i}{\pa_{3-i}})}\, \sum_{\substack{(\hat\alpha_1,\hat\alpha_2) \in \actSetOf{\parallel},\\\hat\alpha_i = \alpha_i}}  \PrOf{\pa}{\strategy}{\ppath'} \cdot \strategy(\ppath',(\hat\alpha_1,\hat\alpha_2)) \cdot \transFctOf{i}(\last{\ppath_i'}, \alpha_i, s_i) \tag{By \Cref{lem:projection:alternative}}\\
      = &~ \PrOf{\pa}{\strategy}{\liftedpaths{(\ppath'_i,\alpha_i,s_i)}{\pa_{3-i}}}
	 = \PrOf{\pa}{\strategy}{\liftedpaths{\ppath}{\pa_{3-i}}}.
       \end{align*}
\end{proof}

\changeValuationProjection*
\begin{proof}
We show $\stratProjOfToValuation{\strategy}{1}{\valuation_1, \valuation_2} = \stratProjOfToValuation{\strategy}{1}{\valuation_1', \valuation_2}$. 
The proof for $\stratProjOfToValuation{\strategy}{2}{\valuation_1, \valuation_2} = \stratProjOfToValuation{\strategy}{2}{\valuation_1, \valuation_2'}$ is symmetric.
For simplicity, we write $\ppa[\valuation_1,\valuation_2] = \ppa_1[\valuation_1] \parallel \ppa_2[\valuation_2]$.
Let $\ppath_1 \in \finPathsOf{\ppa_1}$ and $\alpha_1 \in \actSetOf{1}$.
Our restrictions for $\valuation_1$ and $\valuation_1'$ yield $\PrOf{\ppa[\valuation_1,\valuation_2]}{\strategy}{\liftedpaths{\ppath_1}{\ppa_2}} = 0$ iff $\PrOf{\ppa[\valuation_1',\valuation_2]}{\strategy}{\liftedpaths{\ppath_1}{\ppa_2}} = 0$.
Consequently, $\PrOf{\ppa[\valuation_1,\valuation_2]}{\strategy}{\liftedpaths{\ppath_1}{\ppa_2}} = 0$ implies $\stratProjOfToValuation{\strategy}{1}{\valuation_1, \valuation_2}(\ppath_1, \alpha_1) = 0 = \stratProjOfToValuation{\strategy}{1}{\valuation_1', \valuation_2}(\ppath_1, \alpha_1)$.
Otherwise, for $\hat\valuation_1 \in \set{\valuation_1,\valuation_2}$ we have by \Cref{def:projectionStrategyPA}: 
\[
\stratProjOfToValuation{\strategy}{1}{\hat\valuation_1, \valuation_2}(\ppath_1, \alpha_1) =
\frac{\sum_{s_1 \in \stateSetOf{1}}	\PrOf{\ppa[\hat\valuation_1,\valuation_2]}{\strategy}{\liftedpaths{(\ppath_1,\alpha_1,s_1)}{\ppa_{2}}}}{\PrOf{\ppa[\hat\valuation_1,\valuation_2]}{\strategy}{\liftedpaths{\ppath_1}{\ppa_{2}}}}.
\]
We show that all transition probabilities $\transFctOf{1}(s_1,\alpha_1,s_1')[\hat\valuation_1]$ cancel out in the above fraction, i.e., $\stratProjOfToValuation{\strategy}{1}{\hat\valuation_1, \valuation_2}(\ppath_1, \alpha_1)$ is independent of the valuation $\hat\valuation_1 \in \set{\valuation_1,\valuation_2}$, which concludes the proof.

We first introduce some auxiliary notations.
For $\ppath \in \finPathsOf{\ppa}$, let $\pref{\ppath, m}  \in \finPathsOf{\ppa}$ be the prefix of $\ppath$ of length $m\le |\ppath|$. The set of minimal paths $\ppath$ with projection $\restrOfTo{\ppath}{1} = \ppath_1$ is
\[\min\set{\liftedpaths{\ppath_1}{\ppa_2}} = \bigset{\ppath \in (\liftedpaths{\ppath_1}{\ppa_2}) \mathbin{\big|} \pref{\ppath, m} \notin (\liftedpaths{\ppath_1}{\ppa_2}) \text{ for all } m<|\ppath| }.\]
Any path $\ppath = (s_1^0,s_2^0),(\alpha_1^0,\alpha_2^0),(s_1^1,s_2^1), \dots, (s_1^n,s_2^n) \in (\liftedpaths{\ppath_1}{\ppa_2})$ has a prefix $\pref{\ppath,m} \in \min\set{\liftedpaths{\ppath_1}{\ppa_2}}$ and the subsequent transitions all correspond to asynchronous steps of $\ppa_2$, i.e., $s_1^m = s_1^{m+1} = \dots = s_1^n$ and $\alpha_1^m, \alpha_1^{m+1}, \dots, \alpha_1^{n-1} \notin \actSetOf{1}$.
Consequently, if $\ppath$ is initial:
\begin{align*}
&~\PrOf{\ppa[\hat\valuation_1,\valuation_2]}{\strategy}{\ppath}\\
=&~ 
\prod_{\substack{0\le j<m\\\alpha_1^j \in \actSetOf{1}}}\! \transFctOf{1}(s_1^j,\alpha_1^j,s_1^{j+1})[\hat\valuation_1] ~\cdot
\!\prod_{\substack{0\le j<n\\\alpha_2^j \in \actSetOf{2}}}\! \transFctOf{2}(s_2^j,\alpha_2^j,s_2^{j+1})[\valuation_2] ~\cdot 
\!\prod_{\substack{0\le j<n}}\! \strategy(\pref{\ppath,j}, (\alpha_1^j,\alpha_2^j))\\
=&~
\transFctOf{1}(\ppath_1)[\hat\valuation_1] ~\cdot \!\prod_{\substack{0\le j<n\\\alpha_2^j \in \actSetOf{2}}}\! \transFctOf{2}(s_2^j,\alpha_2^j,s_2^{j+1})[\valuation_2] ~\cdot 
\!\prod_{\substack{0\le j<n}}\! \strategy(\pref{\ppath,j}, (\alpha_1^j,\alpha_2^j)),
\end{align*}
where we set 
$\transFctOf{1}(\ppath_1)[\hat\valuation_1] = \prod_{\substack{0\le j<k}}\! \transFctOf{1}(\hat{s}_1^j,\hat\alpha_1^j,\hat{s}_1^{j+1})[\hat\valuation_1] $ for $\ppath_1 = \hat{s}_1^0,\hat\alpha_1^0,\hat{s}_1^1,\dots,\hat{s}_1^k$.
This yields
\[
\frac{\PrOf{\ppa[\valuation_1,\valuation_2]}{\strategy}{\ppath}}{\transFctOf{1}(\ppath_1)[\valuation_1]}
 = \!\prod_{\substack{0\le j<n\\\alpha_2^j \in \actSetOf{2}}}\! \transFctOf{2}(s_2^j,\alpha_2^j,s_2^{j+1})[\valuation_2] ~\cdot 
 \!\prod_{\substack{0\le j<n}}\! \strategy(\pref{\ppath,j}, (\alpha_1^j,\alpha_2^j))
= \frac{\PrOf{\ppa[\valuation'_1,\valuation_2]}{\strategy}{\ppath}}{\transFctOf{1}(\ppath_1)[\valuation'_1]}.
\]
Finally, we combine the above observation with \Cref{lem:projection:alternative} to conclude
\begin{align*}
    \stratProjOfToValuation{\strategy}{1}{\valuation_1, \valuation_2}(\ppath_1, \alpha_1) 
    =&  \frac{\transFctOf{1}(\ppath_1)[\valuation_1]}{\transFctOf{1}(\ppath_1)[\valuation_1]} \cdot  \frac{\displaystyle\sum_{\ppath \in (\liftedpaths{\ppath_1}{\ppa_{2}})}~\sum_{\substack{(\alpha_1,\alpha_2) \in \actSetOf{\parallel}}} \PrOf{\ppa[\valuation_1,\valuation_2]}{\strategy}{\ppath} \cdot \strategy(\ppath,(\alpha_1,\alpha_2))}{\displaystyle\sum_{\ppath \in \min\set{\liftedpaths{\ppath_1}{\ppa_{2}}}} \PrOf{\ppa[\valuation_1,\valuation_2]}{\strategy}{\ppath}}  \\[12pt]
    =&~ \frac{\displaystyle\sum_{\ppath \in (\liftedpaths{\ppath_1}{\ppa_{2}})}~\sum_{\substack{(\alpha_1,\alpha_2) \in \actSetOf{\parallel}}} \dfrac{\PrOf{\ppa[\valuation_1,\valuation_2]}{\strategy}{\ppath}}{\transFctOf{1}(\ppath_1)[\valuation_1]} \cdot \strategy(\ppath,(\alpha_1,\alpha_2))}{\displaystyle\sum_{\ppath \in \min\set{\liftedpaths{\ppath_1}{\ppa_{2}}}} \dfrac{\PrOf{\ppa[\valuation_1,\valuation_2]}{\strategy}{\ppath}}{\transFctOf{1}(\ppath_1)[\valuation_1]}}
    =    \stratProjOfToValuation{\strategy}{1}{\valuation'_1, \valuation_2}(\ppath_1, \alpha_1).
\end{align*}

\end{proof}
 % of \input{0_sections/projections_ppa} 

\subsection{Proofs of \Cref{sec:verification}}
\label{app:proofs_of_verification}
We show that the verification of safety objectives reduces to maximal reachability properties in a pPA-DFA Product based on \cite[Lemma 1]{Kwi+13}. 
First, we lift the DFA product (see \cite[Definition 12]{Kwi+13}) to pPA. 
\begin{definition}[pPA-DFA Product]\label{def:pPA_dfa_product}
    Given a bad prefix automaton $\bpAutomatonOf{\regLang}$ with $\alphabetOf{\bpAutomatonOf{{\regLang}}} \subseteq  \alphabetOf{\ppa}$, 
    the pPA $\ppa_{\product} = (\ppa \otimes \bpAutomatonOf{\regLang})$ is defined as $\ppa_{\product} = \ppaTupleOf{\product}$, 
    where 
    \begin{itemize}
        \item $\stateSetOf{\product} = \stateSetOf{\ppa} \times \stateSetOf{\bpAutomatonOf{{\regLang}}}$, 
        \item $\initialOf{\product} = (\initialOf{\ppa}, \initialOf{\bpAutomatonOf{\regLang}})$, 
        \item $\parameterSetOf{\product} = \parameterSetOf{\ppa}$,
        \item $\actSetOf{\product} = \actSetOf{\ppa}$,
        \item for each $(s, \alpha) \in \domain(\transFctOf{\ppa})$, with $\syncOf{\ppa}(s, \alpha) = \lab \in \alphabetOf{\bpAutomatonOf{{\regLang}}}$ and $p \in \stateSetOf{\bpAutomatonOf{\regLang}}$ with $q = \transFctOf{\bpAutomatonOf{\regLang}}(p, \lab)$: 
            \[\transFctOf{\product}((s,q), \alpha) = \transFctOf{\ppa}(s, \alpha) \times \indicatorFct{}
                \quad \text{and} \quad \syncOf{\product}((s,q), \alpha) = \lab
            \]
            for each $(s, \alpha) \in \domain(\transFctOf{\ppa})$, with $\syncOf{\ppa}(s, \alpha) = \lab \not \in \alphabetOf{\bpAutomatonOf{{\regLang}}}$ and $p \in \stateSetOf{\bpAutomatonOf{\regLang}}$: 
            \[\transFctOf{\product}((s,q), \alpha) = \transFctOf{\ppa}(s, \alpha) \times \indicatorFct{p}
                \quad \text{and} \quad \syncOf{\product}((s,q), \alpha) = \lab
            \]
\end{itemize}
\end{definition}
It holds that $(\ppa {\otimes} \bpAutomatonOf{\regLang})[\valuation] = {\ppa[\valuation] \otimes\bpAutomatonOf{\regLang} }$. 
Using this, we can lift \cite[Lemma 1]{Kwi+13} to pPAs. 
\begin{lemma}%[Safety via DFA-Product Reachability]
    \label{theo:prob_safety_pPA_prop_reach_poduct_bijection}
    Let $\ppa$ be a pPA and let $\regLang$ be the language of a safety objective for ${\ppa}$ and let $\valuation$ be a well-defined valuation.  
    Let $\star \in \{\prt,\comp\}$. 
    \begin{itemize}	%TODO fixed val first maybe change for multiobjective MC
        \item There is a bijection $f$ between the strategies
        $\strategysetOf{\ppa[\valuation]}{(\mless,)\star}$ of $\ppa$ and the strategies $\strategysetOf{\product[\valuation]}{(\mless,)\star}$ of pPA-DFA product $\product = \ppa {\otimes} \bpAutomatonOf{\regLang}$
        \item Let $bad_{\regLang} = \{ (s, q) \in \stateSetOf{\product} \mid q \in F \}$ = set of product-states that contain a final state of $\bpAutomatonOf{\regLang}$.  			
        $\strategy \in \strategysetOf{\ppa[\valuation]}{\star}$: 
        \[
            \PrOf{\ppa}{\valuation,\strategy}{\regLang} =  1-\PrOf{\product}{\valuation, f(\strategy)}{\Diamond bad_{\regLang}}
        \]
    \end{itemize}
\end{lemma}	
\begin{proof}
    As stated in the proof of \cite[Lemma 1]{Kwi+13}, the statement follows by \cite[Theorem 10.51]{BK08}. 
    \cite[Theorem 10.51]{BK08} can be applied as safety objectives are prefix closed, see \Cref{sec:explain_differences} for a detailed explanation. 
\end{proof}
%
%todo mless: No, otherwise not biijection??? doch???}	
\begin{remark}
    For actions that do not appear in $\regLang$ (and thus not appear in $\bpAutomatonOf{\regLang}$) the state of the automaton remains unchanged.  
    %For 
    This means paths $\pi$ of $\ppa$ such that $\restrOfTo{\pi}{\regLang}$ is finite are cannot reach the $bad_{\regLang}$-labeled state in $(\ppa \product \bpAutomatonOf{\regLang})$. 
    Thus, for correctness of \Cref{theo:prob_safety_pPA_prop_reach_poduct_bijection}, safety objectives need to be prefix closed. 
\end{remark}

\restatableSafetyPartVsComplete*
\begin{proof}
    As $\region$ is well-defined, \Cref{theo:prob_safety_pPA_prop_reach_poduct_bijection} applies and 
    \Cref{lemma_safety_partial_vs_complete} follows analogously to the argument presented in \cite[Proposition 1]{Kwi+13}. 
\end{proof}

\restatableLemmaThreeNewPre*
\restatableLemmaThreeNew*
We give the proof for \Cref{theo:lemma3NewDependentPre} and \Cref{theo:lemma3NewDependent}: 
\begin{proof}   
    $\valuation$ is well-defined for $\ppa_1$ and $\ppa_2$. Thus, $\ppa_1[\valuation]$,  $\ppa_2[\valuation]$, and $(\ppa_1 \parallel \ppa_2)[\valuation] = (\ppa_1[\valuation]) \parallel (\ppa_2[\valuation])$ are PAs. 
    Then, the claims follow by \cite[Lemma 3]{Kwi+13}. 
\end{proof}

 %of \input{0_sections/preliminaries_objectives} and \input{0_sections/partial_vs_complete}

\subsection{Proofs of \Cref{sec:pag}}

\label{app:proofs_of_PAG}
\restatableAsymRule* 
\begin{proof}
    The proof is based on the proof provided for \cite[Theorem 1, Theorem 2]{Kwi+13}. 
    Since $\alphabetOf{\multiobjectiveQuery{(\safe)}{A}} \subseteq \alphabetOf{\ppa_1}$ and $\alphabetOf{\multiobjectiveQuery{(\safe)}{G}} \subseteq \alphabetOf{ \multiobjectiveQuery{(\safe)}{A}}\cup \alphabetOf{ \pa_2}$ hold by assumption, it follows that 
    \begin{align*}
        \alphabetOf{\multiobjectiveQuery{(\safe)}{G}}  
        & \subseteq \alphabetOf{ \multiobjectiveQuery{(\safe)}{A}} \cup \alphabetOf{ \ppa_2} \\
        & \subseteq \alphabetOf{\ppa_1} \cup  \alphabetOf{\ppa_2}  \\
        & = \alphabetOf{\ppa_1 \parallel \ppa_2}. 
    \end{align*} 
    Thus, ${\multiobjectiveQuery{(\safe)}{G}} $ is a valid multiobjective query for $\ppa_1 \parallel \ppa_2$. 
    If $\regionIntersectionOf{\region_{1}}{}{\region_2} = \emptyset$, the conclusion trivially holds. 
    We assume $\regionIntersectionOf{\region_{1}}{}{\region_2} \not = \emptyset$:  
    We show ${\ppa_1}\parallel \alphabetExtensionOfTo{\ppa_2}{\alphabetOf{\multiobjectiveQuery{(\safe)}{A}}}, \regionIntersectionOf{\region_{1}}{}{\region_2} \modelsWrt{\star} \multiobjectiveQuery{(\safe)}{G}$, 
    which directly implies that $\ppa_1 \parallel \ppa_2, \regionIntersectionOf{\region_{1}}{}{\region_2} \modelsWrt{\star} \multiobjectiveQuery{(\safe)}{G}$. %\hm{ok trivial?}
     \begin{itemize}
         \item First, we prove correctness of the rule involving safety properties: 
        Assuming the premises hold, it follows that 
         \begin{itemize}
             \item $\ppa_1, \regionIntersectionOf{\region_{1}}{}{\region_2} \modelsWrt{\comp} \multiobjectiveQuery{\safe}{A}$, and 
             \item $\agTriple{\alphabetExtensionOfTo{\ppa_2}{\alphabetOf{\multiobjectiveQuery{\safe}{A}}}, \regionIntersectionOf{\region_{1}}{}{\region_2}  }{\prt}{\multiobjectiveQuery{\safe}{A}}{\multiobjectiveQuery{\safe}{G}}$. 
         \end{itemize}
         Let $\valuation \in \regionIntersectionOf{\region_{1}}{}{\region_2}$. 
         From $\ppa_1, \regionIntersectionOf{\region_{1}}{}{\region_2} \modelsWrt{\comp} \multiobjectiveQuery{\safe}{A}$ follows that $\ppa_1[\valuation] \modelsWrt{\comp} \multiobjectiveQuery{\safe}{A}$. 
         
         Let $\strategy \in \strategysetOf{( \ppa_1 \parallel \alphabetExtensionOfTo{\ppa_2}{\alphabetOf{\multiobjectiveQuery{(\safe)}{A}}})[\valuation]}{\comp}$. 
         The strategy $\stratProjOfToValuation{\strategy}{1}{\valuation}$ is a (partial) strategy of $\ppa_1[\valuation]$ (see \Cref{sec:strat_projections}). 
         By \Cref{lemma_safety_partial_vs_complete}: % complete implies partial
         \begin{align*}
             \ppa_1[\valuation] & \modelsWrt{\comp} \multiobjectiveQuery{\safe}{A} \\
             & \Rightarrow \ppa_1[\valuation], {\stratProjOfToValuation{\strategy}{1}{\valuation}} \modelsWrt{\prt} \multiobjectiveQuery{\safe}{A} \\
             % apply 3 e =>
             &  \Rightarrow \bigl( \ppa_1 \parallel \alphabetExtensionOfTo{\ppa_2}{\alphabetOf{\multiobjectiveQuery{(\safe)}{A}}} \bigr)[\valuation], {\strategy} \modelsWrt{} \multiobjectiveQuery{\safe}{A} 
             && \text{By \Cref{lemma_3_dDependent} } \\
             % apply 3 f =>
             & \Rightarrow { \alphabetExtensionOfTo{\ppa_2}{\alphabetOf{\multiobjectiveQuery{\safe}{A}} } }[\valuation],{ \stratProjOfToValuation{\strategy}{2}{\valuation} } \modelsWrt{}  \multiobjectiveQuery{\safe}{A}    
             && \text{By \Cref{lemma_3_fDependent}}\\
             % apply triple =>
             & \Rightarrow  \bigl(\alphabetExtensionOfTo{\ppa_2}{\alphabetOf{\multiobjectiveQuery{\safe}{A}} }[\valuation]\bigr), {\stratProjOfToValuation{\strategy}{2}{\valuation}}   \modelsWrt{}  \multiobjectiveQuery{\safe}{G} 
             && \text{$ \stratProjOfToValuation{\strategy}{2}{\valuation} \in \strategysetOf{\alphabetExtensionOfTo{\ppa_2}{\alphabetOf{\multiobjectiveQuery{\safe}{A}} }}{\prt}$} \\
             & && \text{and $\agTriple{\alphabetExtensionOfTo{\ppa_2}{\alphabetOf{\multiobjectiveQuery{\safe}{A}}}, \regionIntersectionOf{\region_{1}}{}{\region_2} }{\prt}{\multiobjectiveQuery{\safe}{A}}{\multiobjectiveQuery{\safe}{G}}$} \\
             % apply 3 f  =>
             & \Rightarrow {(\ppa_1 \parallel \alphabetExtensionOfTo{\ppa_2}{\alphabetOf{\multiobjectiveQuery{(\safe)}{A}}})[\valuation]}, {\strategy} \modelsWrt{} \multiobjectiveQuery{\safe}{G} 
             && \text{By \Cref{lemma_3_fDependent}}   
         \end{align*}  
         \item Next, we prove correctness for the rule involving general properties and fair strategies: 
        Again, the premises imply that 
        \begin{itemize}
            \item ${\ppa_1, \regionIntersectionOf{\region_1}{}{\region_2} \modelsWrt{\fairWrtRegionModel{}{\decomp_1}} \multiobjectiveQuery{}{A}}$
            \item $\agTriple{\alphabetExtensionOfTo{\ppa_2}{\alphabetOf{\multiobjectiveQuery{}{A}}} , \regionIntersectionOf{\region_{1}}{}{\region_2}}{\fairWrtRegionModel{}{\decomp_2}}{\multiobjectiveQuery{}{A}}{\multiobjectiveQuery{}{G}}$
        \end{itemize}
        Let $\valuation \in \regionIntersectionOf{\region_1}{}{\region_2}$. 
        From ${\ppa_1, \regionIntersectionOf{\region_{1}}{}{\region_2} \modelsWrt{\fairWrtRegionModel{}{\decomp_1}} \multiobjectiveQuery{}{A}}$ follows $\ppa_1[\valuation] \modelsWrt{\fairWrtRegionModel{}{\decomp_1}} \multiobjectiveQuery{}{A}[\valuation]$. 

         Let $\strategy \in \strategysetOf{(\ppa_1 \parallel \ppa_2)[\valuation]}{\fairWrtRegionModel{}{\decomp_1 \cup \decomp_2}}$. 
         The strategy $\stratProjOfToValuation{\strategy}{1}{\valuation}$ is a complete strategy of $\ppa_1[\valuation]$ that is $\fairWrtRegionModel{}{\decomp_1}$ by \Cref{theo:strategy_projection_partial_fair_preserved}. 
         Thus, 
         \begin{align*}
             \ppa_1[\valuation] &  \modelsWrt{\fairWrtRegionModel{}{\decomp_1}} \multiobjectiveQuery{}{A}[\valuation]   \\
             & \Rightarrow \ppa_1[\valuation], {\stratProjOfToValuation{\strategy}{1}{\valuation}} \modelsWrt{} {\multiobjectiveQuery{}{A}[\valuation]}  \\
             % apply 3 e =>
             &  \Rightarrow \left( \ppa_1 \parallel \alphabetExtensionOfTo{\ppa_2}{\alphabetOf{\multiobjectiveQuery{}{A}}}\right)[\valuation], {\strategy} \modelsWrt{}  {\multiobjectiveQuery{}{A}[\valuation]}
             && \text{By \Cref{lemma_3_eDependent} } \\
             % apply 3 f =>
             & \Rightarrow {\alphabetExtensionOfTo{\ppa_2}{\alphabetOf{\multiobjectiveQuery{}{A}} }}[\valuation], { \stratProjOfToValuation{\strategy}{2}{\valuation}} \modelsWrt{}  {\multiobjectiveQuery{}{A}}[\valuation]
             && \text{By \Cref{lemma_3_fDependent}}\\
             % apply triple =>
             & \Rightarrow  \alphabetExtensionOfTo{\ppa_2}{\alphabetOf{\multiobjectiveQuery{}{A}} }[\valuation], {\stratProjOfToValuation{\strategy}{2}{\valuation}}  \modelsWrt{}  {\multiobjectiveQuery{}{G}}[\valuation]
             && \text{$ \stratProjOfToValuation{\strategy}{2}{\valuation} \in \strategysetOf{{\alphabetExtensionOfTo{\ppa_2}{\alphabetOf{\multiobjectiveQuery{}{A}}}}[\valuation]}{\fairWrtRegionModel{}{\decomp_2}}$} \\
             & && \text{and $\agTriple{\alphabetExtensionOfTo{\ppa_2}{\alphabetOf{\multiobjectiveQuery{}{A}}},\regionIntersectionOf{\region_1}{}{\region_2}}{\fairWrtRegionModel{}{\decomp_2}}{\multiobjectiveQuery{}{A}}{\multiobjectiveQuery{}{G}}$} \\
             % apply 3 f  =>
             & \Rightarrow {\left(\ppa_1 \parallel \alphabetExtensionOfTo{\ppa_2}{\alphabetOf{\multiobjectiveQuery{}{A}}}\right)[\valuation]}, {\strategy} \modelsWrt{} \multiobjectiveQuery{}{G}[\valuation]
             && \text{By \Cref{lemma_3_fDependent}}   
         \end{align*}  

     \end{itemize}
 \end{proof}

\restatableCircRule*
\begin{proof}
    The proof is based on the proof of \cite[Theorem 5]{Kwi+13}. 
    We show the proof the rule on the right. 
    For the rule on the left the proof works analogously. 
    
    Since $\alphabetOf{\multiobjectiveQuery{}{A}_2} \subseteq \alphabetOf{\multiobjectiveQuery{}{A}_1} \cup \alphabetOf{\ppa_1} $ 
    and $\alphabetOf{\multiobjectiveQuery{}{G}} \subseteq  \alphabetOf{\multiobjectiveQuery{}{A}_2} \cup \alphabetOf{\ppa_2}$, it follows that $\alphabetOf{\multiobjectiveQuery{}{G}}  \subseteq \ppa_1 \cup \ppa_2$. 
    Thus, ${\multiobjectiveQuery{}{G}} $ is a valid multiobjective query for $\ppa_1 \parallel \ppa_2$. 

    If $\regionIntersectionOf{\region_{1}}{\region_2}{\region_3} = \emptyset$, the conclusion trivially holds. 
    We assume $\regionIntersectionOf{\region_{1}}{\region_2}{\region_3} \not = \emptyset$. 
    
    We show that $\alphabetExtensionOfTo{\ppa_1}{\alphabetOf{\multiobjectiveQuery{}{A}_1}}\parallel \alphabetExtensionOfTo{\ppa_2}{\alphabetOf{\multiobjectiveQuery{}{A}_2}}, \regionIntersectionOf{\region_{1}}{\region_2}{\region_3} \modelsWrt{\fairWrtRegionModel{}{\decomp_1 \cup \decomp_2 \cup \decomp_3}} \multiobjectiveQuery{}{G}$ holds,  
    which directly implies $\ppa_1 \parallel \ppa_2, \regionIntersectionOf{\region_{1}}{\region_2}{\region_3} \modelsWrt{\fairWrtRegionModel{}{\decomp_1 \cup \decomp_2 \cup \decomp_3}} \multiobjectiveQuery{}{G}$. %\hm{ok trivial?}
    
    From the premises, it follows that: 
    \begin{itemize}
        \item $\agTriple{\alphabetExtensionOfTo{\ppa_1}{\alphabetOf{\multiobjectiveQuery{}{A}_{1}}}, \regionIntersectionOf{\region_1}{\region_2}{\region_3} }{\fairWrtRegionModel{}{\decomp_1 }}{\multiobjectiveQuery{}{A}_{1}}{\multiobjectiveQuery{}{A}_2} $,
        \item $ \agTriple{\alphabetExtensionOfTo{\ppa_2}{\alphabetOf{\multiobjectiveQuery{}{A}_{2}}}, \regionIntersectionOf{\region_1}{\region_2}{\region_3} }{\fairWrtRegionModel{}{ \decomp_2 }}{\multiobjectiveQuery{}{A}_{2}}{\multiobjectiveQuery{}{G}} $, and
        \item $\ppa_2, \regionIntersectionOf{\region_1}{\region_2}{\region_3} \modelsWrt{\fairWrtRegionModel{}{\decomp_3}} \multiobjectiveQuery{}{A}_1$.
    \end{itemize}
    Let $\valuation \in \regionIntersectionOf{\region_1}{\region_2}{\region_3}$. 
    From $\ppa_2, \regionIntersectionOf{\region_1}{\region_2}{\region_3} \modelsWrt{\fairWrtRegionModel{}{\ppa_2}} \multiobjectiveQuery{}{A}_1$ follows 
    $\ppa_2[\valuation] \modelsWrt{\fairWrtRegionModel{}{\decomp_3}}  \multiobjectiveQuery{}{A}_1[\valuation]$. 
    Together with $\alphabetOf{\multiobjectiveQuery{}{A}_1} \subseteq \alphabetOf{\ppa_2}$, and $\decomp_3 \in \setOfdecompsOf{\alphabetOf{\ppa_2}}$, we obtain that: %\hm{is that trivial, do we need proposition?}
    \begin{align}\label{eq:circproof}
        \alphabetExtensionOfTo{\ppa_2}{\alphabetOf{\multiobjectiveQuery{}{A}_2}}[\valuation] \modelsWrt{\fairWrtRegionModel{}{\decomp_3}} \multiobjectiveQuery{}{A}_1[\valuation]
    \end{align} 
    Let $\strategy \in \strategysetOf{(\alphabetExtensionOfTo{\ppa_1}{\alphabetOf{\multiobjectiveQuery{}{A}_1}} \parallel \alphabetExtensionOfTo{\ppa_2}{\alphabetOf{\multiobjectiveQuery{}{A}_2}})[\valuation]}{\fairWrtRegionModel{}{\decomp_1 \cup \decomp_2 \cup \decomp_3}}$. 
    From \Cref{theo:strategy_projection_partial_fair_preserved} and $\decomp_3 \in \setOfdecompsOf{\alphabetOf{\ppa_2}} \subseteq \setOfdecompsOf{\alphabetOf{\ppa_2} \cup \alphabetOf{\multiobjectiveQuery{}{A}_2}}$, follows that the strategy $\stratProjOfToValuation{\strategy}{2}{\valuation}$ is a $\fairWrtRegionModel{}{\decomp_3}$ strategy of $\alphabetExtensionOfTo{\ppa_2}{\alphabetOf{\multiobjectiveQuery{}{A}_2}}[\valuation]$. 
    Thus, by \Cref{eq:circproof} we obtain:
    \begin{align*}
        \alphabetExtensionOfTo{\ppa_2}{\alphabetOf{\multiobjectiveQuery{}{A}_2}}[\valuation]&  \modelsWrt{\fairWrtRegionModel{}{\decomp_3}} \multiobjectiveQuery{}{A}_1[\valuation] \\
        % def of fair
        & \Rightarrow \alphabetExtensionOfTo{\ppa_2}{\alphabetOf{\multiobjectiveQuery{}{A}_2}}[\valuation], {\stratProjOfToValuation{\strategy}{2}{\valuation}} \modelsWrt{} {\multiobjectiveQuery{}{A}_1[\valuation]}\\
        %
        % apply 3 e =>
        &  \Rightarrow \left(\alphabetExtensionOfTo{\ppa_1}{\alphabetOf{\multiobjectiveQuery{}{A}_1}}\parallel {\alphabetExtensionOfTo{\ppa_2}{\alphabetOf{\multiobjectiveQuery{}{A}_2}}}\right)[\valuation], {\strategy} \modelsWrt{}  \multiobjectiveQuery{}{A}_1[\valuation]
        && \text{By \Cref{lemma_3_eDependent} } \\
        %
        % apply 3 f =>
        & \Rightarrow {\alphabetExtensionOfTo{\ppa_1}{\alphabetOf{\multiobjectiveQuery{}{A}_1} }[\valuation]}, { \stratProjOfToValuation{\strategy}{1}{\valuation}} \modelsWrt{}  \multiobjectiveQuery{}{A}_1[\valuation]
        && \text{By \Cref{lemma_3_fDependent}}
    \end{align*}
    We have that 
    $\stratProjOfToValuation{\strategy}{1}{\valuation}$ is a $\fairWrtRegionModel{}{\decomp_1}$ strategy of ${\alphabetExtensionOfTo{\ppa_1}{\alphabetOf{\multiobjectiveQuery{}{A}_1} }}[\valuation]$. 
    As a consequence, it holds that $\agTriple{\alphabetExtensionOfTo{\ppa_1}{\alphabetOf{\multiobjectiveQuery{}{A}_1}}, \region_1 \cap \region_2 \cap \region_3 }{\fairWrtRegionModel{}{\decomp_1}}{\multiobjectiveQuery{}{A}_1}{\multiobjectiveQuery{}{A}_2}$ from which follows that 
    \begin{align*}
        % apply triple =>
            \alphabetExtensionOfTo{\ppa_1}{\alphabetOf{\multiobjectiveQuery{}{A}_1}}[\valuation]&, { \stratProjOfToValuation{\strategy}{1}{\valuation}} \modelsWrt{}  {\multiobjectiveQuery{}{A}_2[\valuation]} \\
        %
        % apply 3 f  =>
        & \Rightarrow {\left( \alphabetExtensionOfTo{\ppa_1}{\alphabetOf{\multiobjectiveQuery{}{A}_1}}\parallel \alphabetExtensionOfTo{\ppa_2}{\alphabetOf{\multiobjectiveQuery{}{A}_2}}\right)[\valuation]}, {\strategy} \modelsWrt{} {\multiobjectiveQuery{}{A}_2[\valuation]}
        && \text{By \Cref{lemma_3_fDependent}}   \\
        % 
        % next loop
        %
        % apply 3 f  =>
        & \Rightarrow {\alphabetExtensionOfTo{\ppa_2}{\alphabetOf{\multiobjectiveQuery{}{A}_2} }[\valuation] }, {\stratProjOfToValuation{\strategy}{2}{\valuation} } \modelsWrt{}  \multiobjectiveQuery{}{A}_2[\valuation]
        && \text{By \Cref{lemma_3_fDependent}}\\
    \end{align*}  
    We apply the third premise and obtain that 
         $\agTriple{\alphabetExtensionOfTo{\ppa_2}{\alphabetOf{\multiobjectiveQuery{}{A}_2}}, \region_1 \cap \region_2 \cap \region_3 }{\fairWrtRegionModel{}{\decomp_2}}{\multiobjectiveQuery{}{A}_2}{\multiobjectiveQuery{}{G}}$. Then, as $\stratProjOfToValuation{\strategy}{2}{\valuation}$ is a $\fairWrtRegionModel{}{\decomp_2 }$ strategy of ${\alphabetExtensionOfTo{\ppa_2}{\alphabetOf{\multiobjectiveQuery{}{A}_2} }}[\valuation]$ it follows that:  
            \begin{align*}
            % apply triple =>
            \alphabetExtensionOfTo{\ppa_2}{\alphabetOf{\multiobjectiveQuery{}{A}_2} }[\valuation]&, { \stratProjOfToValuation{\strategy}{2}{\valuation}} \modelsWrt{}  {\multiobjectiveQuery{}{G}[\valuation]}  \\
            %
            % apply 3 f  =>
            &  \Rightarrow {{\left(\alphabetExtensionOfTo{\ppa_1}{\alphabetOf{\multiobjectiveQuery{}{A}_1}}\parallel \alphabetExtensionOfTo{\ppa_2}{\alphabetOf{\multiobjectiveQuery{}{A}_2}}\right)[\valuation]}}, {\strategy} \modelsWrt{} \multiobjectiveQuery{}{G}[\valuation]
            && \text{By \Cref{lemma_3_fDependent}}   
        \end{align*}

\end{proof}

 % of \input{0_sections/pag}
\subsection{Proofs of \Cref{sec:pag_mono}}
\label{app:proofs_of_pag_mono}

\restatablMonoRule*
\begin{proof}
We continue the proof from the main paper, which already established the left rule for partial strategies. For fair strategies the proof is similar. However, to deduce that 
        \begin{align*}
            {\solutionFctMdpObjective{{\alphabetExtensionOfTo{\ppa_i}{\alphabetOf{}}},{{\stratProjOfToValuation{\strategy}{i}{\valuation_+, \valuation}}}}{}}{}(\valuation) 
            & \leq  {\solutionFctMdpObjective{{\alphabetExtensionOfTo{\ppa_i}{\alphabetOf{}}},{{\stratProjOfToValuation{\strategy}{i}{\valuation_+, \valuation}}}}{}}{}(\valuation_+),
        \end{align*}
        we need to ensure that we can apply $\monotonicOnRegionParameter{\budarrow}{\solutionFctMdpObjective{\alphabetExtensionOfTo{\ppa_i}{\alphabetOf{}}}{}}{p}{\region_i}{\fairWrtRegionModel{}{\decomp_i}}$. 
        Thus, we need to show that 
        ${\stratProjOfToValuation{\strategy}{i}{\valuation_+, \valuation}} \in  \strategysetOf{\alphabetExtensionOfTo{\ppa_i}{\alphabetOf{}}[{\valuation_+}]}{\fairWrtRegionModel{}{\decomp_i}}$ and 
        ${\stratProjOfToValuation{\strategy}{i}{\valuation_+, \valuation}} \in  \strategysetOf{\alphabetExtensionOfTo{\ppa_i}{\alphabetOf{}}[{\valuation}]}{\fairWrtRegionModel{}{\decomp_i}}$.

        The strategy $\strategy$ of $\alphabetExtensionOfTo{\ppa_1}{\alphabetOf{}} \parallel \alphabetExtensionOfTo{\ppa_2}{\alphabetOf{}}$ is $\fairWrtRegionModel{}{\decomp_1 \cup \decomp_2}$ w.r.t.\ the valuations $\valuation, \text{ and }\valuation_+ \in \regionIntersectionOf{\region_1}{}{\region_2}$. 
        By \Cref{theo:strategy_projection_partial_fair_preserved} follows that the projection  
        ${\stratProjOfToValuation{\strategy}{1}{\valuation, \valuation}}$ is in $\strategysetOf{\alphabetExtensionOfTo{\ppa_1}{\alphabetOf{}}[{\valuation}]}{\fairWrtRegionModel{}{\decomp_1}}$  
        and ${\stratProjOfToValuation{\strategy}{2}{\valuation_+, \valuation_+}}$ is in $\strategysetOf{\alphabetExtensionOfTo{\ppa_2}{\alphabetOf{}}[{\valuation_+}]}{\fairWrtRegionModel{}{\decomp_2}}$. 
        Since $\valuation$ and $\valuation_+$ are graph-preserving, we can apply \Cref{theo:change_valuation_projection} and thus, 
        ${\stratProjOfToValuation{\strategy}{1}{\valuation_+, \valuation}} \in  \strategysetOf{\alphabetExtensionOfTo{\ppa_1}{\alphabetOf{}}[{\valuation}]}{\fairWrtRegionModel{}{\decomp_1}}$  
        and ${\stratProjOfToValuation{\strategy}{2}{\valuation_+, \valuation}} \in  \strategysetOf{\alphabetExtensionOfTo{\ppa_2}{\alphabetOf{}}[{\valuation_+}]}{\fairWrtRegionModel{}{\decomp_2}}$. 
        Since $\valuation$ and $\valuation_+$ are graph-preserving for $\alphabetExtensionOfTo{\ppa_i}{\alphabetOf{}}$, we have 
        $\strategysetOf{\alphabetExtensionOfTo{\ppa_i}{\alphabetOf{}}[{\valuation}]}{\fairWrtRegionModel{}{\decomp_i}} = \strategysetOf{\alphabetExtensionOfTo{\ppa_i}{\alphabetOf{}}[{\valuation_+}]}{\fairWrtRegionModel{}{\decomp_i}}$ 
        from which we deduce that  
        ${\stratProjOfToValuation{\strategy}{1}{\valuation_+, \valuation}} \in  \strategysetOf{\alphabetExtensionOfTo{\ppa_1}{\alphabetOf{}}[{\valuation_+}]}{\fairWrtRegionModel{}{\decomp_1}}$  
        and ${\stratProjOfToValuation{\strategy}{2}{\valuation_+, \valuation}} \in  \strategysetOf{\alphabetExtensionOfTo{\ppa_2}{\alphabetOf{}}[{\valuation}]}{\fairWrtRegionModel{}{\decomp_2}}$. 
\end{proof} % of \input{0_sections/monotonicity}

\section{Partial vs.\ Complete Strategies}\label{app:part_vs_complete}
\Cref{app:equivalence_partial_complete_measures} section establishes that verifying probabilistic, and reward objectives under partial strategies in a given pPA is equivalent to verifying them under complete strategies in an extended pPA. 
This result is extended to multiobjective queries in \Cref{app:equivalence_partial_complete_qmo}. 

\subsection{Equivalence w.r.t.\ Probabilistic and Reward Objectives}\label{app:equivalence_partial_complete_measures}
Next, we define the pPA $\ppa_{\tau}$ which extends the pPA $\ppa$ by introducing transitions from each state to a new sink state which are labeled by a fresh alphabet symbol $\tau$.  
We show that the class of partial strategies of a $\ppa$ is equivalent to the complete strategies of $\ppa_{\tau}$. 
\begin{definition}\label{def:ppaTau}
	Given a pPA $\ppa = \ppaTupleOf{\ppa}$, we define the pPA $\ppa_{\tau}= (\stateSetOf{\ppa} \cupdot \{s_{\tau}\}, \initialOf{\ppa},\parameterSetOf{\ppa}, \actSetOf{\ppa} \cupdot \{\tau\} , \transFctOf{\ppa_{\tau}}, \syncOf{\ppa_{\tau}})$, where
	\begin{itemize}
	\item for each $(s, \alpha) \in \domain(\transFctOf{\ppa})$: 
		$\transFctOf{\ppa_{\tau}}(s, \alpha) = \transFctOf{\ppa}(s, \alpha) 
			\text{ and } 
			\syncOf{\ppa_{\tau}}(s, \alpha) = \syncOf{\ppa}(s, \alpha),$
	\item for each $s \in \stateSetOf{\ppa}$: 
		$\transFctOf{\ppa_{\tau}}(s, \tau) = \indicatorFct{s_{\tau}} 
			\text{ and } 
			\syncOf{\ppa_{\tau}}(s, \tau) = \tau. $
	\end{itemize}
\end{definition}
By definition of $\ppa_{\tau}$ follows, that a valuation is well-defined (graph-preserving) for pPA $\ppa$ iff it is well-defined (graph-preserving) for $\ppa_\tau$. 
From the following theorem follows that verifying objectives w.r.t\ partial strategies in $\ppa$ is equivalent to  verifying objectives w.r.t\ complete strategies in $\ppa_{\tau}$, 
i.e., $\ppa, \region \modelsWrt{(\mless,)\prt} \generalPredicate \Leftrightarrow  \ppa_{\tau}, \region \modelsWrt{(\mless,)\comp}$. 
\Cref{theo_prob_reward_partial_vs_complete} is based on the proof of \cite[Proposition 2]{Kwi+13}. 
However, we modified the construction such that the memoryless property is preserved. 
%\hm{check if graph-preserving needed}
\begin{restatable}{lemma}{restatableProbRewardPartVsComplete}
	\label{theo_prob_reward_partial_vs_complete} 
	Let $\regLang$ be a language and $\rewFct$ be a (parametric) reward function over $\alphabetOf{}$ with $\tau \not \in \alphabetOf{}$. 
	\begin{itemize} 
	\item For any strategy $\strategy \in \strategysetOf{\ppa}{(\mless,)\prt}$ there is a strategy $\strategy_{\tau} \in \strategysetOf{\ppa_{\tau} }{(\mless,)\comp}$ 
		such that for any valuation $\valuation$ that is well-defined for $\ppa$ (and $\rewFct$) 
		\[
			\PrOf{\ppa}{\valuation,\strategy}{\regLang}  = \PrOf{\ppa_{\tau}}{\valuation, \strategy_{\tau}}{\regLang}  \quad \text{ and } \quad
			\ExpTot{\ppa}{\valuation,\strategy}{\rewFct} = \ExpTot{\ppa_{\tau}}{\valuation, \strategy_{\tau}}{\rewFct}.
		\]
	\item  
		For any strategy $\strategy_{\tau} \in \strategysetOf{\ppa_{\tau}}{(\mless,)\comp}$ there is a strategy $\strategy \in  \strategysetOf{\ppa }{(\mless,)\prt}$ 
		such that for any valuation $\valuation$ that is well-defined for $\ppa$ (and $\rewFct$) 
		\[
			\PrOf{\ppa_{\tau}}{\valuation, \strategy_{\tau}}{\regLang}  = \PrOf{\ppa}{\valuation,\strategy}{\regLang}
			 \quad  \text{ and } \quad 
			\ExpTot{\ppa_{\tau}}{\valuation, \strategy_{\tau}}{\rewFct} = \ExpTot{\ppa}{\valuation, \strategy}{\rewFct}.
		\]
	\end{itemize}
\end{restatable}
Intuitively, all ``unused'' probabilities of a partial strategy of $\ppa$ are redirected to a sink state by the complete strategy of $\ppa_{\tau}$. 
In \cite[proposition 2]{Kwi+13}, $\tau$-labeled selfloop are introduced. 
However, the resulting complete strategy requires memory. 
\begin{proof}
       First, given a strategy $\strategy \in \strategysetOf{\ppa}{(\mless,)\prt}$, we construct a strategy $\strategy_{\tau} \in \strategysetOf{\ppa_{\tau} }{(\mless,)\comp}$ as follows: 
       We can apply the reasoning provided in the proof of \cite[proposition 2]{Kwi+13} to our construction: 
        Given a finite path $\pi$ of $\ppa_{\tau}$, and $\alpha \in \actSetOf{\ppa_\tau}$: 
        \begin{align*}
            \strategy_{\tau}(\pi, \alpha) = 
            \begin{cases}
                \strategy(\pi, \alpha) & \text{if  $\pi \in \finPathsOf{\ppa}{}$ and $\alpha \in \actSetOf{\ppa}$} \\
                1- \sum_{\alpha \in \actSetOf{\ppa}} \strategy(\pi, \alpha) & \text{if $\pi  \not \in \finPathsOf{\ppa}{}$ and } \alpha= \tau \\
                1 & \text{if $\pi \in \finPathsOf{\ppa}{}$ and } \alpha= \tau \\
                0 & \text{otherwise}
            \end{cases}
        \end{align*}
        Now, for any valuation $\valuation$ that is well-defined for $\ppa$ (and $\rewFct$), we have $\PrOf{\ppa}{\valuation,\strategy}{\regLang}  = \PrOf{\ppa_{\tau}}{\valuation, \strategy_{\tau}}{\regLang}$ and $\ExpTot{\ppa}{\valuation,\strategy}{\rewFct} = \ExpTot{\ppa_{\tau}}{\valuation, \strategy_{\tau}}{\rewFct}.$
        Note that---in contrast to \cite[proposition 2]{Kwi+13}---the strategy $\strategy_{\tau}$ is memoryless if $\strategy$ is a memoryless strategy. 

        Second, given a strategy $\strategy_{\tau} \in \strategysetOf{\ppa_{\tau}}{(\mless,)\comp}$, we construct a strategy $\strategy \in \strategysetOf{\ppa}{(\mless,)\prt}$ as follows: 
        For ${\pi} \in \finPathsOf{\ppa}{} \subseteq \finPathsOf{\ppa_{\tau}}{}$ and $\alpha \in \actSetOf{\ppa}$ it holds that 
			$\strategy({\pi}, \alpha)  = \strategy_{\tau}({\pi}, \widehat{a}, {\mu} )$. 
        Again, for any valuation $\valuation$ that is well-defined for $\ppa$ (and $\rewFct$) we have $\PrOf{\ppa}{\valuation,\strategy}{\regLang}  = \PrOf{\ppa_{\tau}}{\valuation, \strategy_{\tau}}{\regLang}$ and $\ExpTot{\ppa}{\valuation,\strategy}{\rewFct} = \ExpTot{\ppa_{\tau}}{\valuation, \strategy_{\tau}}{\rewFct}$
        Additionally, if $\strategy_{\tau}$ is memoryless, then $\strategy$ is also memoryless. 
\end{proof}

\subsection{Equivalence w.r.t.\ Multiobjective Queries}\label{app:equivalence_partial_complete_qmo}
As a consequence of \Cref{theo_prob_reward_partial_vs_complete}, we directly obtain \cite[Proposition 2]{Kwi+13} for pPA and parametric mo-queries: 
\begin{corollary}%[Equivalence of Partial and Complete Strategies  w.r.t.\ Mo-Queries]
	\label{lemma_moq_partial_via_complete}
	Let $\ppa$ be a pPA and let $\multiobjectiveQuery{}{X}$ be a (parametric) multiobjective query. 
	Construct pPA $\ppa_{\tau}$ as in \Cref{def:ppaTau}. 
	$\alphabetOf{\multiobjectiveQuery{}{X}} \subseteq \alphabetOf{\ppa}$, where $\tau \not \in \alphabetOf{\ppa} \cup \alphabetOf{\multiobjectiveQuery{}{X}}$ . 
	Let $\region$ be graph-preserving well-defined for $\ppa$ (i.e., also for $\ppa_{\tau}$) and well-defined for $\multiobjectiveQuery{}{X}$. 
	Then, 
	\begin{align*}
		& \exists \strategy \in \strategysetOf{\ppa}{(\mless,)\prt}: \exists \valuation \in \region :  \ppa[\valuation], {\strategy} \modelsWrt{} \multiobjectiveQuery{}{X}[\valuation]
        \\ & \quad \Leftrightarrow \quad  \exists \strategy \in  \strategysetOf{\ppa_{\tau} }{(\mless,)\comp}: \exists \valuation \in \region :  \ppa_{\tau}[\valuation], {\strategy} \modelsWrt{}  \multiobjectiveQuery{}{X}[\valuation]
	\end{align*}

\end{corollary}

\subsection{Equivalence w.r.t.\ Monotonicity}\label{app:equivalence_partial_complete_mono}
As a direct consequence of \Cref{theo_prob_reward_partial_vs_complete}, we obtain that monotonicity for partial strategies w.r.t\ general properties is equivalent to monotonicity for complete strategies in $\ppa_{\tau}$ (see \Cref{def:ppaTau}). 
\begin{restatable}{corollary}{restatablePartialCompleteMono}\label{theo:reducing_partial_mono_to_complete}
    Let $\region$ be a region that is well-defined for pPA $\ppa$ and let $\budarrow \in \{\buparrow, \bdownarrow \}$. 
	For the solution functions $\solutionFctMdpObjective{\ppa}{}{}$ and $\solutionFctMdpObjective{\ppa_{\tau}}{}{}$ for $\ppa$ and $\ppa_{\tau}$ with $\tau \not \in \alphabetOf{\solutionFctMdpObjective{\ppa}{}{}}$ holds 
    $\monotonicOnRegionParameter{\budarrow}{\solutionFctMdpObjective{\ppa}{}}{p}{\region}{(\mless,)\prt}$ iff 
 	$\monotonicOnRegionParameter{\budarrow}{\solutionFctMdpObjective{{\ppa_{\tau}}}{}}{p}{\region}{(\mless,)\comp}$. 	
\end{restatable}
\Cref{theo:reducing_partial_mono_to_complete} also holds for local monotonicity.

\section{Further AG Proof Rules}\label{app:pag_extension}
We extend the remaining proof rules from \cite{Kwi+13} to pPA. This includes \Cref{theo:pag_conj_rule}, which enables reasoning about the conjunction of multi-objective queries, as well as \Cref{theo:pag_asymN_rule}, an asymmetric rule for more than two components that extends \Cref{theo:pag_asym_rule}. 
Additionally, we introduce \Cref{theo:pag_interleaving_rule} for handling interleaving components and \Cref{theo:pag_reward_rule}, which supports modular reasoning about reward structures. 

When the premises include $\ppa, \region  \modelsWrt{\star} \multiobjectiveQuery{}{A}$ we adhere to the following assumptions:
\begin{itemize}
    \item $\multiobjectiveQuery{}{A}$ is a (parametric) mo-query for $\ppa$, i.e., $\multiobjectiveQuery{}{A}$ is defined over the alphabet  $\alphabetOf{\multiobjectiveQuery{}{A}} \subseteq \alphabetOf{\ppa}$, and
    \item $\region$ is a region that is well-defined for $\ppa$ (and $\multiobjectiveQuery{}{A}$). 
\end{itemize}
Similarly, if an AG triple $\agTriple{\alphabetExtensionOfTo{\ppa}{\alphabetOf{\multiobjectiveQuery{}{A}}}, \region}{\star}{\multiobjectiveQuery{}{A}}{\multiobjectiveQuery{}{G}}$ occurs in the premises, we assume:
\begin{itemize}
    \item $\multiobjectiveQuery{}{G}$ is a (parametric) mo-query for $\alphabetExtensionOfTo{\ppa}{\alphabetOf{\multiobjectiveQuery{}{A}}}$, i.e., $\multiobjectiveQuery{}{G}$ is defined over the alphabet $\alphabetOf{\multiobjectiveQuery{}{G}} \subseteq \alphabetOf{\ppa} \cup \alphabetOf{\multiobjectiveQuery{}{A}}$\footnote{$\multiobjectiveQuery{(\safe)}{A} \not \subseteq \alphabetOf{\ppa}$ is possible.}, and 
    \item the region $\region$ is well-defined for $\ppa$ ($\multiobjectiveQuery{}{A}$, and $\multiobjectiveQuery{}{G}$). 
\end{itemize}
Safety mo-queries are marked with a superscript ``$\safe$''. %$\multiobjectiveQuery{\safe}{X}$. 

To reason about the conjunction of multi-objective queries, we generalize \cite[Theorem 2]{Kwi+13} for pPA.
When safety-related multi-objective queries are combined through conjunction, the resulting query remains a safety multi-objective query. 
\begin{restatable}[Conjunction]{theorem}{restatableConjRule}\label{theo:pag_conj_rule} 
   Let $\region_{\regionIntersection} = \regionIntersectionOf{\region_1}{}{\region_2}$, $\multiobjectiveQuery{}{A}_{\land} = \multiobjectiveQuery{}{A}_{1} \land \multiobjectiveQuery{}{A}_{2}$, and $\multiobjectiveQuery{}{G}_{\land} = \multiobjectiveQuery{}{G}_{1} \land \multiobjectiveQuery{}{G}_{2}$.

    \begin{tabularx}{\linewidth}{p{0.45\linewidth}| p{0.45\linewidth}} 
                $
                \infer{
                    \agTriple{\alphabetExtensionOfTo{\ppa}{\alphabetOf{\multiobjectiveQuery{\safe}{A}_{\land}} }, \region_{\regionIntersection} }{\prt}{\multiobjectiveQuery{\safe}{A}_{\land}}{\multiobjectiveQuery{\safe}{G}_{\land}} 
                } 
                {
                    \deduce{
                        \agTriple{\alphabetExtensionOfTo{\ppa}{\alphabetOf{\multiobjectiveQuery{\safe}{A}_{2}}}, \region_2 }{\prt}{\multiobjectiveQuery{\safe}{A}_{2}}{\multiobjectiveQuery{\safe}{G}_2} 
                    }{
                        \agTriple{\alphabetExtensionOfTo{\ppa}{\alphabetOf{\multiobjectiveQuery{\safe}{A}_{1}}}, \region_1 }{\prt}{\multiobjectiveQuery{\safe}{A}_{1}}{\multiobjectiveQuery{\safe}{G}_1} 
                    }
                    } 
                $
        & 
        $
        \infer{
            \agTriple{\alphabetExtensionOfTo{\ppa}{\alphabetOf{\multiobjectiveQuery{}{A}_{\land}} }, \region_{\regionIntersection} }{\fairWrtRegionModel{}{\decomp_1 \cup \decomp_2}}{\multiobjectiveQuery{}{A}_{\land}}{\multiobjectiveQuery{}{G}_{\land}} 
        } 
        {
            \deduce{
                \agTriple{\alphabetExtensionOfTo{\ppa}{\alphabetOf{\multiobjectiveQuery{}{A}_{2}}}, \region_2 }{\fairWrtRegionModel{}{\decomp_2}}{\multiobjectiveQuery{}{A}_{2}}{\multiobjectiveQuery{}{G}_2} 
            }{
                \agTriple{\alphabetExtensionOfTo{\ppa}{\alphabetOf{\multiobjectiveQuery{}{A}_{1}}}, \region_1 }{\fairWrtRegionModel{}{\decomp_1}}{\multiobjectiveQuery{}{A}_{1}}{\multiobjectiveQuery{}{G}_1} 
            }
            } 
        $
        \\ %
        \\
        &  
        for $\decomp_i \in \setOfdecompsOf{\alphabetOf{\ppa} \cup \alphabetOf{\multiobjectiveQuery{}{A}_i}}$. \\
    \end{tabularx}
\end{restatable}
\begin{proof}
    We show the statement for safety properties. For general properties and fair strategies, the proof is analogous. 
    First, we observe that ${\multiobjectiveQuery{\safe}{G}_{\land}}$ is a valid multiobjective query for $\alphabetExtensionOfTo{\ppa}{\alphabetOf{\multiobjectiveQuery{}{A}_{\land}}}$ as 
    $\alphabetOf{\multiobjectiveQuery{\safe}{G}_i} \subseteq \alphabetOf{\ppa} \cup \alphabetOf{\multiobjectiveQuery{\safe}{A}_i} = \alphabetOf{\alphabetExtensionOfTo{\ppa}{\alphabetOf{\multiobjectiveQuery{}{A}_{\land}}}}$ holds. 
    Second, assume the premises hold. 
    As $\alphabetOf{\multiobjectiveQuery{\safe}{G}_i} \subseteq \alphabetOf{\ppa} \cup \alphabetOf{\multiobjectiveQuery{\safe}{A}_i}$,  
    it holds that for $i \in \{1,2\}$: 
        \[ 
            \agTriple{\alphabetExtensionOfTo{\ppa}{\alphabetOf{\multiobjectiveQuery{\safe}{A}_1 \land \multiobjectiveQuery{\safe}{A}_2}}, \region_i}{\prt}{\multiobjectiveQuery{\safe}{A}_i}{\multiobjectiveQuery{\safe}{G}_i}
        \]
	From this it follows that 
        \[ 
            \agTriple{\alphabetExtensionOfTo{\ppa}{\alphabetOf{\multiobjectiveQuery{\safe}{A}_\land}}, \regionIntersectionOf{\region_1}{}{\region_2} }{\prt}{\multiobjectiveQuery{\safe}{A}_i}{\multiobjectiveQuery{\safe}{G}_i}
        \]
    As in the proof of \cite[Theorem 3]{Kwi+13}, by the definition of pAG-triples (\Cref{def_pag_triple}), 
    we observe, that for $\valuation \in \regionIntersectionOf{\region_1}{}{\region_2}$ and $\strategy \in \strategysetOf{\alphabetExtensionOfTo{\ppa}{\alphabetOf{\multiobjectiveQuery{\safe}{A}_\land}}}{\prt}$ it holds that 
	\begin{align*}
	 & \bigwedge_{i\in\{1,2\}} \left(\alphabetExtensionOfTo{\ppa}{\alphabetOf{\multiobjectiveQuery{\safe}{A}_1 \land \multiobjectiveQuery{\safe}{A}_2}} [\valuation], {\strategy} \models \multiobjectiveQuery{\safe}{A}_i  \quad  \rightarrow  \quad  \alphabetExtensionOfTo{\ppa}{\alphabetOf{\multiobjectiveQuery{\safe}{A}_1 \land \multiobjectiveQuery{\safe}{A}_2}}[\valuation], {\strategy} \models \multiobjectiveQuery{\safe}{G}_i \right)  \\
	  \Leftrightarrow\quad  
      & \bigwedge_{i\in\{1,2\}} \left(\alphabetExtensionOfTo{\ppa}{\alphabetOf{\multiobjectiveQuery{\safe}{A}_1 \land \multiobjectiveQuery{\safe}{A}_2}} [\valuation], {\strategy} \not\models \multiobjectiveQuery{\safe}{A}_i  \quad  \lor \quad  \alphabetExtensionOfTo{\ppa}{\alphabetOf{\multiobjectiveQuery{\safe}{A}_1 \land \multiobjectiveQuery{\safe}{A}_2}}[\valuation], {\strategy} \models \multiobjectiveQuery{\safe}{G}_i \right) \\
	  \Rightarrow \quad &
	    \alphabetExtensionOfTo{\ppa}{\alphabetOf{\multiobjectiveQuery{\safe}{A}_\land}} [\valuation],{\strategy} \not\models \multiobjectiveQuery{\safe}{A}_\land  \quad \lor \quad
	    \alphabetExtensionOfTo{\ppa}{\alphabetOf{\multiobjectiveQuery{\safe}{A}_\land}}[\valuation],{\strategy} \models \multiobjectiveQuery{\safe}{G}_\land \\
	\Leftrightarrow \quad &	\alphabetExtensionOfTo{\ppa}{\alphabetOf{\multiobjectiveQuery{\safe}{A}_\land}} [\valuation], {\strategy} \models \multiobjectiveQuery{\safe}{A}_\land  \quad  \rightarrow \quad  \alphabetExtensionOfTo{\ppa}{\alphabetOf{\multiobjectiveQuery{\safe}{A}_\land}}[\valuation], {\strategy} \models \multiobjectiveQuery{\safe}{G}_\land
\end{align*}  
Thus, 
\begin{align*}
	{\agTriple{\alphabetExtensionOfTo{\ppa}{\alphabetOf{\multiobjectiveQuery{\safe}{A}_\land}}, \regionIntersectionOf{\region_1}{}{\region_2} }{\prt}{ \multiobjectiveQuery{\safe}{A}_\land}{ \multiobjectiveQuery{\safe}{G}_\land}}
\end{align*}
    
\end{proof} 
By combining \Cref{theo:pag_asym_rule} and \Cref{theo:pag_conj_rule}, we obtain the 
following analogue to \cite[Theorem 4]{Kwi+13}---the asymmetric proof rule  for $n>2$ components---for pPA.  %
\begin{restatable}[Asymmetric-N]{theorem}{restatableAsymNRule}\label{theo:pag_asymN_rule}
    Let $\region_{\regionIntersection}  = \regionIntersectionOf{\region_1}{\dots}{\region_n}$, and $ \ppa =  \ppa_1 \parallel \dots \parallel \ppa_n$. 

    \begin{tabularx}{\linewidth}{p{0.45\linewidth}|p{0.45\linewidth}} 
           %\\  \hline
           %\parbox{0.4\textwidth}
                $
                \infer{
                    \ppa, \region_{\regionIntersection} \modelsWrt{\comp} \multiobjectiveQuery{\safe}{G}
                } 
                {
                    \deduce{
                        \agTriple{\alphabetExtensionOfTo{\ppa_n}{\alphabetOf{\multiobjectiveQuery{\safe}{A}_{n-1}}}, \region_{n} }{\prt}{\multiobjectiveQuery{\safe}{A}_{n-1}}{\multiobjectiveQuery{\safe}{G}} 
                    }{
                        \deduce{
                            \vdots
                            }{
                            \deduce{
                                \agTriple{\alphabetExtensionOfTo{\ppa_2}{\alphabetOf{\multiobjectiveQuery{\safe}{A}_{1}}}, \region_2 }{\prt}{\multiobjectiveQuery{\safe}{A}_{1}}{\multiobjectiveQuery{\safe}{A}_2} 
                                }{
                                    \deduce{
                                        \ppa_1, \region_1 \modelsWrt{\comp} \multiobjectiveQuery{\safe}{A}_1
                                        }{}
                                }
                        }
                }} 
                $
        & 
            $
                \infer{
                    \ppa, \region_{\regionIntersection} \modelsWrt{\fairWrtRegionModel{}{
                    \decomp_1 \cup  \dots \cup \decomp_n}} \multiobjectiveQuery{}{G}
                }
                {\deduce{
                    \agTriple{\alphabetExtensionOfTo{\ppa_n}{\alphabetOf{\multiobjectiveQuery{}{A}_{n-1}}}, \region_{n}}{\fairWrtRegionModel{}{\decomp_n}}{\multiobjectiveQuery{}{A}_{n-1}}{\multiobjectiveQuery{}{G}} 
                    }{
                        \deduce{
                            \vdots
                        }{
                            \deduce{
                                \agTriple{\alphabetExtensionOfTo{\ppa_2}{\alphabetOf{\multiobjectiveQuery{}{A}_{1}}}, \region_2 }{\fairWrtRegionModel{}{\decomp_2}}{\multiobjectiveQuery{}{A}_{1}}{\multiobjectiveQuery{}{A}_2} 
                            }{
                                \deduce{
                                    \ppa_1, \region_1 \modelsWrt{\fairWrtRegionModel{}{\decomp_1}} \multiobjectiveQuery{}{A}_1
                                }{}
                            }
                        }
                    }
                }
             $
        % \end{align*}
        \\ %
        \\
        &  % $\decomp_1 \in \setOfdecompsOf{\ppa_1}$ and 
       for $\decomp_1 \in \setOfdecompsOf{\alphabetOf{\ppa_1}}$ and $\decomp_i \in \setOfdecompsOf{\alphabetOf{\ppa_i}\cup \alphabetOf{\multiobjectiveQuery{}{A}_{i-1}}}$, for $i>1$.  
    \end{tabularx}
    
\end{restatable}
\begin{proof}
    Analogously to \cite[Theorem 4]{Kwi+13}, the result follows by a repeated application of \Cref{theo:pag_asym_rule} and \Cref{theo:pag_conj_rule}. 
\end{proof}

Next, we extend \cite[Theorem 6]{Kwi+13}\textemdash{}a rule for reasoning about properties with non-synchronized actions.  
%Note: union of safety property is a safety property.
\begin{restatable}[Interleaving]{theorem}{restatableInterleavingRule}\label{theo:pag_interleaving_rule}

  Let $\region_{\regionIntersection} = \regionIntersectionOf{\region_1}{}{\region_2}$, $\multiobjectiveQuery{}{A}_{\land} = \multiobjectiveQuery{}{A}_{1} \land \multiobjectiveQuery{}{A}_{2}$, and $\ppa = \alphabetExtensionOfTo{(\ppa_1 \parallel \ppa_2)}{\alphabetOf{\multiobjectiveQuery{}{A}_{\land}}}$. 

    \begin{tabularx}{\linewidth}{p{0.45\linewidth}|p{0.45\linewidth}} 
           %\\  \hline
           %\parbox{0.4\textwidth}
         $
                \infer{
                    \agTriple{{ \ppa}, \region_{\regionIntersection}}{\prt}{\multiobjectiveQuery{\safe}{A}_{\land} }{\probPredicate{\substack{\geq p_1+p_2 \\ \ - p_1\cdot p_2}}{\regLang_1^{\safe} \cup \regLang_2^{\safe}}}
                } 
                {
                    \deduce{
                        \agTriple{\alphabetExtensionOfTo{\ppa_2}{\alphabetOf{\multiobjectiveQuery{\safe}{A}_{2}}}, \region_2 }{\prt}{\multiobjectiveQuery{\safe}{A}_{2}}{{\probPredicate{\geq p_2}{\regLang_2^{\safe}} }} 
                    }{
                                    \deduce{
                                        \agTriple{\alphabetExtensionOfTo{\ppa_1}{\alphabetOf{\multiobjectiveQuery{\safe}{A}_{1}}}, \region_1 }{\prt}{\multiobjectiveQuery{\safe}{A}_{1}}{\probPredicate{\geq p_1}{\regLang_1^{\safe}}  }
                                        }{}
                }} 
             $
        & 
       $ 
            \infer{
                \agTriple{{ \ppa}, \region_{\regionIntersection}}{{\fairWrtRegionModel{}{ \decomp_1 \cup \decomp_2 }  }  }{\multiobjectiveQuery{}{A}_{\land} }{\probPredicate{ \substack{\sim p_1+p_2 \\ \  - p_1\cdot p_2}}{\regLang_1^{} \cup \regLang_2^{}}}
            } 
            {
                \deduce{
                    \agTriple{\alphabetExtensionOfTo{\ppa_2}{\alphabetOf{\multiobjectiveQuery{}{A}_{2}}}, \region_2 }{\fairWrtRegionModel{}{\decomp_2}}{\multiobjectiveQuery{}{A}_{2}}{{\probPredicate{\sim p_2}{\regLang_2^{}} }} 
                }{
                                \deduce{
                                    \agTriple{\alphabetExtensionOfTo{\ppa_1}{\alphabetOf{\multiobjectiveQuery{}{A}_{1}}}, \region_1 }{\fairWrtRegionModel{}{\decomp_1}}{\multiobjectiveQuery{}{A}_{1}}{\probPredicate{\sim p_1}{\regLang_1^{}}  }
                                    }{}
            }} 
        $
      \\ %
      \\
        &  
        for $\sim \in \{ <, \leq, >, \geq\}$ and 
       $\decomp_i \in \setOfdecompsOf{\alphabetOf{\ppa_i} \cup \alphabetOf{\multiobjectiveQuery{}{A}_{i}}}$. 
    \end{tabularx}
\end{restatable}

\begin{proof}
    The proof is based on the proof of \cite[Theorem 6]{Kwi+13}. 
    First, ${\probPredicate{ \substack{\sim p_1+p_2 \\ \  - p_1\cdot p_2}}{\regLang_1^{(\safe)} \cup \regLang_2^{(\safe)}}}$ is a valid query for $\ppa$, since $\alphabetOf{{\probPredicate{ \substack{\sim p_1+p_2 \\ \  - p_1\cdot p_2}}{\regLang_1^{(\safe)} \cup \regLang_2^{(\safe)}}}} \subseteq \alphabetOf{\ppa} $,  
    We show the proof of correctness of the rule involving general predicates and fair strategies. 
    For safety properties the proof works analogously. 

    %  This also holds for fair rule. 
    If $\regionIntersectionOf{\region_{1}}{}{\region_2} = \emptyset$, the conclusion trivially holds. 
    We assume $\regionIntersectionOf{\region_{1}}{}{\region_2} \not = \emptyset$. 
    We show that $\agTriple{{\alphabetExtensionOfTo{\ppa_1}{\alphabetOf{\multiobjectiveQuery{}{A}_1}}\parallel \alphabetExtensionOfTo{\ppa_2}{\alphabetOf{\multiobjectiveQuery{}{A}_2}}}, \region_{\regionIntersection}}{{\fairWrtRegionModel{}{ \decomp_1 \cup \decomp_2 }  }  }{\multiobjectiveQuery{}{A}_{\land} }{\probPredicate{ \substack{\sim p_1+p_2 \\ \  - p_1\cdot p_2}}{\regLang_1^{} \cup \regLang_2^{}}}$
    holds, which implies that 
    $\agTriple{{ \ppa}, \region_{\regionIntersection}}{{\fairWrtRegionModel{}{ \decomp_1 \cup \decomp_2 }  }  }{\multiobjectiveQuery{}{A}_{\land} }{\probPredicate{ \substack{\sim p_1+p_2 \\ \  - p_1\cdot p_2}}{\regLang_1^{} \cup \regLang_2^{}}}$ holds. 

	Let $\valuation$ be a valuation in $\regionIntersectionOf{\region_1}{}{\region_2}$ and $\strategy$ be a $\fairWrtRegionModel{}{\decomp_1 \cup  \decomp_2}$ strategy of 
	$(\alphabetExtensionOfTo{\ppa_1}{\alphabetOf{\multiobjectiveQuery{}{A}_1}}\parallel \alphabetExtensionOfTo{\ppa_2}{\alphabetOf{\multiobjectiveQuery{}{A}_2}})[\valuation]$. 
	We show	${(\alphabetExtensionOfTo{\ppa_1}{\alphabetOf{\multiobjectiveQuery{}{A}_1}}\parallel \alphabetExtensionOfTo{\ppa_2}{\alphabetOf{\multiobjectiveQuery{}{A}_2}})}[\valuation], {\strategy}\models  \multiobjectiveQuery{}{A}_1 \land \multiobjectiveQuery{}{A}_1[ \valuation]  \text{ implies }
    {(\alphabetExtensionOfTo{\ppa_1}{\alphabetOf{\multiobjectiveQuery{}{A}_1}}\parallel \alphabetExtensionOfTo{\ppa_2}{\alphabetOf{\multiobjectiveQuery{}{A}_2}})}[\valuation], {\strategy} \models  \probPredicate{\substack{\sim p_1+p_2 \\ \  - p_1\cdot p_2}}{\regLang_1 \cup \regLang_2}$. 
	
	Assume ${(\alphabetExtensionOfTo{\ppa_1}{\alphabetOf{\multiobjectiveQuery{}{A}_1}}\parallel \alphabetExtensionOfTo{\ppa_2}{\alphabetOf{\multiobjectiveQuery{}{A}_2}})}[\valuation], {\strategy} \models  \multiobjectiveQuery{}{A}_1 \land \multiobjectiveQuery{}{A}_1[\valuation]$ holds. 
			
    We have $\agTriple{\alphabetExtensionOfTo{\ppa_i}{\multiobjectiveQuery{}{A}_i}, \regionIntersectionOf{\region_1}{}{\region_2} }{\fairWrtRegionModel{}{\decomp_i} }{\multiobjectiveQuery{}{A}_i}{\probPredicate{\sim p_i}{\regLang_i}}$ for $i \in \{1,2\}$ by the premises.  
    Thus, 
	\begin{align}
		& {(\alphabetExtensionOfTo{\ppa_1}{\alphabetOf{\multiobjectiveQuery{}{A}_1}}\parallel \alphabetExtensionOfTo{\ppa_2}{\alphabetOf{\multiobjectiveQuery{}{A}_2}})}[\valuation], {\strategy} \models   \multiobjectiveQuery{}{A}_i[ \valuation ] \nonumber \\
		& \Rightarrow \alphabetExtensionOfTo{\ppa_i }{\alphabetOf{\multiobjectiveQuery{}{A}_i}}[\valuation], {\stratProjOfToValuation{\strategy}{i}{\valuation}} \models   \multiobjectiveQuery{}{A}_i[ \valuation ]
		&&  \text{By \Cref{lemma_3_fDependent}} \nonumber \\
		& \Rightarrow 
		 \alphabetExtensionOfTo{\ppa_i}{\alphabetOf{\multiobjectiveQuery{}{A}_i}}[{\valuation}], \stratProjOfToValuation{\strategy}{i}{\valuation} \models   \probPredicate{\sim p_i}{\regLang_i}
		 && \text{$ \stratProjOfToValuation{\strategy}{i}{\valuation} \in \strategysetOf{{\alphabetExtensionOfTo{\ppa_i}{\alphabetOf{\multiobjectiveQuery{}{A}_i} }}[\valuation]}{\fairWrtRegionModel{}{\decomp_i}}$} \nonumber\\
		 & && \text{and $\agTriple{\alphabetExtensionOfTo{\ppa_i}{\alphabetOf{\multiobjectiveQuery{}{A}_i}}, \regionIntersectionOf{\region_1}{}{\region_2} }{\fairWrtRegionModel{}{\decomp_i}}{\multiobjectiveQuery{}{A}_i}{\probPredicate{\sim p_i}{\regLang_i}}$} \nonumber \\
		 & \Rightarrow \PrOf{\alphabetExtensionOfTo{\ppa_i}{\alphabetOf{\multiobjectiveQuery{}{A}_i}}}{{\valuation},\stratProjOfToValuation{\strategy}{i}{\valuation}}{\regLang_i}  \sim p_i && \label{eq:proof_async_quant_G_sat} 
	\end{align}
    We have $(\alphabetOf{\ppa_1} \cup \alphabetOf{\multiobjectiveQuery{}{A}_1} )\cap (\alphabetOf{\ppa_2} \cup \alphabetOf{\multiobjectiveQuery{}{A}_2}) = \emptyset$. 
    Thus,
    %\hm{is  $\infPathsOf{{\ppa}[\valuation]}$ ok  or use $\infPathsOf{{\ppa}}$}
    \begin{align*}
    &	\PrOf{{(\alphabetExtensionOfTo{\ppa_1}{\alphabetOf{\multiobjectiveQuery{}{A}_1}}\parallel \alphabetExtensionOfTo{\ppa_2}{\alphabetOf{\multiobjectiveQuery{}{A}_2}})}}{\valuation,\strategy}{\{
            \pi \in \infPathsOf{{(\alphabetExtensionOfTo{\ppa_1}{\alphabetOf{\multiobjectiveQuery{}{A}_1}}\parallel \alphabetExtensionOfTo{\ppa_2}{\alphabetOf{\multiobjectiveQuery{}{A}_2}})}[\valuation]} \mid \restrOfTo{\traceOf{\pi}}{\alphabetOf{\regLang_1 \cup \regLang_2}}  \not \in \regLang_1 \cup \regLang_2 \}} \\
        %
        %
        % we use this notation insetad of complement of L, because didnt know if wrt inf or fin words (more correct if change to inf fin words)
        %
        &	\quad  =  \prod_{i \in \{1,2\}}	
            \PrOf{\alphabetExtensionOfTo{\ppa_i}{\alphabetOf{\multiobjectiveQuery{}{A}_i }}}{\valuation,\stratProjOfToValuation{\strategy}{i}{\valuation}}{  \{ \pi \in \infPathsOf{\alphabetExtensionOfTo{\ppa_i }{\alphabetOf{\multiobjectiveQuery{}{A}_i}}[\valuation]} \mid \restrOfTo{\traceOf{\pi}}{\alphabetOf{\regLang_i}}  \not \in \regLang_i \} } 
        \\
        &	\quad  =  \left(1-\PrOf{\alphabetExtensionOfTo{\ppa_1}{\alphabetOf{\multiobjectiveQuery{}{A}_1}}}{{\valuation},\stratProjOfToValuation{\strategy}{1}{\valuation}}{\regLang_1} \right) \cdot
            \left(1 - \PrOf{\alphabetExtensionOfTo{\ppa_2}{2}}{{\valuation},\stratProjOfToValuation{\strategy}{2}{\valuation}}{\regLang_2}  \right)
    \end{align*}
    Then, by \Cref{eq:proof_async_quant_G_sat}: 
    \begin{align*}
        & \PrOf{{(\alphabetExtensionOfTo{\ppa_1}{\alphabetOf{\multiobjectiveQuery{}{A}_1}}\parallel \alphabetExtensionOfTo{\ppa_2}{\alphabetOf{\multiobjectiveQuery{}{A}_2}})}}{\valuation,\strategy}{\regLang_1 \cup \regLang_2} 
        \\
        & = 1- 	\PrOf{{(\alphabetExtensionOfTo{\ppa_1}{\alphabetOf{\multiobjectiveQuery{}{A}_1}}\parallel \alphabetExtensionOfTo{\ppa_2}{\alphabetOf{\multiobjectiveQuery{}{A}_2}})}}{\valuation,\strategy}{ \{
            \pi \in \infPathsOf{(\alphabetExtensionOfTo{\ppa_1}{\alphabetOf{\multiobjectiveQuery{}{A}_1}}\parallel \alphabetExtensionOfTo{\ppa_2}{\alphabetOf{\multiobjectiveQuery{}{A}_2}})[\valuation]}  \mid \restrOfTo{\traceOf{\pi}}{\alphabetOf{\regLang_1 \cup \regLang_2}}  \not \in \regLang_1 \cup \regLang_2 \}} \\
        & =1- \left(1-\underbrace{\PrOf{\alphabetExtensionOfTo{\ppa_1}{\alphabetOf{\multiobjectiveQuery{}{A}_1}}}{{\valuation},\stratProjOfToValuation{\strategy}{1}{\valuation}}{\regLang_1}}_{\sim p_1} \right) \cdot
        \left(1 - \underbrace{\PrOf{\alphabetExtensionOfTo{\ppa_2}{\alphabetOf{\multiobjectiveQuery{}{A}_2}}}{{\valuation},\stratProjOfToValuation{\strategy}{2}{\valuation}}{\regLang_2}}_{\sim p_2} \right) \\
        & \sim p_1 +p_2 -p_1\cdot p_2
    \end{align*}
\end{proof}
Lastly, we lift \cite[Theorem 7]{Kwi+13} to pPA, which reasons about the sum of two (parametric) reward functions $\rewFct_1$, $\rewFct_2$ for $\ppa_1$ and $\ppa_2$, respectively, in the composition $\ppa_1 \parallel \ppa_2$. 
The sum of $\rewFct_1$ and $\rewFct_2$ is reward function $(\rewFct_1 + \rewFct_2) $ on $\alphabetOf{\rewFct_1} \cup  \alphabetOf{\rewFct_2} $ over the parameter set $\parameterSetOf{\rewFct_1} \cup \parameterSetOf{\rewFct_2}$. 
\begin{align*}
    (\rewFct_1 + \rewFct_2)(a) = 
    \begin{cases}
        \rewFct_1(a) + \rewFct_2(a), & \text{for } a \in \alphabetOf{\rewFct_1} \cap \alphabetOf{\rewFct_2} \\
        \rewFct_1(a), & \text{for }a \in \alphabetOf{\rewFct_1} \setminus \alphabetOf{\rewFct_2} \\
        \rewFct_2(a), & \text{for }a \in \alphabetOf{\rewFct_2} \setminus \alphabetOf{\rewFct_1}. 
    \end{cases}
\end{align*}
\begin{restatable}[Reward Sum]{theorem}{restatableRewardRule}\label{theo:pag_reward_rule}

   Let $\region_{\regionIntersection} = \regionIntersectionOf{\region_1}{}{\region_2}$, $\multiobjectiveQuery{}{A}_{\land} = \multiobjectiveQuery{}{A}_{1} \land \multiobjectiveQuery{}{A}_{2}$, and $\ppa = \alphabetExtensionOfTo{(\ppa_1 \parallel \ppa_2)}{\alphabetOf{\multiobjectiveQuery{}{A}_{\land}}}$. 

   $
    \infer{
            \agTriple{{ \ppa}, \region_{\regionIntersection}}{{\fairWrtRegionModel{}{ \decomp_1 \cup \decomp_2 }  }  }{\multiobjectiveQuery{}{A}_{\land} }{\expPredicate{\sim r_1 +r_2 }{\rewFct_1 + \rewFct_2}}
        }
        {\deduce{
                \agTriple{\alphabetExtensionOfTo{\ppa_2}{\multiobjectiveQuery{}{A}_2}, \region_2 }{\fairWrtRegionModel{}{  \decomp_2} }{\multiobjectiveQuery{}{A}_2}{\expPredicate{\sim r_2}{\rewFct_2}}
            }
            {
                \agTriple{\alphabetExtensionOfTo{\ppa_1}{\multiobjectiveQuery{}{A}_1}, \region_1 }{\fairWrtRegionModel{}{ \decomp_1} }{\multiobjectiveQuery{}{A}_1}{\expPredicate{\sim r_1}{\rewFct_1}}
            }
        }
    $

        for  $\sim \in \{ <, \leq, >, \geq\}$ and 
       $\decomp_i \in \setOfdecompsOf{\alphabetOf{\ppa_i}\cup \alphabetOf{\multiobjectiveQuery{}{A}_{i}}}$. 
\end{restatable}
%\restatableRewardRule*
\begin{proof}
        The proof is similar to the proof of \cite[Theorem 7]{Kwi+13}. 
        Since $\alphabetOf{\expPredicate{\sim r_1 + r_2}{\rewFct_1 + \rewFct_2}} \subseteq \alphabetOf{\ppa} $,  
        We have that ${\expPredicate{\sim r_1 + r_2}{\rewFct_1 + \rewFct_2}}$ is a valid query for $\ppa$, 
        If $\regionIntersectionOf{\region_{1}}{}{\region_2} = \emptyset$, the conclusion trivially holds. 
        We assume $\regionIntersectionOf{\region_{1}}{}{\region_2} \not = \emptyset$. 

        We show that $\agTriple{{\alphabetExtensionOfTo{\ppa_1}{\alphabetOf{\multiobjectiveQuery{}{A}_1}}\parallel \alphabetExtensionOfTo{\ppa_2}{\alphabetOf{\multiobjectiveQuery{}{A}_2}}}, \region_{\regionIntersection}}{{\fairWrtRegionModel{}{ \decomp_1 \cup \decomp_2 }  }  }{\multiobjectiveQuery{}{A}_{\land} }{{\expPredicate{\sim r_1 + r_2}{\rewFct_1 + \rewFct_2}}}$
        holds, which implies that 
        $\agTriple{{ \ppa}, \region_{\regionIntersection}}{{\fairWrtRegionModel{}{ \decomp_1 \cup \decomp_2 }  }  }{\multiobjectiveQuery{}{A}_{\land} }{{\expPredicate{\sim r_1 + r_2}{\rewFct_1 + \rewFct_2}}}$ holds. 
    
        Let $\valuation$ be a valuation in $\regionIntersectionOf{\region_1}{}{\region_2}$ and $\strategy$  a $\fairWrtRegionModel{}{\decomp_1 \cup  \decomp_2 }$ strategy of 
        $\alphabetExtensionOfTo{\ppa_1}{\alphabetOf{\multiobjectiveQuery{}{A}_1}}\parallel \alphabetExtensionOfTo{\ppa_2}{\alphabetOf{\multiobjectiveQuery{}{A}_2}}[\valuation]$. 
        Assume 
        $ {(\alphabetExtensionOfTo{\ppa_1}{\alphabetOf{\multiobjectiveQuery{}{A}_1}}\parallel \alphabetExtensionOfTo{\ppa_2}{\alphabetOf{\multiobjectiveQuery{}{A}_2}})}[\valuation],{\strategy} \models  (\multiobjectiveQuery{}{A}_1 \land \multiobjectiveQuery{}{A}_2)[\valuation ]$. 
        We show ${{\alphabetExtensionOfTo{\ppa_1}{\alphabetOf{\multiobjectiveQuery{}{A}_1}}\parallel \alphabetExtensionOfTo{\ppa_2}{\alphabetOf{\multiobjectiveQuery{}{A}_2}}}}[\valuation],{\strategy} \models {\expPredicate{\sim r_1 +r_2 }{\rewFct_1 + \rewFct_2}}[\valuation]$. 
        The premises imply that $\agTriple{\alphabetExtensionOfTo{\ppa_i}{\multiobjectiveQuery{}{A}_i}, \regionIntersectionOf{\region_1}{}{\region_2} }{\fairWrtRegionModel{}{\decomp_i} }{\multiobjectiveQuery{}{A}_i}{\expPredicate{\sim r_i}{\rewFct_i}}$ holds 
        for $i \in \{1,2\}$. 
        Thus, 
        \begin{align}
            & {(\alphabetExtensionOfTo{\ppa_1}{\alphabetOf{\multiobjectiveQuery{}{A}_1}}\parallel \alphabetExtensionOfTo{\ppa_2}{\alphabetOf{\multiobjectiveQuery{}{A}_2}})}[\valuation],{\strategy} \models   \multiobjectiveQuery{}{A}_i[ \valuation ] \nonumber \\
            & \Rightarrow \alphabetExtensionOfTo{\ppa_i }{\alphabetOf{\multiobjectiveQuery{}{A}_i}}[\valuation],{\stratProjOfToValuation{\strategy}{i}{\valuation} } \models   \multiobjectiveQuery{}{A}_i[ \valuation ]
            &&  \text{By \Cref{lemma_3_fDependent}} \nonumber \\
            & \Rightarrow 
            \alphabetExtensionOfTo{\ppa_i }{\alphabetOf{\multiobjectiveQuery{}{A}_i}}[{\valuation}],{\stratProjOfToValuation{\strategy}{i}{\valuation} } \models  {\expPredicate{\sim r_i  }{\rewFct_i }}[\valuation]
            && \text{$ \stratProjOfToValuation{\strategy}{i}{\valuation} \in \strategysetOf{{\alphabetExtensionOfTo{\ppa_i}{\alphabetOf{\multiobjectiveQuery{}{A}_i} }}[\valuation]}{\fairWrtRegionModel{}{\decomp_i}}$} \nonumber\\
            & && \text{as $\agTriple{\alphabetExtensionOfTo{\ppa_i}{\alphabetOf{\multiobjectiveQuery{}{A}_i}}, \regionIntersectionOf{\region_1}{}{\region_2} }{\fairWrtRegionModel{}{\decomp_i}}{\multiobjectiveQuery{}{A}_i}{\expPredicate{\sim r_i  }{\rewFct_i }}$} \nonumber 
        \end{align}
        It follows that 
        \begin{align}	\ExpTot{\alphabetExtensionOfTo{\ppa_i}{\alphabetOf{\multiobjectiveQuery{}{A}_i}}}{{\valuation},\stratProjOfToValuation{\strategy}{i}{\valuation}}{\rewFct_i}  \sim r_i && \label{eq:proof_rew_quant_G_sat} 
        \end{align}
        Thus, 
        \begin{align*}
            & \ExpTot{{(\alphabetExtensionOfTo{\ppa_1}{\alphabetOf{\multiobjectiveQuery{}{A}_1}}\parallel \alphabetExtensionOfTo{\ppa_2}{\alphabetOf{\multiobjectiveQuery{}{A}_2}})}}{\valuation,\strategy}{\rewFct_1+\rewFct_2}  \\
            & = \int_{\pi } (\rewFct_1 + \rewFct_2)[\valuation](\pi)  \,d 	\PrOf{{(\alphabetExtensionOfTo{\ppa_1}{\alphabetOf{\multiobjectiveQuery{}{A}_1}}\parallel \alphabetExtensionOfTo{\ppa_2}{\alphabetOf{\multiobjectiveQuery{}{A}_2}})}}{\valuation,\strategy}{}
            && 
            \text{By definition}
            \\
            &  = \sum_{i \in \{1,2\}} \int_{\pi } (\rewFct_i[\valuation])(\pi)  \,d 	\PrOf{{(\alphabetExtensionOfTo{\ppa_1}{\alphabetOf{\multiobjectiveQuery{}{A}_1}}\parallel \alphabetExtensionOfTo{\ppa_2}{\alphabetOf{\multiobjectiveQuery{}{A}_2}})}}{\valuation,\strategy}{} 
            &&   \\
            & = \sum_{i \in \{1,2\}} \ExpTot{{(\alphabetExtensionOfTo{\ppa_1}{\alphabetOf{\multiobjectiveQuery{}{A}_1}}\parallel \alphabetExtensionOfTo{\ppa_2}{\alphabetOf{\multiobjectiveQuery{}{A}_2}})}}{\valuation,\strategy}{\rewFct_i}  && 
            \text{By definition} \\
                & = \sum_{i \in \{1,2\}} \ExpTot{\left( \alphabetExtensionOfTo{\ppa_1}{\alphabetOf{\multiobjectiveQuery{}{A}_1}} \parallel \alphabetExtensionOfTo{\ppa_2}{\alphabetOf{\multiobjectiveQuery{}{A}_2}}
                    \right)}{\valuation,\strategy}{\rewFct_i}  
                &&  \\
                &  = \sum_{i \in \{1,2\}} \underbrace{\ExpTot{ \alphabetExtensionOfTo{\ppa_i}{\alphabetOf{\multiobjectiveQuery{}{A}_i}} 
                }{\valuation,\stratProjOfToValuation{\strategy}{i}{\valuation}}{\rewFct_i} }_{\sim r_i}
                && \text{By \Cref{lemma_3_dDependent}}\\
                & \sim r_1 +r_2  &&\text{ by \Cref{eq:proof_rew_quant_G_sat}}
        \end{align*}
        Thus, we have shown  
        \[
            {(\alphabetExtensionOfTo{\ppa_1}{\alphabetOf{\multiobjectiveQuery{}{A}_1}}\parallel \alphabetExtensionOfTo{\ppa_2}{\alphabetOf{\multiobjectiveQuery{}{A}_2}})}[\valuation],{\strategy} \models {\expPredicate{\sim r_1 +r_2 }{\rewFct_1 + \rewFct_2}}[\valuation]
        \]
    \end{proof}

\section{Differences to \cite{Kwi+13} --- Explained}
\label{sec:explain_differences}
The safety properties in \cite{Kwi+13} are languages over infinite words, whereas we allow finite words. 
This leads to the following problems in \cite{Kwi+13}: 
\begin{itemize}
	\item  \cite[Lemma 1]{Kwi+13} is inconsistent because $\regLang$ is not prefix-closed: 
	In \cite{Kwi+13}, safety properties are defined over infinite words: 
	$\regLang = \{ w \in  \alphabetOf{}^{{\color{myred}{\omega}}} \mid \text{no prefix of $w$ is in } \regLang_{err} \}$. 
	However, the probability of a language is defined as: $\PrOf{\pa}{\strategy}{{  \regLang}} = \PrOf{\pa}{\strategy}{\{ \pi \in \infPathsOf{\pa}{} \mid {\color{myred} \restrOfTo{\traceOf{\pi}}{\alphabetOf{\regLang}} \in \regLang} \}}$.  
	This definition does not properly account for bad prefixes. 
	Specifically, for a path $\pi$ of a PA $\pa$ and a safety $\regLang$ for $\pa$, the fact that $\restrOfTo{\traceOf{\pi}}{\alphabetOf{\regLang}} \not \in \regLang$ does not imply that $\restrOfTo{\traceOf{\pi}}{\alphabetOf{\regLang}}$ is accepted by the bad-prefix automaton of $\regLang$.
	For example, if \(\alphabetOf{\regLang} \subset \alphabetOf{\pa}\), then \(\restrOfTo{\traceOf{\pi}}{\alphabetOf{\regLang}}\) may be a finite prefix of a word in \(\regLang\) without being a bad prefix. 
	Since \(\regLang\) is not prefix-closed, this leads to situations where $\pi \not \in \{ \pi \in \infPathsOf{\pa}{} \mid \restrOfTo{\traceOf{\pi}}{\alphabetOf{\regLang}} \in \regLang \},$ even though \(\restrOfTo{\traceOf{\pi}}{\alphabetOf{\regLang}}\) is not a bad prefix. 
	\item This inconsistency propagates to \cite[Proposition 1]{Kwi+13} and thus also to the asymmetric AG-rule in \cite[Theorem 1]{Kwi+13}, both of which rely on \cite[Lemma 1]{Kwi+13}.
\end{itemize}
By allowing finite words, our definition ensures that safety properties are inherently prefix-closed, resolving these issues: 
$\regLang = \{ w \in  \alphabetOf{}^{{\color{myred}{\infty}}} \mid \text{no prefix of $w$ is in } \regLang_{err} \}$. 

	\begin{figure}[ht] 
		\centering 
		\begin{subfigure}[t]{.2\textwidth} 
			\centering 
			\begin{tikzpicture}[mdp]
	\node[ps, init=left] (s0)  {$s_0$};
	
	\node[ps, below=1.2 of s0] (s1)  {$s_1$};

	% Probabilistic transitions
	\path[ptrans]
	
	(s0) edge[] node[pos=0.5,right] {\tact{\lab}} node[pos=0.5,dist] (d0a) {} node[pos=0.1,right] {\tprob{}} (s1)
	
	(s1) edge[loop right] node[pos=0.5,below] {\tact{\altaltlab}} node[dist] (d1b) {} node[pos=0.1,right] {\tprob{}} node[pos=0.2,below] {} (s1)
	
	;

\end{tikzpicture}
			\caption{PA $\pa_1$}\label{fig:simple_paM1}
		\end{subfigure}% <- % ist nötig
		\begin{subfigure}{.44\textwidth} 
			\centering 
			\begin{subfigure}[t]{.45\textwidth} 
				\centering 
				\renewcommand{\extended}{false}
				\begin{tikzpicture}[mdp]
	\node[ps, init=left] (t0)  {$t_0$};

	% Probabilistic transitions
	\path[ptrans]
	
	(t0) edge[loop right] node[pos=0.5,below] {{\ifthenelse{\equal{\extended}{true}}{\tactext{\lab,\altaltlab},}{}}\tact{\altlab}} node[dist] (d1b) {} node[pos=0.1,below] {\tprob{}} node[pos=0.1,below] {} (t0)
	;
	
\end{tikzpicture}
				\caption{PA $\pa_2$}\label{fig:simple_paM2}
			\end{subfigure}\vspace{0.5cm}
			
			\begin{subfigure}[t]{.45\textwidth} 
				\centering 
				\renewcommand{\extended}{true}
				\begin{tikzpicture}[mdp]
	\node[ps, init=left] (t0)  {$t_0$};

	% Probabilistic transitions
	\path[ptrans]
	
	(t0) edge[loop right] node[pos=0.5,below] {{\ifthenelse{\equal{\extended}{true}}{\tactext{\lab,\altaltlab},}{}}\tact{\altlab}} node[dist] (d1b) {} node[pos=0.1,below] {\tprob{}} node[pos=0.1,below] {} (t0)
	;
	
\end{tikzpicture}
				\caption{PA $\alphabetExtensionOfTo{\pa_2}{{\{\lab,\altaltlab\}}}$}\label{fig:simple_paM2_extended}
			\end{subfigure} 
		\end{subfigure} 
		\begin{subfigure}[t]{.3\textwidth} 
			\centering
			\renewcommand{\extended}{false}
			\begin{tikzpicture}[mdp]
	\node[pswide, init=left] (00)  {$s_0, t_0$};	
	
	\node[pswide, below=1.1 of 00] (10)  {$s_1, t_0$};
	
	% Probabilistic transitions
	\path[ptrans]
	
	(00) edge[loop right] node[pos=0.5,below right] {\tact{\altlab}} node[dist] (d1b) {} node[pos=0.15,below] {\tprob{}} node[pos=0.25,below] {} (00)
	
	(10) edge[loop right] node[pos=0.5,below right] {\tact{\altlab},\tact{\altaltlab}} node[dist] (d1b) {} node[pos=0.15,below] {\tprob{}} node[pos=0.25,below] {} (10)
	
	(00) edge[] node[pos=0.5,right] {\tact{\lab}} node[pos=0.5,dist] (d0a) {} node[pos=0.25,below] {\tprob{}} (10)
	
	;
\end{tikzpicture}
			\caption{PA $\pa_1 \parallel \pa_2$}\label{fig:simple_pas_composed}
		\end{subfigure} 
		\caption{PAs used in \Cref{ex:error_AGsafety_rule_1}} 
	\end{figure}

We present a counterexample to \cite[Lemma 1]{Kwi+13} and \cite[Theorem 1]{Kwi+13}, for safety properties containing only infinite words and $\multiobjectiveQuery{\safe}{G}  \not = \alphabetOf{\pa_1 \parallel \pa_2}$ (a similar counterexample exists if $\multiobjectiveQuery{\safe}{A}  \not = \alphabetOf{\pa_1}$). 
\begin{example}\label{ex:error_AGsafety_rule_1}

		Consider the safety property $\regLang = \{ w \in \{\lab,\altaltlab\}^{\omega} \mid \vert w \vert_{\lab} \leq 1 \}$, i.e., $\alphabetOf{\regLang} = {\lab,\altaltlab}$.  
		Importantly, $\regLang$ consists only of infinite words.
		Now, consider the PAs $\pa_1$ and $\pa_2$ in \Cref{fig:simple_paM1,fig:simple_paM2} 
		and their composition $\pa_1 \parallel \pa_2$ in \Cref{fig:simple_pas_composed}. 
		We define the complete strategy $\strategy$ that always chooses action $b$ in $(s_0, t_0)$ with probability 1. 
		Then, $\PrOf{\pa_1 \parallel \pa_2}{\strategy}{\regLang} = 0$ since $\restrOfTo{b^{\omega}}{\alphabetOf{\regLang}} = \epsilon \not \in \regLang$. 
		Thus, we obtain $\pa_1 \parallel \pa_2  \not \modelsWrt{\comp} \multiobjectiveQuery{\safe}{G}$. 
		However, by \cite[Lemma 1]{Kwi+13}, it would follow that $\pa_1 \parallel \pa_2  \modelsWrt{\comp} \multiobjectiveQuery{\safe}{G}$ as there is no strategy $\strategy'$  such that 
        \[
             1- \PrOf{(\pa_1 \parallel \pa_2) {\otimes} \bpAutomatonOf{\regLang}}{\strategy'}{\Diamond bad_{\regLang}} = 1= \PrOf{\pa_1 \parallel \pa_2}{\strategy}{\regLang}. 
        \]
		Since $\PrOf{\pa_1 \parallel \pa_2}{\strategy}{\regLang} = 0$, there cannot be a bijection, contradicting \cite[Lemma 1]{Kwi+13}. 

		This also provides a counterexample to \cite[Theorem 1]{Kwi+13}.
		Let $\multiobjectiveQuery{\safe}{A} = \multiobjectiveQuery{\safe}{G}  = \{ \probPredicate{\geq 1}{\regLang} \}$.    
		For $\pa_1$ in \Cref{fig:simple_paM1}, we have that $\pa_1 \modelsWrt{\comp} \multiobjectiveQuery{\safe}{A}$. 
		Additionally, for $\alphabetExtensionOfTo{\pa_2}{\alphabetOf{\multiobjectiveQuery{\safe}{A}}}$ in \Cref{fig:simple_paM2_extended}, we obtain $\agTriple{\alphabetExtensionOfTo{\pa_2}{\alphabetOf{\regLang}}}{\prt}{\multiobjectiveQuery{\safe}{A}}{\multiobjectiveQuery{\safe}{G}}$, which trivially holds, as it is a tautology. 
		By \cite[Theorem 1]{Kwi+13} this implies that $\pa_1 \parallel \pa_2  \modelsWrt{\comp} \multiobjectiveQuery{\safe}{G}$.  
\end{example}
\cite[Lemma 1, Theorem 1]{Kwi+13} yields results we would expect using our definition of safety properties, which consider both finite and infinite words.

}{
    % for Conference verson: no appendix
}

\end{document}